\documentclass[12pt,a4paper]{article}

\usepackage{mathtools}
\usepackage{amsfonts}
\usepackage{amssymb}
\usepackage{amsthm}
\usepackage{amsmath}
\usepackage{bbm}
\usepackage{extpfeil}
\usepackage{extarrows}
\usepackage{hyperref}
\usepackage{graphicx}

\newcommand\abs[1]{\left|#1\right|}
\renewcommand{\d}[1]{\ensuremath{\operatorname{d}\!{#1}}}
\DeclareMathOperator*{\esssup}{ess\,sup}
\DeclareMathOperator*{\essinf}{ess\,inf}
\DeclareMathOperator*{\argmin}{arg\,min}

\newtheorem{theorem}{Theorem}[subsection]
\newtheorem{proposition}[theorem]{Proposition}
\newtheorem{corollary}[theorem]{Corollary}
\newtheorem{lemma}[theorem]{Lemma}
\theoremstyle{definition}
\newtheorem{definitionx}[theorem]{Definition}
\newtheorem{remarkx}[theorem]{Remark}

\newenvironment{definition}
  {\pushQED{\qed}\definitionx}
  {\popQED\enddefinitionx}
\newenvironment{remark}
  {\pushQED{\qed}\remarkx}
  {\popQED\endremarkx}

\title{\textbf{Forward indifference valuation and hedging of basis risk under\\partial information}
\vspace{3em}
\begin{figure}[ht]
\centering
\includegraphics[scale=0.45]{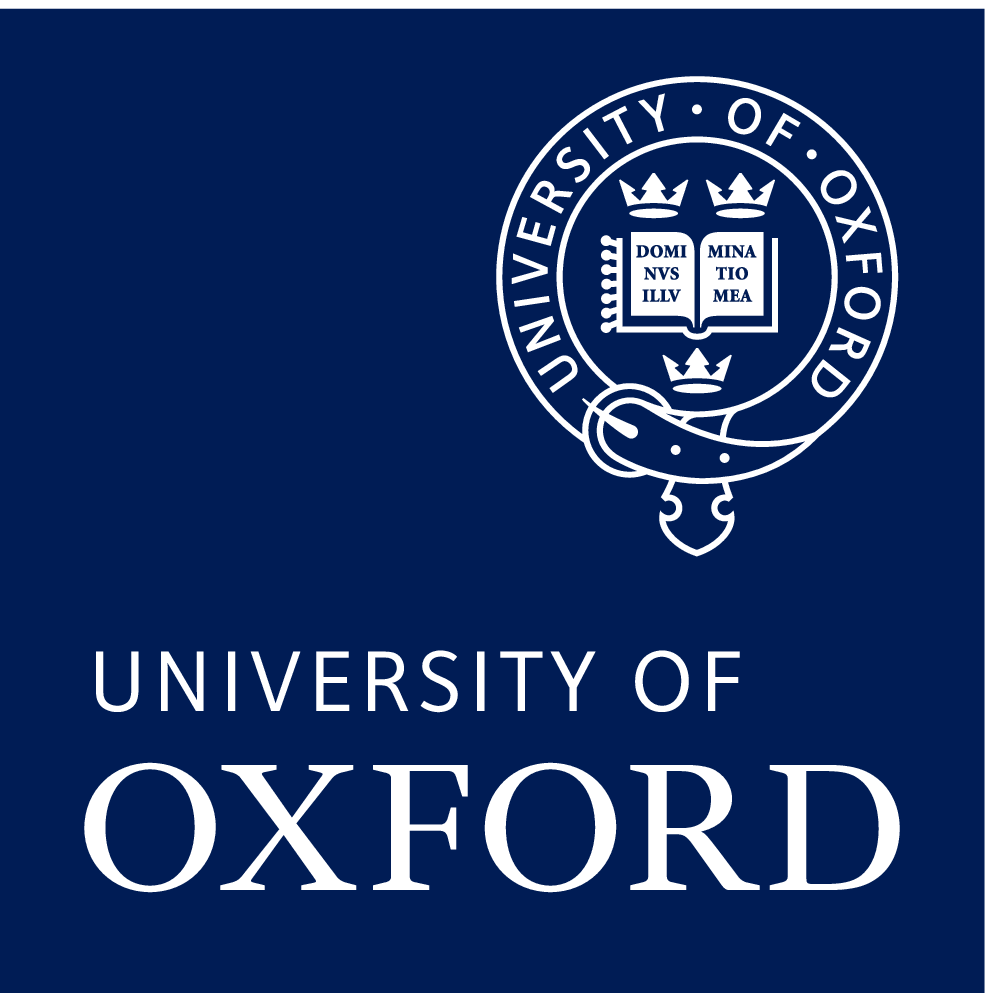}
\end{figure}
\vspace{3em}
}
\author{Mahan Tahvildari\\
Lincoln College\\
Mathematical Institute\\
University of Oxford
\vspace{3em}\\
A thesis submitted in partial fulfilment of the\\
\textit{Master of Science in Mathematical Finance}}
\date{December 16, 2018}
\begin{document}

\maketitle
\thispagestyle{empty}

\newpage
\thispagestyle{empty}
\begin{center}
\textit{I dedicate this work to my parents.}
\end{center}
\vspace{3em}
\begin{center}
``Be happy for this moment. This moment is your life.''\\
\textit{Omar Khayy\'{a}m, Persian mathematician, astronomer, and poet}
\end{center}
\newpage
\thispagestyle{empty}
\section*{Acknowledgments}
I would like to thank my supervisor Professor Dr Michael Monoyios for his guidance, constructive suggestions and encouragement for this research work. I would also like to express my great appreciation to my employer d-fine, that made it possible for me to study at Oxford University by granting me a full scholarship. Finally, I wish to thank my parents for their restless support and encouragement throughout my studies.
\newpage
\thispagestyle{empty}
\begin{abstract}
We study the hedging and valuation of European and American claims on a non-traded asset $Y$, when a traded stock $S$ is available for hedging, with $S$ and $Y$ following correlated geometric Brownian motions. This is an \textit{incomplete market}, often called a \textit{basis risk model}. The market agent's risk preferences are modelled using a so-called \textit{forward performance process (forward utility)}, which is a time-decreasing utility of exponential type. Moreover, the market agent (investor) does not know with certainty the values of the asset price drifts. This market setting with drift parameter uncertainty is the \textit{partial information scenario}. We discuss the stochastic control problem obtained by setting up the hedging portfolio and derive the \textit{optimal hedging strategy}. Furthermore, a \textit{(dual) forward indifference price} representation of the claim and its PDE are obtained. With these results, the \textit{residual risk process} representing the basis risk (hedging error), \textit{pay-off decompositions} and \textit{asymptotic expansions} of the indifference price in the European case are derived. We develop the analogous \textit{stochastic control and stopping problem} with an American claim and obtain the corresponding forward indifference price valuation formula.
\end{abstract}
\newpage
\thispagestyle{empty}
\tableofcontents
\newpage
\setcounter{page}{1}
\section{Introduction}
A fundamental theory of mathematical finance is the problem of a market agent who invests in a financial market in order to maximise trxivehe \textit{expected utility} of his terminal wealth under his individual preferences. Problems of expected utility maximisation go back at least to the two seminal articles of Merton~\cite{Me1969}, \cite{Me1971} (see also Merton~\cite{Me1992}), who studied the framework of time-continuous models, and the seminal article of Samuelson~\cite{Sa1969} treating the time-discrete case. Merton derived a non-linear partial differential equation \textit{(Hamilton-Jacobi-Bellman (HJB) equation)} for the value function of the maximisation problem using methods from stochastic control theory.\par
The modern approach for solving such problems uses \textit{dual characterisations} of portfolios through defining an appropriate set of martingale measures. Harrison and Pliska~\cite{HP1981} developed a general stochastic model of a continuous, multi-dimensional, \textit{complete} market and obtained the corresponding general Black-Scholes pricing formula. The setting of a complete market, where the martingale measure is unique, was also studied by Pliska~\cite{P1986}, Cox and Huang~\cite{CH1989}, \cite{CH1991} and by Karatzas, Lehoczky and Shreve~\cite{KLS1987}. One of the main results is, that the \textit{marginal utility} of the optimal terminal wealth is equal to the density of the martingale measure modulo a constant.\par
The setting of an \textit{incomplete} market, where perfect hedging is not possible, is a more difficult case and was studied via time-discrete models by He and Pearson~\cite{HP1991} and by Karatzas, Lehoczky, Shreve and Xu~\cite{KLSX1991}, who realised that the use of \textit{dual methods} from convex analysis provided comprehensive solutions to stochastic control problems. The \textit{dual variational problem} of the \textit{primal problem} is formulated and solved with \textit{convex dual} relationship as Bismut~\cite{Bi1973} demonstrated. Kramkov and Schachermayer~\cite{KS1999} studied the classical utility maximisation problem under weaker assumptions on the model and on the utility function. Rogers \cite{R2003} delved deeper into the theory by applying methods from functional analysis and presented various examples solved with duality methods (see also \v{Z}itkovi\'{c} \cite{Zi2009} and Berrier, Rogers and Tehranchi \cite{BRT2009}). Davis~\cite{D2000}, \cite{D2006}, Rouge and Karoui~\cite{RK2000}, Henderson and Hobson~\cite{H2002}, \cite{HH2002}, Musiela and Zariphopoulou \cite{MZ2004} investigated utility-based hedging in an incomplete market case, where hedging becomes imperfect and a \textit{hedging error}, the \textit{basis risk}, remains. It is the risk associated with the trading of a derivative security on a non-traded underlying asset, hedged with a imperfectly correlated traded asset. Examples are weather derivatives or options on illiquid securities. Ankirchner and Imkeller~\cite{AI2001} introduced a typical example for a cross hedge, where an airline company wants to manage kerosene price risk. Ankirchner et al.~\cite{ADHP2010}, \cite{AID2010}, \cite{AIP2008} dealt also with applied basis risk models. Monoyios~\cite{M2004} derived \textit{perturbation series} giving accurate analytic approximations for the price and hedging strategy of the claim using an exponential utility and carried out an numerical performance analysis between the improved optimal hedge and the naive hedge with the traded asset. Kallsen and Rheinl\"ander dealt with classical utility-based pricing and hedging using an quadratic hedging approach and extended the results obtained by Mania and Schweizer \cite{MS2005}, Becherer \cite{B2006} and Kramkov and S\^{i}rbu \cite{KS2007}. Zariphopoulou \cite{Z2001} studied optimisation models with power utility and produced reduced form solutions of the indifference price by applying a \textit{distortion method} to the indifference price PDE. The setting with exponential utility in a multi-dimensional model was treated by Musiela and Zariphopoulou~\cite{MZ2004}. Monoyios~\cite{M2006} derived representations for the optimal martingale measures in a two-factor Markovian model by using the distortion power solution for the primal problem to obtain a dual \textit{entropic} representation of the stochastic control problem. We refer to the introductions of the aforementioned papers for more references in the field of classical utility-based optimisation problems.\par
Monoyios~\cite{M2007} explored the impact of drift parameter uncertainty in an incomplete market model having an European option on a non-traded asset hedging a correlated traded stock. He developed analytic expansions for the indifference price and hedging strategies. The key approach is the development of a \textit{filtering approach}, the \textit{Kalman-Bucy filter}, in which the investor updates the \textit{market price of risk} parameter from the observations of the asset prices. Applications of filtering can be found in Kallianpur \cite{K1980}, Rogers and Williams \cite{RW2001} and Fujisaki et al. \cite{FKK1972}. Filtering originates from signal processing by Wiener \cite{W1964} and Kolmogorov \cite{I1989} during the 1940s. In the 1960s, it was further development by Kalman and Bucy~\cite{K1960},~\cite{KB1961}. The setting, in which the investor does not observe the assets' Brownian motions is called the \textit{partial information scenario}. Problems under partial information scenarios were also studied by Rogers~\cite{R2001}, Lakner~\cite{L1998} and Brendle~\cite{Br2006}. Monoyios~\cite{M2010} used a two-dimensional Kalman-Bucy filter with Gaussian prior distribution in a partial information model and derived the optimal hedging strategy and indifference price representations using dual methods. Dependent on the prior estimations of the asset price drifts, the price representations formulas uses the \textit{minimal entropy martingale measure} or the \textit{minimal martingale measure}.\par
Musiela and Zariphopoulou~\cite{MZ2007}, \cite{MZ2009}, \cite{MZ2010} introduced a new class of \textit{forward utilities (forward performances)} that are generated forward in time. They discussed associated value functions, optimal investment strategies and indifference price representations. They defined the concept of forward performance processes in order to quantify the dynamically changing preferences of an investor. Independently, Henderson defined in \cite{H2007} and Henderson and Hobson~\cite{HH2007} analysed the same class of dynamic utilities, but called them \textit{horizon-unbiased}. Forward utilities are defined by the \textit{dynamic programming principle} and ensure more flexibility as they are specified for today and not for a fixed future time. Berrier and Theranchi~\cite{BT2009} broadened the definitions by adding a process for the investor's consumption.\par
In this work, we investigate the utility-based valuation of European and American claims on a non-traded asset $Y$, when a correlated traded stock $S$ is available for hedging, with $S$ and $Y$ following correlated geometric Brownian motions, and when the agent's risk preferences are modelled using the forward performance process from Musiela and Zariphopoulou~\cite{MZ2009}, and when the agent does not know with certainty the values of the asset price drifts. Since the market becomes incomplete, we retain an unhedgeable (basis) risk. The basis risk model will first be constructed under a standard full information hypothesis, where the drifts of both assets are known constants. In this setting, the utility-based valuation of European claims on $Y$ has been well studied, using classical (as opposed to forward) exponential utility. The partial information case, where the asset drifts are taken as unknown constants, whose values are filtered from price observations, has also been studied for European claim valuation using classical utility (see, for example, Monoyios~\cite{M2010}). The thesis will investigate the valuation and hedging of European and then American claims on $Y$ with a exponential forward utility under partial information. We apply the partial information model with the Kalman-Bucy filter from Monoyios~\cite{M2010} to get analogous results for valuation and hedging with forward instead of classical utility. The novel approach is the embedding of the specific partial information model, making the market prices of risks depending on both asset prices, into the aforementioned forward performance framework. We compare the optimal hedging strategies and indifference price representations for European and American claims associated with forward and classical utility under the partial information scenario. One of the key results is the change of optimal measure from the minimal entropy martingale measure $\mathbb{Q}^E$ to the minimal martingale measure $\mathbb{Q}^M$. In the European option's case, we obtain the dual representation of the forward indifference price with its semi-linear PDE of second order, the residual risk, a pay-off decomposition of the European claim and an asymptotic expansion of the forward indifference price. In the case of an American claim, we define the control and stopping problem and derive the dual representation of the forward indifference price under the partial information model.\par
Oberman and Zariphopoulou~\cite{OZ2003} and Leung and Sircar~\cite{LS2009} studied the valuation and hedging of American options in a basis risk model under full information using classical utility. Leung, Sircar and Zariphopoulou~\cite{LSZ2012} investigated the full information model using forward utility to price \textit{executive stock options (ESOs)}. ESOs are American calls issued by a company to its employees (mostly executives) as a form of variable payments as instruments for motivation (cf.~Kraizberg et al.~\cite{KTW2002}, Chen et al.~\cite{CCC2014}, Brandes et al.~\cite{BDLH2003}). We extend the framework of Leung, Sircar and Zariphopoulou~\cite{LSZ2012} to the partial information model, derive the indifference price valuation and prepare the groundwork for future applications.\par
The remainder of the dissertation is organised as follows. In Section~\ref{sec:basis risk model} the basis risk market model in the full and partial information scenarios, the concept of filtering and forward utilities defined via a certain class of risk tolerance functions are treated. The forward utility-based valuation and hedging problem with an European option is dealt in Section~\ref{sec:european option}. It begins with the setting of perfect hedging in a complete market and continues with the incomplete market case, followed by the formulation of the performance maximisation of the investor's hedging portfolio. The problem is solved with dual methods and results in the optimal hedging strategy and the dual representation formula for the forward indifference price. Furthermore, the residual risk of the strategy, option's pay-off decompositions and an asymptotic expansion for the forward indifference price are derived. In Section~\ref{sec:american option} we set up the partial information model with an American option, which can be early exercised and develop the (dual) optimal control and stopping optimisation problem and obtain the entropic representation of the forward indifference price. We conclude in Section~\ref{sec:conclusions} by performing an analysis of essential model assumptions and results obtained in this work, and discuss alternatives from present topics as well as future directions for research.
\section{Basis risk model}
\label{sec:basis risk model}
In this section the financial market is modelled by a \textit{basis risk model} premised on the geometric Brownian motion. A distinction is made in the assumption of the asset price \textit{Sharpe ratios}, which leads to the \textit{full information scenario} for certain Sharpe ratios and \textit{partial information scenario} for uncertain Sharpe ratios. The \textit{Kalman-Bucy filtering approach} is developed and applied to the partial information scenario to transform it into the case of full information. Lastly, the concept of \textit{forward utility} is introduced as a dynamic extension of the classical utility theory and used within the basis risk model. Herein, a useful function called \textit{local risk tolerance} serves the classification of \textit{forward utility functions}.
\subsection{Full information scenario}
\label{subsec:full information scenario}
The classical basis risk model defined in this subsection was initially explored by Davis~\cite{D2006}.
Consider a filtered probability space $(\Omega, \mathcal{F}, \mathbb{F}:=(\mathcal{F}_t)_{0\leq t\leq T}, \mathbb{P})$ as the setting of a financial market, where the terminal filtration $\mathbb{F}$ is generated by the two-dimensional standard $\mathbb{P}$-Brownian motion $(W^S,W^\perp)$ with correlation between the Wiener processes $W^S:=(W_t^S)_{0\leq t\leq T}$ and $W^\perp:=(W_t^\perp)_{0\leq t\leq T}$. A traded stock price $S:=(S_t)_{0\leq t\leq T}$ follows a geometric Brownian motion process given by
\begin{equation}
\label{eq:dS}
  \d S_t = \sigma^S S_t (\lambda^S\d t + \d W_t^S),
\end{equation}
in which the stock's volatility $\sigma^S > 0$ and its \textit{market price of risk (MPR)} or \textit{Sharpe ratio} $\lambda^S = \frac{\mu^S - r_m}{\sigma^S}$ with drift $\mu^S$ are known constants. For simplicity, the risk-free market interest rate $r_m$ is taken to be zero. A non-traded asset $Y:=(Y_t)_{0\leq t\leq T}$ follows the correlated geometric Brownian motion
\begin{equation}
\label{eq:dY}
  \d Y_t = \sigma^Y Y_t (\lambda^Y\d t + \d W_t^Y),
\end{equation}
with $\sigma^Y>0$ and $\lambda^Y$ known constants. The Brownian motion $W^Y:=(W_t)_{0\leq t\leq T}$ from the non-traded asset dynamics is correlated with the stock's Brownian motion $W^S$ according to $W_t^Y = \rho W_t^S + \sqrt{1-\rho^2}W_t^\perp$ with a known constant $\rho\in[-1,1]$ as the correlation coefficient. In the case $\abs{\rho}=1$, the market is called \textit{complete} and perfect hedging is possible; see Subsection~\ref{subsec:Perfect hedging in a complete market}. If $\abs{\rho}\neq 1$, the market is called \textit{incomplete}.\par
An investor with initial wealth $x>0$ dynamically rebalances his portfolio allocations between the stock and the riskless money market account according to his $\mathbb{F}$-predictable \textit{(portfolio or trading) strategy} $\theta:=(\theta_t)_{0\leq t\leq T}$ ($\pi:=(\pi_t)_{0\leq t\leq T}$), that is an $S$-integrable process representing the number of shares held in the portfolio (respectively the cash amount $\pi_t:=\theta_tS_t$ invested in the stock). Contextually, both $\theta$ and $\pi$ are called strategy. Under self-financing trading condition, the investor's \textit{portfolio wealth} is denoted by the positive process $X:=(X_t)_{0\leq t \leq T}$ and satisfies
\begin{equation}
\label{eq:dX}
\d X_t = \theta_t \d S_t = \sigma^S\pi_t (\lambda^S\d t + \d W_t^S),\quad X_0=x_0.
\end{equation}
The process $(\theta\cdot S)=((\theta\cdot S)_t)_{0\leq t\leq T}$ given by the stochastic integral
\begin{equation*}
(\theta\cdot S)_t:=\int_0^t \theta_u\d S_u = \int_0^t\d X_u = X_t-x_0,
\end{equation*}
represents the profit and loss from trading up to time $t\in[0,T]$. The next definition gives the space of \textit{admissible trading strategies} to make the market model suitable for measure changes.
\begin{definition}[Relative entropy and admissible strategies]
\label{def:relative entropy and admissible strategies}
The set of \textit{equivalent local martingale measures} $\mathcal{M}_e:=\{\mathbb{Q}\sim\mathbb{P}\,\vert\, S \text{ is a local }(\mathbb{Q},\mathbb{F})\text{-martingale}\}$ and its subset $\mathcal{M}_{e,f}:=\{\mathbb{Q}\in\mathcal{M}_e\,\vert\, \mathcal{H}(\mathbb{Q},\mathbb{P})<\infty\}$ of measures with finite \textit{relative entropy}
\begin{equation}
\label{eq:relative entropy}
\mathcal{H}(\mathbb{Q},\mathbb{P}):=\mathbb{E}\left[\frac{\d{\mathbb{Q}}}{\d{\mathbb{P}}}\log\frac{\d{\mathbb{Q}}}{\d{\mathbb{P}}}\right]
\end{equation}
between $\mathbb{Q}$ and $\mathbb{P}$ are assumed to be non-empty. The set of \textit{admissible strategies} is
\begin{equation}
\label{eq:admissible strategies}
\Theta:=\{\theta\in\Theta^{p}\,\vert\, (\theta\cdot S)\text{ is a }(\mathbb{Q},\mathbb{F})\text{-martingale for all }\mathbb{Q}\in\mathcal{M}_{e,f}\},
\end{equation}
where $\Theta^p$ is the superset of $(\mathbb{P},\mathbb{F})$-predictable and $S$-integrable strategies. An \textit{admissible strategy} satisfies $\int_0^t \pi_u^2 \d u<\infty$ almost surely for each $t\in[0,T]$.
\end{definition}
Condition (\ref{eq:admissible strategies}) for admissible strategies is taken from Becherer \cite[pp.~28--29]{B2004} (see also Mania and Schweizer \cite[p.~2116]{MS2005}) and appears as one of the candidate sets ($\Theta_2$) examined in Delbaen et al. \cite[p.~104]{DGRSSS2002}. The latter paper prove that for three different choices of $\Theta$ the resulting primal and dual problem have the same value and thus establish in particular a robustness result for the duality of classical exponential utility-based hedging.\par
The relative entropy was introduced in information theory by Kullback and Leibler \cite{KL1951} and developed by Kullback in his book \cite{K1959}. It is $\mathcal{H}(\mathbb{Q},\mathbb{P})\geq 0$ with equality if and only if $\mathbb{Q}=\mathbb{P}$. The profit and loss process $(\theta\cdot S)$ is identical to $(X_t-x)_{0\leq t\leq T}$ and hence the martingale property in (\ref{eq:admissible strategies}) holds also for the wealth process $X$. With the choice of admissible strategies, arbitrage opportunities for the investor are excluded. More on arbitrage and self-financing strategies can be found in Jeanblanc et al. \cite[pp.~81--84]{JYC2009}. Since the MPRs $\lambda^S, \lambda^Y$ are assumed to be constants, the investor has access to the so-called \textit{background filtration} $\mathbb{F}$ and hence is able to observe the Brownian motion process $(W^S,W^\perp)$, as well as the stock price process $S$. This set-up is referred to as a \textit{full information scenario}.
\begin{remark}[Solution to the stock price SDE]
\label{rem:solution to the stock price SDE}
To solve (\ref{eq:dS}), firstly apply It\^{o}'s lemma on the logarithmic stock prices,
\begin{equation*}
\d{\left(\log S_t\right)}=\frac{\d S_t}{S_t}-\frac12\frac{\d\langle S\rangle_t}{S_t^2}=\sigma^S\left(\left(\lambda^S-\frac{\sigma^S}{2}\right)\d t+\d W_t^S\right),
\end{equation*}
and the integrate over $[0,t]$ to obtain
\begin{equation*}
S_t=S_0\exp\left(\sigma^S\left(\lambda^S-\frac{\sigma^S}{2}\right)t + W_t^S\right).
\end{equation*}
If the MPR or volatility were not constant, then an integral would remain in the expressed solution.
\end{remark}
\subsection{Partial information scenario}
\label{subsec:partial information scenario}
Based on historical data analyses from Ang and Bekaert \cite{AB2007}, Clarke et al. \cite{CdT2006} and French et al. \cite{FSS1987}, one might assume that the parameter values for an annual stock return (drift) and volatility are $\mu^S=8\%$ and $\sigma^S=16\%$ respectively, so that the Sharpe ratio is $\lambda^S=0.5$ per annum. Nevertheless, as outlined in Rogers \cite[pp.~144--145]{R2001}, a proper estimation of $\lambda^S,\lambda^Y$ in the asset price dynamics (\ref{eq:dS}), (\ref{eq:dY}) is practically impossible, due to the lack of long-term historical market data. The subsequent argument for the \textit{MPR parameter uncertainty} is taken from Monoyios \cite[pp.~342--343]{M2007}. The normalised stock returns $\frac{1}{\sigma^S}\frac{\d S_t}{S_t}=\lambda^S\d t+\d W_t^S$ can be observed by the investor over a time interval $[0,t]$, to make the best estimate
\begin{equation*}
\overline{\lambda^S}(t)=\frac{1}{t}\int_0^t\frac{1}{\sigma^S}\frac{\d S_u}{S_u}=\lambda^S + \frac{W_t^S}{t}\sim\mathcal{N}\left(\lambda^S,\frac{1}{t}\right),
\end{equation*}
which leads to a 95\% confidence interval $\left[\overline{\lambda^S}(t)-\frac{1.96}{\sqrt{t}},\overline{\lambda^S}(t)+\frac{1.96}{\sqrt{t}}\right]$ for $\lambda^S$. In order to determine with 95\% confidence the observation time $t$ for the estimated value $\overline{\lambda^S}(t)$ being 5\% to within of its true value $\lambda^S$, meaning $\abs{\overline{\lambda^S}(t)-\lambda^S}\leq 0.05$, the equality $\frac{1.96}{\sqrt{t}}=0.05$ needs to be solved, which gives $t\approx 1537$ years. This calculation shows, how intrinsic the MPR parameter uncertainty in log-normal models is, since reliable historical price data for such a long period is not available, considering that two of the first formal exchanges worldwide, the Frankfurt Stock Exchange and the London Stock Exchange, were established in the late 16th and 17th centuries respectively, according to Holtfrerich \cite[p.~77]{H1999} and Michie \cite[p.~15]{Mi1999}.\par
The asset volatilities $\sigma^S, \sigma^Y$ and the correlation $\rho$ are assumed to be known constants, because they can be inferred from quadratic and co-variations
\begin{equation*}
\d\langle S\rangle_t =(\sigma^S)^2S_t^2 \d t, \quad\d\langle Y\rangle_t =(\sigma^Y)^2Y_t^2 \d t,\quad \d\langle S, Y\rangle_t =\rho\sigma^S\sigma^Y S_tY_t\d t,
\end{equation*}
through the best estimators
\begin{equation*}
\sigma^S=\sqrt{\frac{1}{t}\int_0^t\frac{\d\langle S\rangle_u}{S_u^2}},\quad
\sigma^Y=\sqrt{\frac{1}{t}\int_0^t\frac{\d\langle Y\rangle_u}{Y_u^2}},\quad
\rho=\frac{1}{\sigma^S\sigma^Yt}\int_0^t\frac{\d\langle S, Y\rangle_u}{S_uY_u}.
\end{equation*}
when price observations are taken to be approximately continuous. The problem of estimating quadratic variation using realised variance is discussed in Barndorff-Nielsen and Shephard \cite{BNS2002}. Chakraborti et al. \cite{CKKKO2003} analysed asset correlations on an empirical basis. If the requirement of constant MPRs is omitted, the agent will have no access to the background filtration $\mathbb{F}$, but instead, only to the so-called \textit{observation filtration} $\widehat{\mathbb{F}}:=(\widehat{\mathcal{F}}_t)_{0\leq t\leq T}$, which is generated by the asset price processes $S$ and $Y$. Hence, only the observation of $(S,Y)$ but not the Brownian motion process $(W^S,W^\perp)$ is possible. The values of the parameters $\lambda^S,\lambda^Y$ become uncertain, so they can be modelled as random variables. This set-up is referred to as a \textit{partial information scenario}. An agent with full information (partial information) is called \textit{outsider (insider)} (see Henderson, Klad\'{i}vko and Monoyios~\cite{HKM2018}).
\subsection{Kalman-Bucy filtering}
\label{subsec:Kalman-Bucy filtering}
General filtering theory deals with the estimation of an unobservable stochastic process given a related observable process. Treatments of filtering theory can be found in Kallianpur \cite[Chapter 10]{K1980}, Rogers and Williams \cite[pp.~322--331]{RW2001} and Fujisaki et al. \cite{FKK1972}. Wiener \cite{W1964} and Kolmogorov \cite{I1989} paved the way for filtering problems in the frequency domain in signal processing theory during the 1940s. In the 1960s linear filtering theory was developed further by Kalman~\cite{K1960} and Kalman and Bucy~\cite{KB1961}, where filtering problems were considered in the time rather than frequency domain with state space representations.\par
The partial information scenario can be converted into a full information scenario by changing from the background to the observation filtration using the so-called \textit{Bayesian approach} in a \textit{Kalman-Bucy filtering framework}. Following Monoyios \cite{M2010}, the asset MPRs are modelled as $\mathcal{F}_0$-random variables with a given initial distribution conditional on $\widehat{\mathcal{F}}_0$.
\begin{definition}[Observation and signal process]
\label{def:observation and signal process}
Define the two-dimensional\linebreak \textit{observation process} $\Xi:=(\Xi_t)_{0\leq t\leq T}$ by
\begin{equation}
\label{eq:observation}
\Xi_t:=\begin{pmatrix}
          \xi_t^S\\
          \xi_t^Y
       \end{pmatrix}:=\begin{pmatrix}
                         \frac{1}{\sigma^S}\int_0^t \frac{\d S_u}{S_u}\\
                         \frac{1}{\sigma^Y}\int_0^t \frac{\d Y_u}{Y_u}
                      \end{pmatrix}=\begin{pmatrix}
                                       \lambda^S t + W_t^S\\
                                       \lambda^Y t + W_t^Y
                                    \end{pmatrix},
\end{equation}
given the dynamics (\ref{eq:dS}) and (\ref{eq:dY}), generating the \textit{observation filtration} $\widehat{\mathbb{F}}:=(\widehat{\mathcal{F}}_t)_{0\leq t\leq T}$, $\widehat{\mathcal{F}}_t:=(\xi_u^S, \xi_u^Y\;\vert\; 0\leq u \leq t)$. The corresponding unobservable \textit{signal process} is given by
\begin{equation}
\label{eq:signal}
\Lambda:=\begin{pmatrix}
          \lambda^S\\
          \lambda^Y
       \end{pmatrix},
\end{equation}
which is an unknown two-dimensional constant in this market model. Moreover, assume a Gaussian \textit{prior distribution}
\begin{equation}
\label{eq:prior}
\Lambda\,\vert\,\widehat{\mathcal{F}}_0 \sim\mathcal{N}(\Lambda_0,\Sigma_0),\; \Lambda_0:=\begin{pmatrix}
              \lambda_0^S\\
              \lambda_0^Y
           \end{pmatrix},\; \Sigma_0:=\begin{pmatrix}
                                         z_0^S & c_0\\
                                         c_0   & z_0^Y
                                      \end{pmatrix},\; c_0:=\rho\min\{z_0^S,z_0^Y\},
\end{equation}
for given constant parameters $\lambda_0^S,\lambda_0^Y, z_0^S,z_0^Y$.
\end{definition}
From Remark \ref{rem:solution to the stock price SDE}, the solutions of the asset prices to the SDEs (\ref{eq:dS}), (\ref{eq:dY}) are
\begin{equation}
\label{eq:asset price solutions full}
\begin{aligned}
S_t&=S_0\exp\left(\sigma^S((\lambda^S-\tfrac12\sigma^S)t+W_t^S)\right),\\
Y_t&=Y_0\exp\left(\sigma^Y((\lambda^Y-\tfrac12\sigma^Y)t+W_t^Y)\right),
\end{aligned}
\end{equation}
from which the observation process may be expressed as deterministic functions of the asset prices and time,
\begin{equation*}
\xi_t^S=\xi^S(t,S_t)=\frac{1}{\sigma^S}\log\left(\frac{S_t}{S_0}\right)+\frac12\sigma^St,\quad \xi_t^Y=\xi^Y(t,Y_t)=\frac{1}{\sigma^Y}\log\left(\frac{Y_t}{Y_0}\right)+\frac12\sigma^Yt.
\end{equation*}
For any process $\zeta$ expressed by a function of time and current asset prices, the abbreviation $\zeta_t:=\zeta(t,S_t,Y_t)$ may be used. The SDEs of the observation and signal process (\ref{eq:observation}), (\ref{eq:signal}) are
\begin{equation*}
\d\Xi_t = \Lambda\d t+\begin{pmatrix}
                         1    & 0\\
                         \rho & \sqrt{1-\rho^2}
                      \end{pmatrix}\begin{pmatrix}
                                      W_t^S\\
                                      W_t^\perp
                                   \end{pmatrix}, \quad \d\Lambda=\begin{pmatrix}
                                                 0 \\
                                                 0
                                               \end{pmatrix}.
\end{equation*}
According to (\ref{eq:prior}), an unbiased estimator of $\Lambda$ is Gaussian with initial estimations for $\lambda_0^S,\lambda_0^Y, z_0^S, z_0^Y$.\par
The idea behind the Kalman-Bucy filter is to choose a prior distribution with specific parameter values for the MPR process $\Lambda$ and continuously update it over time. The prior distribution initialises the probability law of $\Lambda$ conditional on $\widehat{\mathcal{F}}_0$, and through filtering done in the next definition, this is updated with the evolution of the asset prices under the observation filtration $\widehat{\mathbb{F}}$. An in-depth discussion of this filtering procedure is made in Section~\ref{sec:conclusions}. We describe in Remark~\ref{rem:OU process} a partial information model apart from the Kalman-Bucy filter.
\begin{definition}[Kalman-Bucy filter]
\label{def:Kalman-Bucy filter}
The optimal filter process $\widehat{\Lambda}:=(\widehat{\Lambda}_t)_{0\leq t \leq T}$ defined by $\widehat{\Lambda}_t:=\mathbb{E}[\Lambda\,\vert\,\widehat{\mathcal{F}}_t]$, is the two-dimensional MPR process under conditional expectation, $\widehat{\lambda}_t^i:=\mathbb{E}[\lambda^i\,\vert\,\widehat{\mathcal{F}}_t]$ for $0\leq t\leq T$, $i\in\{S,Y\}$. The conditional covariance matrix process $\Sigma=(\Sigma_t)_{0\leq t\leq T}$ is given by
\begin{equation}
\label{eq:covariance}
\Sigma=\begin{pmatrix}
           z_t^S & c_t\\
           c_t   & z_t^Y
        \end{pmatrix},\quad
     z_t^i:=\mathbb{E}[(\lambda^i-\widehat{\lambda}_t^i)^2\,\vert\,\widehat{\mathcal{F}}_t],\quad
c_t:=\mathbb{E}[(\lambda^S-\widehat{\lambda}_t^S)(\lambda^Y-\widehat{\lambda}_t^Y)\,\vert\,\widehat{\mathcal{F}}_t],
\end{equation}
for $0\leq t\leq T$, $i\in\{S,Y\}$.
\end{definition}
The Kalman-Bucy filter transforms the partial into a full information scenario by replacing the constant parameters $\lambda^S,\lambda^Y$ by stochastic processes $\widehat{\lambda}^S,\widehat{\lambda}^Y$ and changing the filtration of the probability space from $\mathbb{F}$ to the observation filtration $\widehat{\mathbb{F}}$. The upcoming result of the partial information model under $\widehat{\mathbb{F}}$ come from Monoyios \cite{M2010}.
\begin{proposition}[Model under partial information]
\label{thm:model under partial information}
The Kalman-Bucy filter from Definition~\ref{def:Kalman-Bucy filter} converts the model from the partial to the full information scenario with asset price SDEs
\begin{equation}
\label{eq:dS, dY partial}
\d S_t = \sigma^S S_t (\widehat{\lambda}_t^S\d t + \d{\widehat{W}}_t^S),\quad \d Y_t = \sigma^Y Y_t (\widehat{\lambda}_t^Y\d t + \d{\widehat{W}}_t^Y), 
\end{equation}
on the filtered probability space $(\Omega, \widehat{\mathcal{F}}, \widehat{\mathbb{F}}, \mathbb{P})$, where $\widehat{W}^S,\widehat{W}^Y$ are $(\mathbb{P},\widehat{\mathbb{F}})$-Brownian motions with correlation $\rho$ according to $\widehat{W}_t^Y=\rho\widehat{W}_t^S+\sqrt{1-\rho^2}\widehat{W}_t^\perp$ and $\widehat{\lambda}^S, \widehat{\lambda}^Y$ are $\widehat{\mathbb{F}}$-adapted processes. For $\abs{\rho}\neq 1$, $z_0^i\leq z_0^j$ with $i,j\in\{S,Y\}$, the drift processes are
\begin{equation*}
\widehat{\lambda}_t^i=\frac{\lambda_0^i+z_0^i\xi_t^i}{1+z_0^i t},\quad \widehat{\lambda}_t^j=\frac{\lambda_0^j+w_0\xi_t^j}{1+w_0 t}-\rho\left(\frac{\lambda_0^i+w_0\xi_t^i}{1+w_0 t}-\widehat{\lambda}_t^i\right),\quad 0\leq t\leq T
\end{equation*}
where $w_0:=\frac{z_0^j-\rho^2 z_0^i}{1-\rho^2}$ for $z_0^i<z_0^j$ and $w_0:=z_0^i$ for $z_0^S=z_0^Y$. They satisfy the SDEs
\begin{equation}
\label{eq:dlambda}
\begin{aligned}
\d{\widehat{\lambda}}_t^i&=z_t^i\d{\widehat{W}}_t^i=z_t^i(\d\xi_t^i-\widehat{\lambda}_t^i\d t),\quad \widehat{\lambda}_0^i=\lambda_0^i,\\
\d{\widehat{\lambda}}_t^j-\rho\d{\widehat{\lambda}}_t^i&=w_t(\d{\widehat{W}}_t^j-\rho\d{\widehat{W}}_t^i)=w_t(\d{(\xi_t^j-\rho\xi_t^j)}-(\widehat{\lambda}_t^j-\rho\widehat{\lambda}_t^i)\d t),\quad \widehat{\lambda}_0^j=\lambda_0^j,
\end{aligned}
\end{equation}
with the entries of the covariance matrix in (\ref{eq:covariance}), given by
\begin{equation*}
z_t^i=\frac{z_0^i}{1+z_0^i t},\quad z_t^j=\rho^2 z_t^i+(1-\rho^2)w_t,\quad w_t:=\frac{w_0}{1+w_0 t},\quad c_t=\rho z_t^i,\quad 0\leq t\leq T.
\end{equation*}
\end{proposition}
\begin{proof}
A proof can be found in Monoyios \cite[Proposition 1]{M2010}.
\end{proof}
Thus, under partial information, the investor's portfolio wealth dynamics from (\ref{eq:dX}) is transformed into
\begin{equation}
\label{eq:dX partial}
\d X_t=\sigma^S \pi_t \left(\widehat{\lambda}_t^S\d t + \d{\widehat{W}}_t^S\right).
\end{equation}
If the prior variances $z_0^S, z_0^Y$ are identical, then $z_t:=z_t^S=z_t^Y=w_t$ and hence $\d{\widehat{\lambda}}_t^Y=z_t\d{\widehat{W}}_t^Y$ holds for all $t\in[0,T]$. This case of filtering is similar to the two one-dimensional Kalman-Bucy filters on each asset as developed in \cite{M2007}. According to Proposition~\ref{thm:model under partial information}, the MPR of the asset prices have the dependencies
\begin{alignat}{3}
\notag
\widehat{\lambda}_t^S&=\widehat{\lambda}^S(t,S_t),\quad &&\widehat{\lambda}_t^Y=\widehat{\lambda}^Y(t,S_t,Y_t),\quad &&\text{if }z_0^S<z_0^Y,\\
\notag
\widehat{\lambda}_t^S&=\widehat{\lambda}^S(t,S_t),\quad &&\widehat{\lambda}_t^Y=\widehat{\lambda}^Y(t,Y_t),\quad &&\text{if }z_0^S=z_0^Y,\\
\label{eq:MPR dependencies}
\widehat{\lambda}_t^S&=\widehat{\lambda}^S(t,S_t,Y_t),\quad &&\widehat{\lambda}_t^Y=\widehat{\lambda}^Y(t,Y_t),\quad &&\text{if }z_0^S>z_0^Y,
\end{alignat}
solving the SDEs
\begin{alignat*}{3}
\d{\widehat{\lambda}}_t^S&=z_t^S\d{\widehat{W}}_t^S,\quad &&\d{\widehat{\lambda}}_t^Y-\rho\d{\widehat{\lambda}}_t^S=w_t(\d{\widehat{W}}_t^Y-\rho\d{\widehat{W}}_t^S),\quad &&\text{if }z_0^S<z_0^Y,\\
\d{\widehat{\lambda}}_t^S&=w_t\d{\widehat{W}}_t^S,\quad &&\d{\widehat{\lambda}}_t^Y=w_t\d{\widehat{W}}_t^Y,\quad &&\text{if }z_0^S=z_0^Y,\\
\d{\widehat{\lambda}}_t^Y&=z_t^Y\d{\widehat{W}}_t^Y,\quad &&\d{\widehat{\lambda}}_t^S-\rho\d{\widehat{\lambda}}_t^Y=w_t(\d{\widehat{W}}_t^S-\rho\d{\widehat{W}}_t^Y),\quad &&\text{if }z_0^S>z_0^Y.
\end{alignat*}
In the remainder of this thesis, except where otherwise stated, we are working with the partial information model of Proposition \ref{thm:model under partial information}.\newpage
\begin{remark}[Ornstein-Uhlenbeck model for MPRs]
\label{rem:OU process}
A more complicated method of modelling the MPRs would implicate an unknown stochastic process for each unknown MPR. For instance, the MPR dynamics could be expressed as processes of \textit{Ornstein-Uhlenbeck type},
\begin{equation}
\label{eq:Ornstein-Uhlenbeck}
\d\lambda_t^i = \eta^i (\nu^i-\lambda_t^i)\d t + \delta^i\d B_t^i,\quad i=S,Y,
\end{equation}
with Brownian motions $B^S,B^Y$ and constant \textit{mean reversion rates} $\eta^S, \eta^Y$, \textit{mean reversion levels} $\nu^S,\nu^Y$ and volatilities $\delta^S,\delta^Y$. The mean reversion level represents the equilibrium or long-term mean of the MPR variable and the mean reversion rate the velocity by which the MPR variable reverts to its equilibrium. The volatility defines the impact of stochastic shocks on the MPR change. The resulting issue would contain the estimations of these unknown parameters. Brendle \cite{Br2006} modelled the drifts by a multidimensional Ornstein-Uhlenbeck model in the context of a power-utility based optimal portfolio problem under partial information. However, it is not clear how the model parameters from above can be estimated using real market data, because they are assumed as known constants. This model is not pursued here to seek maximally explicit formulas for the valuation and optimal hedge. More applications of stochastic differential equations of Ornstein-Uhlenbeck type in financial economics were treated by, for instance, Barndorff-Nielsen and Shephard \cite{BNS2001}.
\end{remark}
\subsection{Forward utility and local risk tolerance}
\label{subsec:forward utility and local risk tolerance}
In the \textit{classical utility framework}, the expected utility criteria is typically formulated through a deterministic, concave and increasing function of terminal wealth, where both the investment time horizon $T$ and the associated risk preferences are chosen a priori. The \textit{value function} as the optimal solution in the relevant market model has the fundamental property of supermartingality for arbitrary investment strategies and martingality at an optimum, which is a consequence of the \textit{dynamic programming principle} (see, for example, Merton \cite[p.~249]{Me1969}). Since the classical utility $U(x)$ is fixed at a time $T$ and its \textit{value function} $v(t,x)$ generated at previous times $t\in[0,T]$ with the wealth argument $x$, it is also called the \textit{backward utility} by Musiela and Zariphopoulou \cite[pp.~304,~315]{MZ2007}. As depicted therein, backward utilities does not accurately capture future changes in the risk preference as the market environment evolves. Therefore, they introduce a new class of dynamic utilities that are constructed forward in time, which offers flexibility with regards to the a priori choices mentioned above while the natural optimality properties of the value function process is preserved. Contrary to the classical utility framework, the \textit{forward utility} $U_t(x)$ is normalised at present time $t$ but not for a fixed investment horizon $T$, and generated for all future times via a self-financing criterion. The forward measurement criterion is defined by Musiela and Zariphopoulou \cite{MZ2009} in terms of a family of stochastic processes on $[0,\infty)$ indexed by a wealth argument.\par
In this section, only proofs are outlined in cases that they are instructive, otherwise the reference to the original source is given.
\begin{definition}[Forward performance process]
\label{def:forward performance process}
Let $\Theta_t\subseteq\Theta$ be the subset of strategies starting at $t$. An $\widehat{\mathcal{F}}_t$-adapted stochastic process $U:=(U_t(x))_{t\geq 0}$, where
\begin{itemize}
\item[(i)] for each $t\geq 0$ the function $U_t\colon x\mapsto U_t(x)$ is concave and increasing in $x\in\mathbb{R}$,
\item[(ii)] for each $t\geq 0$ and each self-financing strategy $\theta\in\Theta_t,\;\mathbb{E}[U_t(X_t)]^+<\infty$,
\begin{equation*}
\quad\mathbb{E}[U_s(X_s)\,\vert\,\widehat{\mathcal{F}}_t]\leq U_t(X_t),\quad s\geq t,
\end{equation*}
\item[(iii)] there exists a self-financing (optimal) strategy $\theta^*\in\Theta_t$, for which
\begin{equation*}
\mathbb{E}[U_s(X_s^*)\,\vert\,\widehat{\mathcal{F}}_t]=U_t(X_t^*),\quad s\geq t,
\end{equation*}
\item[(iv)] it satisfies the initial datum $U_0(x)=u_0(x)$ at $t=0$ for all $x\in\mathbb{R}$, with a concave and increasing function $u_0\colon\mathbb{R}\to\mathbb{R}$ of wealth,
\end{itemize}
is called a \textit{forward performance process}.
\end{definition}
The function $U_t$ in (i) is the \textit{(forward) performance function}. The conditions (ii) and (iii) represent the supermartingality and martingality properties, respectively.\par
Among others, forward formulations of optimal control problems were proposed and studied in the past by Seinfeld and Lapidus \cite{SL1968} and Vit \cite{V1977} for the deterministic case and Kurtz \cite{Ku1984} for the stochastic case. As in \cite{MZ2009}, we will consider a special class of time-decreasing and time-monotone forward performance processes expressed by a deterministic function $u(x,t)$ of wealth and time, where the time argument is replaced by an increasing process $A:=(A_t)_{t\geq 0}$ depending on the market coefficients and not the investor's preferences. On the contrary, the function $u$ is independent of market changes and only depends on the initial datum $u_0$ satisfying a market independent differential constraint for $t\geq 0$.
\begin{definition}[Mean-variance trade-off process]
\label{def:mean-variance trade-off}
For the stock's market price of risk process $\widehat{\lambda}^S$, define by
\begin{equation}
\label{eq:mean-variance trade-off}
A_t:=\int_0^t\left(\widehat{\lambda}_u^S\right)^2\d u
\end{equation}
the \textit{(mean-variance) trade-off process} $A:=(A_t)_{t\geq 0}$.
\end{definition}
The naming originates from Pham et al. \cite[pp.~173--174]{PRS1998} in the context of mean-variance hedging of continuous processes. By Definition~\ref{def:observation and signal process} and Definition~\ref{def:Kalman-Bucy filter}, the process $\lambda^S(t,S_t,Y_t)$ and likewise the trade-off process $A(t,S_t,Y_t)$ are, in general, dependent of the asset prices $S,Y$. In the full information scenario, when $\widehat{\mathbb{F}}=\mathbb{F}$, the trade-off process simplifies to $A_t=\left(\lambda^S\right)^2t$ as the MPR $\widehat{\lambda}_t^S=\lambda^S$ becomes constant.
\begin{proposition}[Forward performance process and general optimal strategy]
\label{thm:forward performance process and general optimal strategy}
Let $u\colon\mathbb{R}\times[0,\infty)\to\mathbb{R}$ be a concave and increasing function of the wealth argument with $u\in\mathcal{C}^{3,1}$, satisfying the non-linear partial differential equation
\begin{equation}
\label{eq:u PDE quadratic form}
u_tu_{xx}=\frac12u_{x}^2
\end{equation}
and the initial condition $u(x,0)=u_0(x)$, where $u_0\in\mathcal{C}^3(\mathbb{R})$. Then, the time-decreasing process $U:=(U_t(x))_{t\geq 0}$ defined by
\begin{equation}
\label{eq:forward performance process}
U_t(x):=u(x,A_t),
\end{equation}
is a forward performance process with the time argument replaced by the mean-variance process (\ref{eq:mean-variance trade-off}) of Definition~\ref{def:mean-variance trade-off}. Moreover, the optimal trading strategy is given by
\begin{equation}
\label{eq:general optimal strategy forward utility}
\pi_t^*=-\frac{\widehat{\lambda}_t^S}{\sigma^S}\frac{u_x(X_t^*,A_t)}{u_{xx}(X_t^*,A_t)},
\end{equation}
where $X^*$ is the associated wealth process following (\ref{eq:dX}) with $\pi_t^*=\theta_t^*S_t$.
\end{proposition}
\begin{proof}
We refer to Musiela and Zariphopoulou \cite[Proposition 3]{MZ2009}.
\end{proof}
The monotonicity of $u(x,A_t)$ follows from the related assumptions on $u$ and the time-monotonicity is obtained from the definition of $A$ and from the fact that the time-derivative of $u$ is negative, $u_t<0$. The process $U$ will be simply also referred to as \textit{forward utility} or \textit{dynamic utility}. The mean-variance trade-off process $A$ behaves as a stochastic time change of the deterministic utility function $u(x,t)$. Optimal portfolio choice problems under space-time monotonicity was studied in detail by Musiela and Zariphopoulou \cite{MZ2010}. As the representation (\ref{eq:general optimal strategy forward utility}) shows, the optimal strategy does not directly depend on $u$ but on the differential quantity $-\frac{u_x}{u_{xx}}$, which was separately analysed by Zariphopoulou and Zhou \cite{ZZ2007}.
\begin{definition}[Local risk tolerance]
\label{def:local risk tolerance}
The \textit{local risk tolerance function} is
\begin{equation}
\label{eq:local risk tolerance}
r\colon\mathbbm{R}\times[0,\infty)\longrightarrow [0,\infty),\quad (x,t)\longmapsto -\frac{u_x(x,t)}{u_{xx}(x,t)},
\end{equation}
with initial function $r(x,0)=r_0(x)=-\frac{u_x(x,0)}{u_{xx}(x,0)}=-\frac{u_0'(x)}{u_0''(x)}$ and $u$ satisfying (\ref{eq:u PDE quadratic form}). The \textit{local risk tolerance process} $R:=(R_t)_{t\geq 0}$ is defined by $R_t:=r(X_t,A_t)$.
\end{definition}
Using Definition~\ref{def:local risk tolerance} in (\ref{eq:general optimal strategy forward utility}), the dynamics (\ref{eq:dX partial}) of the optimal wealth $X^*$ can be expressed as
\begin{equation}
\label{eq:dX optimum risk tolerance}
\d X_t^*=R_t^*\widehat{\lambda}_t^S\left(\widehat{\lambda}_t^S\d t +\d{\widehat{W}}_t^S\right)
\end{equation}
with $R_t^*=r(X_t^*,A_t)$ as the \textit{local risk tolerance process benchmarked at optimal wealth}. This brings up the question whether $u$ can be indirectly derived from $r$.
\begin{corollary}[Transport equation]
\label{thm:transport equation}
The utility function $u$ satisfies the \textit{transport equation}
\begin{equation}
\label{eq:transport equation}
u_t+\frac12r(x,t)u_x=0.
\end{equation}
With the knowledge of $r$ the first-order partial differential equation (\ref{eq:transport equation}) can be solved to yield $u$.
\end{corollary}
\begin{proof}
The transport equation (\ref{eq:transport equation}) follows from (\ref{eq:u PDE quadratic form}) and (\ref{eq:local risk tolerance}). It can be solved using the method of characteristics. Consider $\frac{\d{}}{\d t}u(\tilde{x}(t),t)=\tilde{x}'(t)u_x(\tilde{x}(t),t)+u_t(\tilde{x}(t),t)$ with the characteristic curves $\tilde{x}(t)$. The solution are the curves whose slope is equal to half of the risk tolerance, i. e. $\tilde{x}'(t)=\frac12 r(\tilde{x}(t),t)$ with initial value $\tilde{x}(0)=x$. Then, the function $u$ can be successively constructed through the initial condition $u_0$ computed from Definition~\ref{def:local risk tolerance} and its evaluation along the characteristic curves.
\end{proof}
This means that for an infinitesimal time interval $(0,\epsilon)$, the performance level $u(x+\frac12 r_0(x)\epsilon,\epsilon)$ at time $\epsilon$ is identical to $u(x,0)$ at $t=0$, when the wealth is moved from $x$ to a higher level $x+\frac12 r_0(x)\epsilon$. The infinitesimal amount $\frac12 r_0(x)\epsilon$ can be interpreted as the compensation required by the investor in order to satisfy his impatience in the time interval $(0,\epsilon)$. More about the theory of investor's impatience can be found in Fisher \cite[Chapter IV]{F1930}, Koopmans \cite[p.~296]{Ko1960} and Diamond et al. \cite{DKW1964}.\par
Apart from the transport equation (\ref{eq:transport equation}), the local risk tolerance function $r$ solves an \textit{autonomous non-linear heat equation}, which gives an alternative approach to construct $u$ from $r$.
\begin{corollary}[Fast diffusion equation]
Let $u\in\mathcal{C}^4(\mathbbm{R}\times[0,\infty))$ satisfy the conditions from Proposition~\ref{thm:forward performance process and general optimal strategy}. Then, the associated local risk tolerance $r$ is the solution of an equation of fast diffusion type, namely
\begin{equation}
\label{eq:r fast diffusion equation}
r_t+\frac12r^2r_{xx}=0\quad\text{and}\quad r(x,0)=-\frac{u_0'(x)}{u_0''(x)}.
\end{equation}
\end{corollary}
\begin{proof}
We will quote the proof from Musiela and Zariphopoulou \cite[Proposition 6]{MZ2009}. Differentiating the non-linear partial differential equation $u_t=\frac12\frac{u_x^2}{u_{xx}}$ from (\ref{eq:u PDE quadratic form}) with respect to $t$ yields
\begin{equation}
\label{eq:u derivative tx}
u_{tx}=u_x-\frac12 u_x\left(\frac{u_xu_{xxx}}{u_{xx}^2}\right),
\end{equation}
and a second derivation with respect to $x$ gives
\begin{equation}
\label{eq:u derivative txx}
u_{txx}=u_{xx}-u_{xx}\left(\frac{u_xu_{xxx}}{u_{xx}^2}\right)-\frac12 u_x\left(\frac{u_xu_{xxx}}{u_{xx}^2}\right)_x.
\end{equation}
The spatial derivatives of the local risk tolerance (\ref{eq:local risk tolerance}) are
\begin{equation}
\label{eq:r derivatives}
r_x = -\frac{u_{xx}^2-u_xu_{xxx}}{u_{xx}^2}=-1+\frac{u_xu_{xxx}}{u_{xx}^2},\quad r_{xx}=\left(\frac{u_xu_{xxx}}{u_{xx}^2}\right)_x.
\end{equation}
The preceding identities (\ref{eq:r derivatives}), (\ref{eq:u derivative txx}) and (\ref{eq:u derivative tx}) imply
\begin{align*}
r_t+\frac12r^2r_{xx}&=-\frac{u_{tx}}{u_{xx}}+\frac{u_xu_{txx}}{u_{xx}^2}+\frac12\frac{u_x^2}{u_{xx}^2}\left(\frac{u_xu_{xxx}}{u_{xx}^2}\right)_x\\
&=-\frac{u_{tx}}{u_{xx}}+\frac{u_x}{u_{xx}}-\frac12\frac{u_x}{u_{xx}}\left(\frac{u_xu_{xxx}}{u_{xx}^2}\right)\\
&=0,
\end{align*}
which proves the assertion.
\end{proof}
After choosing an initial condition $r_0(x)=r(x,0)=-\frac{u'(x)}{u''(x)}$, the initial datum $u(x,0)=u_0(x)$ and furthermore, with (\ref{eq:r fast diffusion equation}) and $r_0$, the values of $r(x,t)$ for $t>0$ can be retrieved. The function $u(x,t)$, $t>0$ can be computed through successive integration from (\ref{eq:local risk tolerance}) if certain quantities are correctly specified.
\begin{corollary}[Autonomous SDE system for $(X^*,R^*)$]
Let $r$ satisfy (\ref{eq:r fast diffusion equation}) and let $A$ be as in (\ref{eq:mean-variance trade-off}). Then, the processes $X^*$ and $R^*$ solve the system
\begin{align*}
\d X_t^*=R_t^*\widehat{\lambda}_t^S\left(\widehat{\lambda}_t^S\d t +\d{\widehat{W}}_t^S\right),\quad \d R_t^*=r_x(X_t^*,A_t)\d X_t^*,
\end{align*}
for $t>0$.
\end{corollary}
\begin{proof}
The first equation of the optimal wealth dynamics is taken from (\ref{eq:dX optimum risk tolerance}). Its quadratic variation is
\begin{equation}
\label{eq:dX optimum quadratic variation}
\d\langle X^*\rangle_t=R_t^2\left(\widehat{\lambda}_t^S\right)^2\d t\xlongequal{(\ref{eq:mean-variance trade-off})}R_t^2\d A_t.
\end{equation}
Using It\^{o}'s lemma, the dynamics of the risk tolerance process at optimum wealth can be deduced by
\begin{align*}
\d R_t^*&=\d r(X_t^*,A_t)=r_x(X_t^*,A_t)\d X_t^*+r_t(X_t^*,A_t)\d A_t+\frac12r_{xx}(X_t^*,A_t)\d\langle X^*\rangle_t\\
&=r_x(X_t^*,A_t)\d X_t^*+\left(r_t(X_t^*,A_t)+\frac12r_{xx}(X_t^*,A_t)R_t^2\right)\d A_t\\
&\xlongequal{(\ref{eq:dX optimum quadratic variation})}r_x(X_t^*,A_t)\d X_t^*+\left(r_t(X_t^*,A_t)+\frac12r^2(X_t^*,A_t)r_{xx}(X_t^*,A_t)\right)\d A_t\\
&\xlongequal{(\ref{eq:r fast diffusion equation})}r_x(X_t^*,A_t)\d X_t^*,
\end{align*}
because $r$ solves the fast diffusion equation.
\end{proof}
The reciprocal of the local risk tolerance is called \textit{local risk aversion}, which solves a similar partial differential equation of second order. The risk aversion is a well-known parameter in utility theory to express the investor's \textit{risk preference}.
\begin{corollary}[Local risk aversion]
\label{thm:local risk aversion}
The \textit{local risk aversion function}, defined as
\begin{equation}
\label{eq:local risk aversion}
\gamma\colon\mathbb{R}\times[0,\infty)\longrightarrow(0,\infty),\quad(x,t)\longmapsto\frac{1}{r(x,t)},
\end{equation}
satisfies the partial differential equation
\begin{equation}
\label{eq:gamma PDE}
\gamma_t=\frac12\left(\frac{1}{\gamma}\right)_{xx},\quad \gamma(x,0)=-\frac{u_0''(x)}{u_0'(x)},
\end{equation}
where $u$ is the local risk tolerance function from Definition~\ref{def:local risk tolerance}.
\end{corollary}
\begin{proof}
By (\ref{eq:local risk aversion}), insert $r=\frac{1}{\gamma}$ into the fast diffusion equation (\ref{eq:r fast diffusion equation}) solved by $r$ to get
\begin{equation*}
0\xlongequal{\text{(\ref{eq:r fast diffusion equation})}}r_t+\frac12r^2r_{xx}\xlongequal{\text{(\ref{eq:local risk aversion})}}-\frac{\gamma_t}{\gamma^2}+\frac{1}{2\gamma^2}\left(\frac{1}{\gamma}\right)_{xx}=-\frac{1}{\gamma^2}\left(\gamma_t-\frac12\left(\frac{1}{\gamma}\right)_{xx}\right),
\end{equation*}
which directly implies the partial differential equation (\ref{eq:gamma PDE}) for $\gamma$.
\end{proof}
The partial differential equation (\ref{eq:gamma PDE}) is of \textit{porous medium type}; see for example Vasquez \cite{V2006}. This and the \textit{fast diffusion equation} (\ref{eq:r fast diffusion equation}) may not have well-defined global solutions for arbitrary initial conditions. Zariphopoulou and Zhou \cite{ZZ2007} introduced a two-parameter family of so-called \textit{asymptotically linear local risk tolerance functions} solving (\ref{eq:r fast diffusion equation}), which includes the most common cases that lead to \textit{exponential}, \textit{power}, and \textit{logarithmic utilities}.
\begin{proposition}[Asymptotically linear local risk tolerance]
\label{thm:asymptotically linear local risk tolerance}
Let $\alpha,\beta>0$ be constant parameters, then the function
\begin{equation}
\label{eq:asymptotically linear local risk tolerance}
r(x,t)=\sqrt{\alpha x^2+\beta e^{-\alpha t}},\quad (x,t)\in\mathbbm{R}\times[0,\infty),
\end{equation}
solves (\ref{eq:r fast diffusion equation}) with initial datum $r_0(x)=\sqrt{\alpha x^2+\beta}$. The limiting cases lead to local risk tolerance functions with corresponding utilities for $t\geq 0$ as follows:
\begin{alignat}{3}
\label{eq:exponential}
&\lim_{\alpha\to 0}r(x,t)=\sqrt{\beta}=:r_e,\quad &&u(x,t)=-e^{-\frac{x}{\sqrt{\beta}}+\frac{t}{2}},\;x\in\mathbbm{R},\quad &&\text{(exponential)};\\
\label{eq:power}
&\lim_{\beta\to 0}r(x,t)=\sqrt{\alpha}x,\quad &&u(x,t)=\frac{x^\delta}{\delta}e^{-\frac12\frac{\delta}{1-\delta}t},\;x\geq 0,\;\alpha\neq 1,\quad &&\text{(power)};\\
\label{eq:logarithmic}
&\lim_{\beta\to 0}r(x,t)=x,\quad &&u(x,t)=\log(x)-\frac{t}{2},\;x>0,\;\alpha=1,\quad &&\text{(logarithmic)},
\end{alignat}
where $\delta:=\frac{\sqrt{\alpha}-1}{\sqrt{\alpha}}$.
\end{proposition}
\begin{proof}
The first partial derivatives of $r$ are $r_t=-\frac12\alpha\beta e^{-\alpha t}r^{-1}$ and $r_x=\alpha x r^{-1}$. The second derivative with respect to $x$ is $r_{xx}=\alpha r^{-1}-(\alpha x )^2r^{-3}$. As a result, (\ref{eq:asymptotically linear local risk tolerance}) solves the fast diffusion equation
\begin{align*}
r_t+\frac12 r^2r_{xx}&=-\frac12\alpha\beta e^{-\alpha t}r^{-1}+\frac12\left(\alpha r-(\alpha x)^2 r^{-1}\right)=-\frac12\alpha\underbrace{\left(\alpha x^2+\beta e^{-\alpha t}\right)}_{=r^2}r^{-1}+\frac12\alpha r\\
&=-\frac12(r-r)=0.
\end{align*}
To construct the utilities in the limiting cases, Definition~\ref{def:local risk tolerance} and the transport equation from Corollary~\ref{thm:transport equation} can be applied. For the exponential case (\ref{eq:exponential}), consider $\sqrt{\beta}=\frac{u_x(x,t)}{u_{xx}(x,t)}$ from (\ref{eq:local risk tolerance}) and make the exponential ansatz $u_0(x)=e^{-\frac{x}{\sqrt{\beta}}}$ for $t=0$ as the left side is independent of time. By (\ref{eq:transport equation}), it is $u_t+\frac12\sqrt{\beta}u_{x}=0$, which yields the product solution $u(x,t)=u_0(x)e^{\frac{t}{2}}=e^{-\frac{x}{\beta}+\frac{t}{2}}$.\par
In the power case (\ref{eq:power}), the risk tolerance is expressed by $\sqrt{\alpha}x=-\frac{u_0'(x)}{u_0''(x)}$. Similar to the exponential case, make a multiplicative ansatz $u(x,t)=u_0(x)\tilde{u}(t)$, but with a monomial initial function $u_0=\frac{x^\delta}{\delta}$ to solve the problem, because of
\begin{equation*}
-\frac{u_0'}{u_0''}=-\frac{x^{\delta-1}}{(\delta-1)x^{\delta-2}}=-\frac{x}{\delta-1}=\sqrt{\alpha}x.
\end{equation*}
Further, the transport equation gives the homogeneous ordinary differential equation of first order $\frac{x^\delta}{\delta}\tilde{u}'+\frac12\sqrt{\alpha}x^\delta\tilde{u}=0$. Excluding the trivial solution $u=\tilde{u}=0$ for $x=0$ simplifies the equation to $\tilde{u}'+\frac12\frac{\delta}{1-\delta}\tilde{u}=0$. After a rearrangement, we get the logarithmic derivative $\frac{\d{}}{\d t}\log(\tilde{u}(t))=\frac{\tilde{u}'(t)}{\tilde{u}(t)}=-\frac12\frac{\delta}{1-\delta}$, which is solved by simple integration and taking the inverse function, i. e. $\tilde{u}(t)=e^{-\frac12\frac{\delta}{1-\delta}t}$. The restriction $x\geq 0$ is needed to ensure only real solutions.\par
The logarithmic case (\ref{eq:logarithmic}) is easier after the power case. The logarithm function $u_0(x)=\log(x)$ solves $x=-\frac{u_0'}{u_0''}$ and the transport equation becomes $u_t+\frac12xu_x=0$. Since an attempt to solve the problem through a multiplicative separation fails, we try an additive approach to get an appropriate solution $u(x,t)=u_0(x)+\tilde{u}(t)$. With this, the transport equation is apparently solved by $\tilde{u}(t)=-\tfrac{t}{2}$. Obviously, the domain of the utility is defined only for $x>0$.
\end{proof}
The family (\ref{eq:asymptotically linear local risk tolerance}) is called \textit{asymptotically linear} due to its limiting behaviour
\begin{equation*}
\lim_{\abs{x}\to\infty}\frac{r(x,t)}{\abs{x}}=\sqrt{\alpha},\quad t\geq 0.
\end{equation*}
The local risk tolerances in (\ref{eq:exponential}), (\ref{eq:power}) and (\ref{eq:logarithmic}) are referred to as \textit{exponential}, \textit{power} and \textit{logarithmic} \textit{risk tolerance}, respectively. The form of $u$ only depends on the range of the parameter $\alpha$, specifically, one the cases $\alpha=1$ and $\alpha\neq 1$.
\begin{proposition}[Class of forward utility functions]
\label{thm:class of forward utility functions}
Let $r$ be an asymptotically linear local risk tolerance function as defined in (\ref{eq:asymptotically linear local risk tolerance}) with $\alpha,\beta>0$. The corresponding utility function is given by
\begin{align*}
u(x,t)=\begin{cases}
          M\frac{\left(\sqrt{\alpha}\right)^{1+\frac{1}{\sqrt{\alpha}}}}{\alpha-1}e^{\frac{1-\sqrt{\alpha}}{2}t}\frac{\left(\frac{\beta}{\sqrt{\alpha}}e^{-\alpha t}+(1+\sqrt{\alpha})x\left(\sqrt{\alpha}x+\sqrt{\alpha x^2+\beta e^{-\alpha t}}\right)\right)}{\left(\sqrt{\alpha x}+\sqrt{\alpha x^2+\beta e^{-\alpha t}}\right)^{1+\frac{1}{\sqrt{\alpha}}}}+N,\quad \alpha\neq 1\\
          \frac{M}{2}\left(\log\left(x+\sqrt{x^2+\beta e^{-t}}\right)-\frac{e^t}{\beta}x\left(x-\sqrt{x^2+\beta e^{-t}}\right)-\frac{t}{2}\right)+N,\;\; \alpha=1,
       \end{cases}
\end{align*}
for $(x,t)\in\mathbbm{R}\times[0,\infty)$, where $M>0,N\in\mathbbm{R}$ are constants derived from integration.
\end{proposition}
\begin{proof}
We refer to Zariphopoulou and Zhou \cite[Proposition 3.2]{ZZ2007}.
\end{proof}
To preserve the monotonicity of $u$, the constraint $M>0$ is necessary. As the utility is well-defined for all $x\in\mathbbm{R}$ with exception of the situation $\beta\to 0$, the non-negativity limitation on the investor's wealth is omitted. This property is useful for indifference valuation.
\section{Exponential forward valuation and hedging of European options under partial information}
\label{sec:european option}
The previous section prepared for the \textit{option's indifference pricing} and \textit{optimal hedging} of basis risk in an incomplete market model with partial information using a forward exponential utility approach. In this section, we derive the \textit{optimal hedging strategy}, the \textit{dual representation} of the forward indifference price with a PDE, the \textit{residual risk}, \textit{pay-off decompositions} and \textit{asymptotic expansions} of the indifference price as results.
\subsection{Perfect hedging in a complete market}
\label{subsec:Perfect hedging in a complete market}
Suppose the market is complete, this means that the Brownian motions $\widehat{W}^S,\widehat{W}^Y$ of the assets $S,Y$ are perfectly negatively or positively correlated with correlation coefficient $\abs{\rho}=1$. In this case, $Y$ effectively becomes a traded asset and perfect hedging of the stock $S$ by an \textit{European contingent claim} (European option) $C$ on $Y$ is possible due to the \textit{no-arbitrage requirement} of the market. More about the arbitrage theory of capital asset pricing can be found in Delbaen and Schachermayer \cite[p.~473]{DS1994} and Ross \cite{R1976}. The complete market case under full information was treated by Monoyios \cite[p.~334]{M2007}. An important result is that the perfect hedge does not require the knowledge of the MPR processes $\widehat{\lambda}^S,\widehat{\lambda}^Y$, making the \textit{hedging strategy} in the full and the partial information scenario identical.
\begin{proposition}[Pricing in a complete market]
In a complete market, that means a correlation of $\abs{\rho}=1$, the claim price process $C:=(C(t,Y_t))_{0\leq t\leq T}$ is given by the Black-Scholes pricing formula.
\end{proposition}
\begin{proof}
We will give a proof based on Davis \cite{D2006} and Monoyios \cite{M2004}. Apply the scenario under partial information with its notation from Subsection~\ref{subsec:partial information scenario}, because the calculations and results under full information are exactly the same. Without loss of generality let the correlation coefficient be $\rho=1$, implying the identity $\widehat{W}^Y=\rho \widehat{W}^S+\sqrt{1-\rho^2}\,\widehat{W}^\perp=\widehat{W}^S$. The no-arbitrage theory requires an unique market price of risk, since the random process $W^S$ is the only existing risk factor in the basis risk model. Therefore, the MPRs are related by $\widehat{\lambda}^S=\frac{\widehat{\mu}^S-r_m}{\sigma^S}=\frac{\widehat{\mu}^Y-r_m}{\sigma^Y}=\widehat{\lambda}^Y$ with $r=0$ according to Subsection \ref{subsec:full information scenario}. Like (\ref{eq:asset price solutions full}), the solutions of the asset price dynamics (\ref{eq:dS, dY partial}) are
\begin{equation}
\label{eq:dS, dY complete market}
\begin{aligned}
S_t &=S_0\exp\left(\sigma^S\left(\int_0^t\widehat{\lambda}_u^S\d u-\tfrac12\sigma^St+\widehat{W}_t^S)\right)\right),\\
Y_t &=Y_0\exp\left(\sigma^Y\left(\int_0^t\widehat{\lambda}_u^S\d u-\tfrac12\sigma^Yt+\widehat{W}_t^S)\right)\right).
\end{aligned}
\end{equation}
Thus, the asset $Y$ is a function of the stock $S$, given by
\begin{equation*}
\frac{Y_t}{Y_0}\xlongequal{(\ref{eq:dS, dY complete market})}\left(\frac{S_t}{S_0}\right)^{\frac{\sigma^Y}{\sigma^S}}\exp\left(\frac12\sigma^Y(\sigma^S-\sigma^Y)t\right).
\end{equation*}
Apply the It\^{o} lemma on the \textit{(contingent) claim price process} (value process of the European option on $Y$) $C:=\left(C(t,Y_t)\right)_{0\leq t\leq T}$, so that
\begin{align}
\label{eq:dC}
\d C&=C_t\d t+C_y\d Y_t+\frac12C_{yy}\d\langle Y\rangle_t\notag\\
&= \left(C_t+C_y\sigma^YY_t\widehat{\lambda}_t^S+\frac12C_{yy}\left(\sigma^Y\right)^2Y_t^2\right)\d t+C_y\sigma^YY_t\d{\widehat{W}}_t^S.
\end{align}
The replication conditions are $X_t^*=C(t,Y_t),\;\d X_t^*=\d C(t,Y_t),\;0\leq t\leq T,$ for the investor's optimal wealth $X_t^*$. A comparison of the random terms between the wealth dynamics (\ref{eq:dX partial}) and (\ref{eq:dC}) provides the \textit{perfect hedging strategy}
\begin{equation}
\label{eq:perfect hedge}
\theta_t^*=\frac{\sigma^Y}{\sigma^S}\frac{Y_t}{S_t}C_y(t,Y_t),
\end{equation}
which is independent of the MPRs and so is conform with both the full and partial information scenarios. The claim price process solves the Black-Scholes SDE
\begin{equation*}
C_t(t,Y_t)+\frac12 \left(\sigma^Y\right)^2Y_t^2C_y(t,Y_t)=0,
\end{equation*}
with a bounded continuous process $C(T,y)$ and the non-negative random variable $C(Y_T):=C(T,Y_T)$ as the \textit{pay-off at expiry T} of the European contingent claim.
\end{proof}
\subsection{Forward performance problem in an incomplete market}
\label{subsec:Forward performance problem in an incomplete market}
Now, suppose the market is incomplete, meaning that the correlation of $\widehat{W}^S,\widehat{W}^Y$ is not perfect, $\abs{\rho}\neq 1$. Then the claim is not perfectly replicable in general. The ensuing \textit{indifference valuation and hedging problem} of the claim is embedded in a \textit{exponential forward performance maximisation framework}. Firstly, we define essential terms of the valuation and hedging theory regardless of the specific forward utility and the information scenario.
\begin{definition}[Value process, indifference price and optimal hedging strategy]
\label{def:value process, indifference price and optimal hedging strategy}
Presume, the investor holds a long position in the stock $S$ and a short position in the claim $C$ on the non-traded asset $Y$ to hedge the stock. The maximal $\widehat{\mathcal{F}}_t$-conditional expected forward performance of \textit{terminal portfolio wealth} $X_T-C(Y_T)$ from trading,
\begin{equation}
\label{eq:value process C}
v^C(t,X_t,S_t,Y_t):=\esssup_{\theta\in\Theta_t}{\mathbb{E}\left[U_T(X_T-C(Y_T))\,\middle\vert\,\widehat{\mathcal{F}}_t\right]},\quad 0\leq t\leq T,
\end{equation}
is called \textit{(primal forward) value process}. When no claim is sold, the value process is
\begin{equation}
\label{eq:value process 0}
v^0(t,X_t,S_t,Y_t):=\esssup_{\theta\in\Theta_t}{\mathbb{E}\left[U_T(X_T)\,\middle\vert\,\widehat{\mathcal{F}}_t\right]},\quad 0\leq t\leq T.
\end{equation}
The terminal values (maximal expected performances) are
\begin{equation}
\label{eq:terminal value}
v^C(T,X_T,S_T,Y_T)=U_T(X_T-C(Y_T)),\quad v^0(T,X_T,S_T,Y_T)=U_T(X_T).
\end{equation}
The \textit{(forward performance) indifference price process} $p$ is defined by (Hodges and Neuberger \cite[p.~226]{HN1989})
\begin{equation}
\label{eq:indifference price}
v^C(t,X_t+p(t,S_t,Y_t),S_t,Y_t)=v^0(t,X_t,S_t,Y_t),\quad 0\leq t\leq T.
\end{equation}
Evaluating the value and indifference price processes given the deterministic point $( X_t,S_t,Y_t)=(x,s,y)$ delivers the \textit{value and indifference price functions} $v(t,x,s,y)$ and $p(t,s,y)$, respectively. We abbreviate with $\mathbb{E}_{t,x,s,y}[\,\cdot\,]$ the conditional expectation $\mathbb{E}[\,\cdot\,\vert\,(t,X_t,S_t,Y_t)=(t,x,s,y)]$. As in Becherer \cite[p.~7]{B2003} defined, the \textit{optimal hedging strategy}
\begin{equation}
\label{eq:optimal hedging strategy}
\theta^H:=(\theta_t^H)_{0\leq t\leq T},\quad\theta_t^H:=\theta_t^C-\theta_t^0,\quad 
\end{equation}
is the difference between the \textit{optimal strategy} $\theta^C:=(\theta_t^C)_{0\leq t \leq T}$ for the problem with the claim (\ref{eq:value process C}) and the optimal strategy $\theta:=(\theta_t^0)_{0\leq t\leq T}$ without the claim (\ref{eq:value process 0}).
\end{definition}
The notion of \textit{essential supremum} $\esssup$ (likewise \textit{essential infimum} $\essinf$) is taken from Karatzas and Shreve \cite[p.~323]{KS1998}. For a real-valued function $f$, it is $\esssup f:=\inf\{a\in\mathbbm{R}\,\vert\,\mu_\mathbbm{R}(f^{-1}(a,\infty))=0\}$ with the Lebesgue measure $\mu_\mathbbm{R}$.\par
By Definition~\ref{def:forward performance process}, it is $v^0(t,X_t,S_t,Y_t)=U_t(X_t^0)$ with the associated optimal wealth $X_t^0$ in absence of the claim. In terms of the stock-weighted trading strategy $\pi=\theta S$, the optimal hedging strategy is $\pi^H=\pi^C-\pi^0$. The portfolio strategies are denoted as $\theta_t=\theta(t,S_t,Y_t)$ ($\pi_t=\pi(t,S_t,Y_t)$) to express them as functions of the asset prices. The indifference price $p$ implicitly defined in (\ref{eq:indifference price}) is also called the \textit{writer's indifference price}, since the option $C$ in the portfolio is sold. The solution to the optimisation problem (\ref{eq:value process C}) in classical utility theory is well-studied in Zariphopoulou \cite{Z2001} and Monoyios \cite[p.~248]{M2004} using the so-called \textit{distortion transformation} to linearise the \textit{Hamilton-Jacobi-Bellman (HJB) equation} for the value function. References for the HJB equation are, for example, Pham \cite[pp.~42--46]{P2009} and \cite[p.~130]{KS1998}.
\begin{theorem}[Optimal strategy in terms of the value process]
\label{thm:optimal strategy in terms of the value process}
The general solution to the maximisation problems (\ref{eq:value process C}), (\ref{eq:value process 0}) with terminal performance values (\ref{eq:terminal value}) in terms of the value process is the optimal strategy process
\begin{equation}
\label{eq:optimal strategy value process}
\theta^*(t,S_t,Y_t)=-\left(\frac{\widehat{\lambda}_t^S v_x+\sigma^S S_t v_{xs}+\rho\sigma^Y Y_t v_{xy}}{\sigma^S S_t v_{xx}}\right),\quad 0\leq t\leq T,
\end{equation}
for $\theta^*=\theta^C,\theta^0$ and $v=v^C,v^0$.
\end{theorem}
\begin{proof}
The value function $v(t,x,s,y)$ solves the non-linear HJB equation
\begin{align}
\label{eq:HJB supremum v}
v_t +\sup_{\theta\in\Theta_t} \Big(&\sigma^S s \widehat{\lambda}_t^S(\theta_t v_x + v_s) + \sigma^Y y\widehat{\lambda}_t^Yv_y + \theta_t\left(\sigma^S s\right)^2v_{xs}+\rho\sigma^S\sigma^Y sy(\theta_t v_{xy}+v_{sy})\notag\\
&+\frac12\left(\sigma^S s\right)^2(\theta_t^2v_{xx}+v_{ss}) + \frac12\left(\sigma^Y y\right)^2v_{yy}\Big)=0,
\end{align}
where the supremum is derived through differentiation with respect to $\theta_t$,
\begin{equation*}
\sigma^S s\widehat{\lambda}_t^S v_x + \left(\sigma^S s\right)^2 v_{xs}+ \rho\sigma^S\sigma^Y syv_{xy} + \left(\sigma^S s\right)^2\theta_t^* v_{xx}=0,
\end{equation*}
which gives the \textit{optimal strategy function} $\theta^*(t,s,y)$. Evaluating the optimal strategy function at the random point $(t,S_t,Y_t)$  provides the optimal strategy process (\ref{eq:optimal strategy value process}).
\end{proof}
In terms of the cash value, the optimal strategy is
\begin{equation}
\label{eq:optimal strategy value function cash}
\pi_t^*=\theta_t^*S_t=-\frac{\widehat{\lambda}_t^S v_x}{\sigma^Sv_{xx}}-\frac{S_tv_{xs}}{v_{xx}}-\frac{\rho\sigma^Y Y_t v_{xy}}{\sigma^S  v_{xx}}.
\end{equation}
The first term $\pi_t^M:=-\frac{\widehat{\lambda}_t^S v_x}{\sigma^Sv_{xx}}$ of (\ref{eq:optimal strategy value function cash}) is called the \textit{Merton strategy}, because it made its first appearance in \cite[p.~250]{Me1969} as Merton's optimal solution in the setting of a simplified market with only one asset. Since in the full information scenario the value process does not directly depend on $S$ regarding the known MPR $\lambda^S$, the mixed partial derivative $v_{xs}$ is zero and thus the \textit{partial information component} $\pi_t^S:=\frac{S_tv_{xs}}{v_{xx}}$ of the strategy vanishes. The last term $\pi_t^Y:=\frac{\rho\sigma^Y Y_tv_{xy}}{v_{xx}}$ is induced by the claim on the non-traded asset $Y$ and reflects the \textit{hedging component} of the strategy. The sensitivity of the marginal utility of wealth with respect to changes of the option's price is measured by $v_{xy}$. For the uncorrelated case $\rho=0$, the hedging component $\pi_t^Y$ becomes zero and the stock cannot be hedged by the option. Therefore, the optimal strategy in the uncorrelated full information scenario is identical to the one of Merton.\par
By Proposition~\ref{thm:asymptotically linear local risk tolerance}, the limiting case (\ref{eq:exponential}) of the exponential linear local risk tolerance $\lim_{\alpha\to 0}r(x,t)=r_e=\sqrt{\beta}$ corresponds to the exponential utility function $u(x,t)=-e^{-\frac{x}{\sqrt{\beta}}+\frac{t}{2}}$. By Corollary~\ref{thm:local risk aversion}, the utility has the more familiar form $u(x,t)=-e^{-\gamma x+\frac{t}{2}}$ with the local risk aversion $\gamma=1/r_e=1/\sqrt{\beta}>0$. Applying Proposition~\ref{thm:forward performance process and general optimal strategy}, the \textit{exponential forward performance process} is
\begin{equation}
\label{eq:exponential forward performance}
U_t(x)=u(x,A_t)=-\exp\left(-\gamma x+\frac12\int_0^t\left(\widehat{\lambda}_u^S\right)^2\d u\right),\quad (x,t)\in\mathbbm{R}\times[0,\infty).
\end{equation}
In comparison to the classical exponential utility function $u(x)=-e^{-\gamma x}$ the dynamic utility decreases in time, valuing less future utility. The primal value process (\ref{eq:value process C}) is the maximal expected forward performance
\begin{equation}
\label{eq:primal problem}
\!\!v^C(t,X_t,S_t,Y_t)=\esssup_{\theta\in\Theta_t}\mathbb{E}\left[-\exp\left(-\gamma(X_T-C(Y_T))+\frac12\int_0^T\left(\widehat{\lambda}_u^S\right)^2\d u\right)\,\middle\vert\,\widehat{\mathcal{F}}_t\right].\!\!
\end{equation}
\begin{remark}[Wealth independence of indifference price]
As with classical utility, the forward indifference price $p$ defined in (\ref{eq:indifference price}) is independent of initial wealth $X_t$ under $\widehat{\mathcal{F}}_t$. This becomes clear, when regarding at
\begin{equation*}
v^C(t,X_t,S_t,Y_t)=e^{-\gamma X_t+\frac12\int_0^t\left(\widehat{\lambda}_u^S\right)^2\d u}\esssup_{\theta\in\Theta_t}\mathbb{E}\left[-e^{-\gamma\left(\int_t^T\theta_u\d S_u - C(Y_T)\right)+\frac12\int_t^T\left(\widehat{\lambda}_u^S\right)^2\d u}\,\middle\vert\,\widehat{\mathcal{F}}_t\right],
\end{equation*}
since it is $X_T=X_t+\int_t^T\theta_u\d S_u$ by (\ref{eq:dX}), where $X_t$ factors out of the problem.
\end{remark}
Since under full information, the stock's MPR $\lambda^S$ is observable and therefore a known constant with respect to the background filtration $\mathbb{F}$, the trade-off process $A_t=\int_0^t \left(\lambda^S\right)^2\d s=\left(\lambda^S\right)^2 t$ becomes a deterministic linear function of time. Hence, the investor's risk preference simplifies to the exponential performance process
\begin{equation*}
U_t(x)=-e^{-\gamma x+\frac12\left(\lambda^S\right)^2t},\quad (x,t)\in\mathbbm{R}\times[0,\infty).
\end{equation*}
Factoring out the MPR term, one gets the classical primal problem
\begin{equation*}
v^C(t,X_t,S_t,Y_t)=e^{\frac12\left(\lambda^S\right)^2T}\esssup_{\theta\in\Theta_t}\mathbb{E}\left[-e^{-\gamma (X_T-C(Y_T))}\,\vert\,\mathcal{F}_t\right],\quad 0\leq t\leq T.
\end{equation*}
\subsection{Dual representation of the stochastic control problem}
In the 1970s and 80s, Bismut \cite{Bi1973}, Karatzas et al. \cite{KLS1987}, \cite{KLSX1991} and Cox \& Huang \cite{CH1989} realised that the use \textit{dual methods} from convex analysis provided valuable comprehension of solutions to optimal stochastic control problems, which are more general than the original problem from Merton \cite{Me1969}. Kramkov and Schachermayer \cite{KS1999} studied the dual approach for solving maximisation problems under classical utility. Rogers \cite{R2003} delved deeper into the theory by applying methods from functional analysis and presenting various examples solved with duality methods. We follow \v{Z}itkovi\'{c} \cite{Zi2009} and Berrier et al. \cite{BRT2009} to briefly introduce the dual approach for solving forward performance maximisation problems.  Firstly, recall the notion of relative entropy $\mathcal{H}(\mathbb{Q},\mathbb{P})$ for equivalent local martingale measures $\mathbb{Q}\sim\mathbb{P}$ from Definition~\ref{def:relative entropy and admissible strategies}. For the subset $\mathcal{M}_{e,f}$ of these measures with finite relative entropy, we introduce the \textit{Radon-Nikodym derivative} (see Shreve \cite[pp.~65--79]{S2005} or Platen and Heath \cite[pp.~338--339]{PH2010}), allowing us to perform measure changes.
\begin{definition}[Radon-Nikodym derivative process]
\label{def:Radon-Nikodym derivative process}
For measures $\mathbb{Q}\in\mathcal{M}_{e,f}$, the positive likelihood ratio $(\mathbb{P},\widehat{\mathbb{F}})$-martingale process $Z^\mathbb{Q}=(Z_t^\mathbb{Q})_{0\leq t\leq T}$ defined by
\begin{equation}
\label{eq:Radon-Nikodym derivative process}
Z_t^\mathbb{Q}=\left.\frac{\d{\mathbb{Q}}}{\d{\mathbb{P}}}\right\vert_{\widehat{\mathcal{F}}_t},
\end{equation}
is called the \textit{Radon-Nikodym derivative process}. It is the \textit{density process} of $\mathbb{Q}$ with respect to $\mathbb{P}$.
\end{definition}
For admissible portfolio strategies (\ref{eq:admissible strategies}), $Z^\mathbb{Q}X:=(Z_t^\mathbb{Q}X_t)_{0\leq t\leq T}$ is a non-negative $(\mathbb{P},\widehat{\mathbb{F}})$-local martingale, hence a supermartingale satisfying
\begin{equation*}
Z_0^\mathbb{Q}=1,\quad \mathbb{E}[Z_T^\mathbb{Q}]=1,\quad \mathbb{E}[Z_T^\mathbb{Q}X_T\,\vert\,\widehat{\mathcal{F}}_t]\leq Z_t^\mathbb{Q}X_t\;\text{ almost surely}
\end{equation*}
(see \cite[pp.~2180--2181]{Zi2009} and \cite[p.~1]{BRT2009}).\par
The map $\mathbb{Q}\mapsto Z^\mathbb{Q}$ induced by (\ref{eq:Radon-Nikodym derivative process}) creates an one-to-one correspondence between the class $\mathcal{M}_{e,f}$ of equivalent locale martingale measures with finite relative entropy and the set of density processes $\mathcal{Z}:=\left\{Z^\mathbb{Q}\mid\mathbb{Q}\in\mathcal{M}_{e,f}\right\}$. If $t=T$, then the Radon-Nikodym derivative is $Z_T^\mathbb{Q}=\frac{\d{\mathbb{Q}}}{\d{\mathbb{P}}}$ and the relative entropy (\ref{eq:relative entropy}) can be expressed by
\begin{equation*}
\mathcal{H}(\mathbb{Q},\mathbb{P})=\mathbb{E}\left[Z_T^\mathbb{Q}\log Z_T^\mathbb{Q}\right]=\mathbb{E}^\mathbb{Q}\left[\log Z_T^\mathbb{Q}\right],
\end{equation*}
where $\mathbb{E}^\mathbb{Q}$ denotes the expectation with respect to $\mathbb{Q}$, whereas $\mathbb{E}$ is the $\mathbb{P}$-expectation. The relative entropy can be interpreted as a measure of distance, even though it is not a metric. The density process and relative entropy will be generalised to conditional versions in order to formulate the dual problem to Definition~\ref{def:value process, indifference price and optimal hedging strategy}.
\begin{definition}[Conditional density and conditional relative entropy]
\label{def:conditional relative entropy}
The ratio
\begin{equation}
\label{eq:conditional density}
Z_{t,T}^\mathbb{Q}:=\frac{Z_T^\mathbb{Q}}{Z_t^\mathbb{Q}},\quad 0\leq t\leq T,
\end{equation}
is called \textit{conditional density process} and motivates the \textit{conditional relative entropy} 
\begin{equation}
\label{eq:conditional relative entropy}
\mathcal{H}_t(\mathbb{Q},\mathbb{P}):=\mathbb{E}^\mathbb{Q}\left[\log Z_{t,T}^\mathbb{Q}\,\middle\vert\,\widehat{\mathcal{F}}_t\right]
\end{equation}
of $\mathbb{Q}$ with respect to $\mathbb{P}$ over the interval $[t,T]$.
\end{definition}
At $t=0$, the conditionality becomes trivial with density $Z_{0,T}^\mathbb{Q}=Z_T^\mathbb{Q}$ and relative entropy $\mathcal{H}_0(\mathbb{Q},\mathbb{P})=\mathcal{H}(\mathbb{Q},\mathbb{P})$. Frittelli \cite[p.~42]{F2000} showed the existence and uniqueness of a \textit{minimal entropy martingale measure (MEMM)} $\mathbb{Q}^E$, that minimises $\mathcal{H}(\mathbb{Q},\mathbb{P})$ over all $\mathbb{Q}\in\mathcal{M}_{e,f}$. According to Kabanov and Stricker \cite[pp.~131--132]{KS2002}, $\mathbb{Q}^E$ also minimises the $\mathcal{H}_t(\mathbb{Q},\mathbb{P})$ for an arbitrary $t\in[0,T]$, so that we can write
\begin{equation}
\label{eq:minimal martingale measure}
\mathbb{Q}^E:=\argmin_{\mathbb{Q}\in\mathcal{M}_{e,f}} \mathcal{H}_t(\mathbb{Q},\mathbb{P}).
\end{equation}
We say, the \textit{minimal conditional density process} $(Z_{t,T}^{\mathbb{Q}^E})_{0\leq t\leq T}$ minimises the \textit{conditional relative entropy process} $(\mathcal{H}_t(\mathbb{Q},\mathbb{P}))_{0\leq t\leq T}$.\par
Our aim is to give the optimal strategy from Theorem~\ref{thm:optimal strategy in terms of the value process} in terms of derivatives of the indifference price from Definition~\ref{def:value process, indifference price and optimal hedging strategy}, which we will approach through framing the \textit{dual problem}.
\begin{definition}[Convex conjugate (dual) performance and its inverse marginal]
\label{def:convex conjugate}
The \textit{(convex) conjugate} (or \textit{dual}) $\tilde{U}_t\colon (0,\infty)\to\mathbb{R}$ of the performance $U_t$ is
\begin{equation}
\label{eq:convex dual}
\tilde{U}_t(\tilde{x})=\esssup_{x>0}\left[U_t(x)-x\tilde{x}\right]=U_t(I_t(\tilde{x}))-\tilde{x}I(\tilde{x}),\quad t\geq 0,\,\tilde{x}>0,
\end{equation}
where $I_t:=(U_t')^{-1}$ denotes its \textit{inverse of the marginal} $U_t':=\frac{\d{}}{\d x}U_t$ satisfying
\begin{equation}
\label{eq:inverse of marginal}
U_t'(I_t(\tilde{x}))=I_t(U_t'(\tilde{x}))=\tilde{x},\quad t\geq 0,\,\tilde{x}>0.
\end{equation}
The conjugate function $\tilde{U}_t$ solves the bidual relation
\begin{equation*}
U_t(x)=\essinf_{\tilde{x}>0}\left[\tilde{U}_t(\tilde{x})+x\tilde{x}\right]=\tilde{U}_t(U_t'(x))+xU_t'(x),\quad t\geq 0,\, x>0,
\end{equation*}
as well as $\tilde{U}_t(\tilde{x})\geq U_t(x)-x\tilde{x}$ with equality if and only if $x=I_t(\tilde{x})$. The \textit{marginal dual performance} $\tilde{U}_t'$ satisfies the identity $\tilde{U}_t'(\tilde{x})=-I_t(\tilde{x})$.
\end{definition}
Both $U_t'$ and $I_t$ are continuous, strictly decreasing and map $(0,\infty)$ onto itself satisfying the \textit{Inada conditions} (see F\"are and Primont \cite{FP2002}, Inada \cite{I1963})
\begin{equation*}
I_t(0^+)=U_t'(0^+)=\infty,\quad I_t(\infty)=U_t'(\infty)=0,
\end{equation*}
where we have abbreviated $I_t(0^+)=\infty$ for the limit $\lim_{\tilde{x}\to 0^+}I_t(\tilde{x})=\infty$ (analogous the other limits). The conjugate function $\tilde{U}_t$ is convex, decreasing, continuously differentiable with the limits
\begin{equation*}
\tilde{U}_t'(0+)=-\infty,\quad\tilde{U}_t'(\infty)=0^+,\quad \tilde{U}_t(0^+)=U_t(\infty),\quad \tilde{U}_t(\infty)=U_t(0^+).
\end{equation*}
The dual function $\tilde{U}_t(\tilde{x})$ is the \textit{Legendre-transform} of $-U_t(-x)$ (cf.~Rockafellar \cite[p.~251]{R1970}). Pliska \cite{P1986} showed some useful applicatios for computing value functions and optimal strategies. The methods and the exposition of the results given there are similar to the corresponding methods used by \cite{R1970}.
\begin{lemma}[Dual value process and dual problem]
\label{thm:dual value process}
For the primal value function $v=v^C$ from Definition~\ref{def:value process, indifference price and optimal hedging strategy} the \textit{dual value process} is
\begin{equation}
\label{eq:dual value process}
\tilde{v}(t,\tilde{X}_t,S_t,Y_t):=\essinf_{Z^\mathbb{Q}\in\mathcal{Z}}\mathbb{E}\left[\tilde{U}_T(\tilde{X}_tZ_{t,T}^\mathbb{Q})-\tilde{X}_tZ_{t,T}^\mathbb{Q} C(Y_T)\,\middle\vert\,\widehat{\mathcal{F}}_t\right],\quad 0\leq t\leq T.
\end{equation}
The primal and dual value functions are conjugate with respect to the wealth space,
\begin{equation}
\begin{aligned}
\label{eq:bidual relation}
\tilde{v}(t,\tilde{x},s,y)&=\sup_{x>0}{[v(t,x,s,y)-x\tilde{x}]},\quad \tilde{x}>0,\\
v(t,x,s,y)&=\inf_{\tilde{x}>0}{[\tilde{v}(t,\tilde{x},s,y)+x\tilde{x}]},\quad x>0.
\end{aligned}
\end{equation}
The partial derivatives of he primal and dual value functions at the optimum are related by
\begin{equation}
\label{eq:conjugate derivative}
v_x(t,x^*,s,y)=\tilde{x}^*,\quad\tilde{v}_{\tilde{x}}(t,\tilde{x}^*,s,y)=-x^*.
\end{equation}
\end{lemma}
\begin{proof}
We refer to the theorems in Kramkov and Schachermayer \cite[pp.~908--911]{KS1999}, Rogers \cite[pp.~107--113]{R2003} and \v{Z}itkovi\'{c} \cite[pp.~2184--2188]{Zi2009}.
\end{proof}
As noticed in Mania and Schweizer \cite[p.~2116]{MS2005}, the terminology ``primal'' corresponds for any problem optimising over the portfolio strategy and ``dual'' when the optimiser is the density $Z^\mathbb{Q}$ resp. measure $\mathbb{Q}$.\par
Since the main results of duality theory for solving stochastic control problems are worked out, they will be applied to the primal performance maximisation problem (\ref{eq:primal problem}) of exponential forward type $U_t(X_t)=-\exp\left(-\gamma X_t+\frac12A_t\right)$ to obtain the corresponding dual problem. The next theorem covers the valuation of the dual performance process and the dual entropic representation of the problem.
\begin{theorem}[Dual forward performance problem]
\label{thm:dual forward performance problem}
The dual representation of the primal optimisation problem (\ref{eq:primal problem}) is given by
\begin{equation}
\label{eq:dual problem}
\!\!v^C(t,X_t,S_t,Y_t)=-\exp\!\left(\!-\gamma X_t-\essinf_{Z^\mathbb{Q}\in\mathcal{Z}}\!\left(\!\mathcal{H}_t(\mathbb{Q},\mathbb{P})\!-\!\gamma \mathbb{E}^\mathbb{Q}\!\left[C(Y_T)+\frac{1}{2\gamma} A_T\,\middle\vert\widehat{\mathcal{F}}_t\right]\!\right)\!\!\right)\!\!.\!\!
\end{equation}
\end{theorem}
\begin{proof}
Definition~\ref{def:convex conjugate} is considered to calculate the dual performance process. Inserting the derivative $U_t'(x)=\gamma\exp\left(-\gamma x+\frac12 A_t\right)$ into (\ref{eq:inverse of marginal}), leads to its inverse
\begin{equation*}
I_t(\tilde{x})=-\frac{1}{\gamma}\left(\log\left(\frac{\tilde{x}}{\gamma}\right)-\frac12 A_t\right).
\end{equation*}
Putting the inverse into (\ref{eq:convex dual}), provides the dual performance
\begin{equation}
\label{eq:dual performance process}
\tilde{U}_t(\tilde{x})=U_t(I_t(\tilde{x}))-\tilde{x}I(\tilde{x})=\frac{\tilde{x}}{\gamma}\left(\log\left(\frac{\tilde{x}}{\gamma}\right)-1-\frac12 A_t\right).
\end{equation}
Before moving to the dual value function, consider the conditional relative entropy
\begin{align}
\notag
\label{eq:conditional entropy under P-expectation}
\mathcal{H}_t(\mathbb{Q},\mathbb{P})&\xlongequal{(\ref{eq:conditional relative entropy})}\mathbb{E}^\mathbb{Q}\left[\log Z_{t,T}^\mathbb{Q}\,\middle\vert\,\widehat{\mathcal{F}}_t\right]=\frac{1}{Z_t^\mathbb{Q}}\mathbb{E}\left[Z_T^\mathbb{Q}\log Z_{t,T}^\mathbb{Q}\,\middle\vert\,\widehat{\mathcal{F}}_t\right]\\&=\mathbb{E}\left[\frac{Z_T^\mathbb{Q}}{Z_t^\mathbb{Q}}\log Z_{t,T}^\mathbb{Q}\,\middle\vert\,\widehat{\mathcal{F}}_t\right]\xlongequal{(\ref{eq:conditional density})}\mathbb{E}\left[Z_{t,T}^\mathbb{Q}\log Z_{t,T}^\mathbb{Q}\,\middle\vert\,\widehat{\mathcal{F}}_t\right],
\end{align}
where in the second equation we have applied Lemma 5.2.2 from Shreve \cite[p.~212]{S2004}, concerning the expression of conditional expectations of random variables under measure change. Obviously, the conditional density has the expectation
\begin{equation}
\label{eq:conditional density expectation}
\mathbb{E}\left[Z_{t,T}^\mathbb{Q}\,\middle\vert\,\widehat{\mathcal{F}}_t\right]=\frac{1}{Z_t^\mathbb{Q}}\mathbb{E}\left[Z_T^\mathbb{Q}\,\middle\vert\,\widehat{\mathcal{F}}_t\right]=\mathbb{E}^\mathbb{Q}\left[1\,\middle\vert\,\widehat{\mathcal{F}}_t\right]=1.
\end{equation}
Then, by Lemma~\ref{thm:dual value process}, the dual value function is given by
\begin{align}
\label{eq:dual value function entropic representation}
\notag
\tilde{v}(t,\tilde{x},s,y)&\xlongequal{(\ref{eq:dual value process})}\essinf_{Z^\mathbb{Q}\in\mathcal{Z}}\mathbb{E}_{t,x,s,y}\left[\tilde{U}_T(\tilde{x}Z_{t,T}^\mathbb{Q})-\tilde{x}Z_{t,T}^\mathbb{Q} C(Y_T)\right]\\\notag
&\xlongequal{(\ref{eq:dual performance process})}\essinf_{Z^\mathbb{Q}\in\mathcal{Z}}\mathbb{E}_{t,x,s,y}\left[\frac{\tilde{x}}{\gamma}Z_{t,T}^\mathbb{Q}\left(\log\left(\frac{\tilde{x}}{\gamma}Z_{t,T}^\mathbb{Q}\right)-1-\frac12 A_T\right)-\tilde{x}Z_{t,T}^\mathbb{Q}C(Y_T)\right]\\\notag
&=\frac{\tilde{x}}{\gamma}\left(\log\left(\frac{\tilde{x}}{\gamma
}\right)-1\right)\essinf_{Z^\mathbb{Q}\in\mathcal{Z}}\mathbb{E}_{t,x,s,y}\left[Z_{t,T}^\mathbb{Q}\right]\\\notag
&\quad+\frac{\tilde{x}}{\gamma}\essinf_{Z^\mathbb{Q}\in\mathcal{Z}}\mathbb{E}_{t,x,s,y}\left[Z_{t,T}^\mathbb{Q}\log Z_{t,T}^\mathbb{Q}-\gamma Z_{t,T}^\mathbb{Q}\left(C(Y_T)+\frac{1}{2\gamma}A_T\right)\right]\\\notag
&\xlongequal{(\ref{eq:dual performance process}), (\ref{eq:conditional density expectation})}\tilde{U}_0(\tilde{x})+\frac{\tilde{x}}{\gamma}\essinf_{Z^\mathbb{Q}\in\mathcal{Z}}\mathbb{E}_{t,x,s,y}\left[Z_{t,T}^\mathbb{Q}\log Z_{t,T}^\mathbb{Q}-\gamma Z_{t,T}^\mathbb{Q}\left(C(Y_T)+\frac{1}{2\gamma}A_T\right)\right]\\
&\xlongequal{(\ref{eq:conditional entropy under P-expectation})}\tilde{U}_0(\tilde{x})+\frac{\tilde{x}}{\gamma}\essinf_{Z^\mathbb{Q}\in\mathcal{Z}}\left(\mathcal{H}_t(\mathbb{Q},\mathbb{P})-\gamma \mathbb{E}_{t,x,s,y}^\mathbb{Q}\left[C(Y_T)+\frac{1}{2\gamma}A_T\right]\right).
\end{align}
Thus, the dual forward performance problem amounts to the minimisation of
\begin{equation}
\label{eq:minimal entropy function}
H^C(t,S_t,Y_t):=\essinf_{Z^\mathbb{Q}\in\mathcal{Z}}\left(\mathcal{H}_t(\mathbb{Q},\mathbb{P})-\gamma \mathbb{E}^\mathbb{Q}\left[C(Y_T)+\frac{1}{2\gamma}A_T\,\middle\vert\,\widehat{\mathcal{F}}_t\right]\right),\quad 0\leq t\leq T,
\end{equation}
which will be referred to as the \textit{minimal entropy process}. Given (\ref{eq:conjugate derivative}), the derivative of the dual value function at the optimum $\tilde{x}=\tilde{x}^*,\, x=x^*$ satisfies
\begin{equation*}
\tilde{v}_{\tilde{x}}(t,\tilde{x}^*,s,y)=\tilde{U}_0'(\tilde{x}^*)+\frac{1}{\gamma}H^C(t,s,y)=\frac{1}{\gamma}\left(\log\left(\frac{\tilde{x}^*}{\gamma}\right)+H^C(t,s,y)\right)=-x^*.
\end{equation*}
As the latter equation defines the functional expression of $x^*$ by $\tilde{x}^*$, rearrange it to get the inverse expression
\begin{equation*}
\tilde{x}^*=\gamma\exp\left(-\gamma x^*-H^C(t,s,y)\right).
\end{equation*}
Using this in the bidual relation (\ref{eq:bidual relation}) delivers
\begin{align*}
v(t,x^*,s,y)&=\tilde{v}(t,\tilde{x}^*,s,y)+x^*\tilde{x}^*\\
&=\frac{\tilde{x}^*}{\gamma}\left(\log\left(\frac{\tilde{x}^*}{\gamma}\right)-1\right)+\frac{\tilde{x}^*}{\gamma}H^C(t,s,y)+x^*\tilde{x}^*\\
&=\frac{\tilde{x}^*}{\gamma}\left(\log\left(\frac{\tilde{x}^*}{\gamma}\right)-1+H^C(t,s,y)+\gamma x^*\right)\\
&=-\exp\left(-\gamma x^*-H^C(t,s,y)\right),
\end{align*}
which proves (\ref{eq:dual problem}).
\end{proof}
When writing the primal forward performance problem (\ref{eq:primal problem}) as
\begin{equation*}
v^C(t,X_t,S_t,Y_t)=\esssup_{\theta\in\Theta_t}\mathbb{E}\left[-\exp\left(-\gamma\left(X_T-\left[C(Y_T)+\frac{1}{2\gamma}A_T\right]\right)\right)\,\middle\vert\,\widehat{\mathcal{F}}_t\right],
\end{equation*}
and comparing both this and its dual representation (\ref{eq:dual problem}) to the classical utility case from \cite{M2010}, the claim pay-off term in the forward model is $C(Y_T)+\frac{1}{2\gamma}A_T$ instead of $C(Y_T)$ in the classical model. Thus, the forward case adds value to the terminal value of the option.
\begin{corollary}[Dual representation of the indifference price]
\label{thm:dual representation of the indifference price}
The indifference price process has the entropic representation
\begin{equation}
\label{eq:indifference price entropic representation}
p(t,S_t,Y_t)=-\frac{1}{\gamma}\left(H^C(t,S_t,Y_t)-H^0(t,S_t,Y_t)\right),\quad 0\leq t \leq T.
\end{equation}
\end{corollary}
\begin{proof}
Denote by $H^0$ the minimal entropy process (\ref{eq:minimal entropy function}) with no claim present, or equivalently, $C=0$. The expression (\ref{eq:indifference price entropic representation}) follows directly from Theorem~\ref{thm:dual forward performance problem} and the definition of the indifference price (\ref{eq:indifference price}).
\end{proof}
Next, we give the optimal hedging strategy in terms indifference price derivatives, which is derived analogously to the classical case from Monoyios \cite[Theorem 1]{M2010}.
\begin{theorem}[Optimal hedging strategy in terms of the indifference price]
\label{thm:optimal hedging strategy indifference price}
Suppose the forward indifference price function $p$ is of class $\mathcal{C}^{1,2,2}([0,T]\times[0,\infty)^2)$. Then the optimal hedging strategy for a short position in the claim is given by
\begin{equation}
\label{eq:optimal hedging strategy indifference price}
\theta^H(t,S_t,Y_t)=p_s(t,S_t,Y_t)+\rho\frac{\sigma^YY_t}{\sigma^SS_t}p_y(t,S_t,Y_t),\quad 0\leq t\leq T.
\end{equation}
\end{theorem}
\begin{proof}
By differentiating the entropic representation of the value function given in Theorem~\ref{thm:dual forward performance problem}, we obtain the partial derivatives
\begin{equation}
\label{eq:value function partial derivatives}
v_x=-\gamma v,\quad v_{xx}=\gamma^2v,\quad v_{xy}=\gamma H_yv,\quad v_{xs}=\gamma H_sv,
\end{equation}
for the case $v=v^C,\, H=H^C$ with the claim and the case $v=v^0,\, H=H^0$ without the claim. Apply them to Theorem~\ref{thm:optimal strategy in terms of the value process} to obtain the optimal strategy in terms of derivatives of the minimal entropy process,
\begin{align*}
\theta^*(t,S_t,Y_t)&\xlongequal{(\ref{eq:optimal strategy value process})}-\left(\frac{\widehat{\lambda}_t^S v_x+\sigma^S S_t v_{xs}+\rho\sigma^Y Y_t v_{xy}}{\sigma^S S_t v_{xx}}\right)\\
&\xlongequal{(\ref{eq:value function partial derivatives})}\frac{\widehat{\lambda}_t^S}{\gamma\sigma^SS_t}-\frac{1}{\gamma}\left(H_s(t,S_t,Y_t)+\rho\frac{\sigma^Y Y_t}{\sigma^SS_t}H_y(t,S_t,Y_t)\right)
\end{align*}
for $\theta=\theta^C,\theta^0$.
Finally, consider the optimal hedging strategy formula $\theta^H=\theta^C-\theta^0$ from (\ref{eq:optimal hedging strategy}) and use Corollary~\ref{thm:dual representation of the indifference price} to eliminate the Merton strategy term and obtain the optimal hedging strategy expressed by the indifference price (\ref{eq:optimal hedging strategy indifference price}). The required regularity of the indifference price is shown in \cite[Subsection 3.3]{M2010}.
\end{proof}
\subsection{Forward indifference price valuation}
After we have defined the dual stochastic control problem, we are going to give a more explicit representation formula for the forward indifference price from Corollary\ref{thm:dual representation of the indifference price} following Monoyios \cite[Subsection 3.2]{M2010} and Leung et al. \cite[Subsection 3.2]{LSZ2012}. For this, we will characterise the martingale measure $\mathbb{Q}$ by giving the corresponding density process $Z^\mathbb{Q}$ and then perform a measure change to the basis risk model. The measures $\mathbb{Q}\in\mathcal{M}_{e,f}$ characterised by their densities $Z^\mathbb{Q}$, are parametrised via $\widehat{\mathbb{F}}$-adapted processes $\psi:=(\psi_t)_{0\leq t\leq T}$ satisfying $\int_0^T\psi_u^2\d u<\infty$ $\mathbb{P}$-a.s.\! and $\mathbb{E}[Z_T^\mathbb{Q}]=1$, according to the stochastic exponential
\begin{equation}
\begin{aligned}
\label{eq:measure change}
Z_t^\mathbb{Q} &:= \mathcal{E}\left(-\widehat{\lambda}^S\cdot \widehat{W}^S-\psi\cdot\widehat{W}^\perp\right)_t\\
&=\exp\left(-\int_0^t \widehat{\lambda}_u^S \d{\widehat{W}}_u^S-\int_0^t \psi_u\d{\widehat{W}}_u^\perp-\frac12\int_0^t\left[\left(\widehat{\lambda}_u^S\right)^2 +\psi_u^2\right] \d u\right).
\end{aligned}
\end{equation}
Since $\int_0^t((\widehat{\lambda}_u^S)^2 +\psi_u^2)\d u=A_t+\int_0^t\psi_u^2\d u$ is a strictly positive square-integrable continuous process, \textit{Novikov's condition}
\begin{equation}
\label{eq:Novikov's condition}
\mathbb{E}\left[\exp\left(\frac12\int_0^t\left[\left(\widehat{\lambda}_u^S\right)^2 +\psi_u^2 \right]\d u\right)\right]<\infty
\end{equation}
is fulfilled. Denote with $\Psi$ the set of integrands $\psi$ such that (\ref{eq:Novikov's condition}) is satisfied. Hence, the density process $Z^\mathbb{Q}$ is indeed a $(\mathbb{P},\widehat{\mathbb{F}})$-martingale. Applying \textit{Girsanov's theorem} from \cite[pp.~224--225]{S2004} for a measure change to $\mathbb{Q}$, provides the two-dimensional Brownian motion $(\widehat{W}^{S,\mathbb{Q}},\widehat{W}^{\perp,\mathbb{Q}})$ defined by
\begin{equation}
\label{eq:measure change Brownian motions}
\widehat{W}_t^{S,\mathbb{Q}}=\widehat{W}_t^S + \int_0^t\widehat{\lambda}_u^S\d u, \quad \widehat{W}_t^{\perp,\mathbb{Q}} := \widehat{W}_t^\perp+\int_0^t\psi_u\d u,\quad 0\leq t \leq T.
\end{equation}
The integrand process $\psi$ is commonly referred to as the \textit{volatility risk premium} for the second Brownian motion $\widehat{W}^\perp$. The so-called \textit{minimal martingale measure (MMM)} $\mathbb{Q}^M$ corresponds to the case $\psi=0$. It was originally introduced by F\"ollmer and Schweizer \cite{FS1991} for the risk-minimised (optimal) quadratic hedging strategy in an incomplete market. It alters the MPR of the stock's Brownian motion, but does not change the MPR of Brownian motions orthogonal to those driving the stock. By (\ref{eq:measure change}), it has the Radon-Nikodym derivative process
\begin{equation*}
Z_t^{\mathbb{Q}^M}=\mathcal{E}\left(-\widehat{\lambda}^S\cdot \widehat{W}^S\right)_t=\exp\left(-\int_0^t \widehat{\lambda}_u^S \d{\widehat{W}}_u^S-\frac12\int_0^t\left(\widehat{\lambda}_u^S\right)^2\d u\right).
\end{equation*}
The second equation of (\ref{eq:measure change Brownian motions}) implies that $\widehat{W}_t^\perp$ is also a $(\mathbb{Q}^M,\widehat{\mathbb{F}})$-Brownian motion.
\begin{proposition}[Representation of conditional relative entropy]
\label{thm:representation of conditional relative entropy}
The conditional relative entropy between $\mathbb{Q}$ and $\mathbb{P}$ satisfies
\begin{equation}
\label{eq:conditional relative entropy representation}
\mathcal{H}_t(\mathbb{Q},\mathbb{P})=\frac12\mathbb{E}^\mathbb{Q}\left[\int_t^T\left[\left(\widehat{\lambda}_u^S\right)^2 +\psi_u^2\right]\d u\,\middle\vert\,\widehat{\mathcal{F}}_t\right]<\infty.
\end{equation}
In the full information scenario, the conditional relative entropy simplifies to
\begin{equation}
\label{eq:conditional relative entropy full info}
\mathcal{H}_t(\mathbb{Q},\mathbb{P})=\frac12\left(\lambda^S\right)^2(T-t)+\frac12\mathbb{E}^\mathbb{Q}\left[\int_t^T\psi_u^2\d u\,\middle\vert\,\widehat{\mathcal{F}}_t\right],
\end{equation}
which is minimised by $\psi=0$. Hence, in the full information scenario, the MEMM $\mathbb{Q}^E$ coincides with the MMM $\mathbb{Q}^M$. Using $\mathbb{Q}^M$ as the reference measure, the conditional relative entropy above can be additively decomposed to
\begin{equation*}
\mathcal{H}_t(\mathbb{Q},\mathbb{P})=\mathcal{H}_t(\mathbb{Q},\mathbb{Q}^M)+\mathcal{H}_t(\mathbb{Q}^M,\mathbb{P})
\end{equation*}
with
\begin{equation}
\label{eq:conditional relative entropy under MMM}
\mathcal{H}_t(\mathbb{Q},\mathbb{Q^M})=\frac12 \mathbb{E}^\mathbb{Q}\left[\int_t^T\psi_u^2\d u\,\middle\vert\,\widehat{\mathcal{F}}_t\right],\quad\mathcal{H}_t(\mathbb{Q}^M,\mathbb{P})=\frac12 \mathbb{E}^\mathbb{Q}\left[\int_t^T\left(\widehat{\lambda}_u^S\right)^2\d u\,\middle\vert\,\widehat{\mathcal{F}}_t\right].
\end{equation}
\end{proposition}
\begin{proof}
The conditional relative entropy from Definition (\ref{def:conditional relative entropy}) is
\begin{align*}
\mathcal{H}_t(\mathbb{Q},\mathbb{P})&\xlongequal{(\ref{eq:conditional relative entropy})}\mathbb{E}^\mathbb{Q}\left[\log Z_{t,T}^\mathbb{Q}\,\middle\vert\,\widehat{\mathcal{F}}_t\right]\xlongequal{(\ref{eq:conditional density})}\mathbb{E}^\mathbb{Q}\left[\log\frac{Z_T^\mathbb{Q}}{Z_t^\mathbb{Q}}\,\middle\vert\,\widehat{\mathcal{F}}_t\right]\\
&\xlongequal{(\ref{eq:measure change})}\mathbb{E}^\mathbb{Q}\left[-\int_t^T \widehat{\lambda}_u^S \d{\widehat{W}}_u^S-\int_t^T \psi_u\d{\widehat{W}}_u^\perp-\frac12\int_t^T\left[\left(\widehat{\lambda}_u^S\right)^2 +\psi_u^2\right] \d u\,\middle\vert\,\widehat{\mathcal{F}}_t\right]\\
&\xlongequal{(\ref{eq:measure change Brownian motions})}\mathbb{E}^\mathbb{Q}\left[-\int_t^T \widehat{\lambda}_u^S \d{\widehat{W}}_u^{S,\mathbb{Q}}-\int_t^T \psi_u\d{\widehat{W}}_u^{\perp,\mathbb{Q}}+\frac12\int_t^T\left[\left(\widehat{\lambda}_u^S\right)^2 +\psi_u^2\right] \d u\,\middle\vert\,\widehat{\mathcal{F}}_t\right]\\
&=\frac12\mathbb{E}^\mathbb{Q}\left[\int_t^T\left[\left(\widehat{\lambda}_u^S\right)^2 +\psi_u^2\right]\d u\,\middle\vert\,\widehat{\mathcal{F}}_t\right],
\end{align*}
where the integrability on the right hand side implied by Novikov's condition (\ref{eq:Novikov's condition}) is associated with the finite conditional relative entropy condition. The second assertion (\ref{eq:conditional relative entropy full info}) directly follows from (\ref{eq:conditional relative entropy representation}) for $\widehat{\lambda}_t^S=\lambda^S$. To show (\ref{eq:conditional relative entropy under MMM}), consider the Radon-Nikodym derivative under $\mathbb{Q}^M$,
\begin{equation}
\label{eq:Radon-Nikodym derivative under MMM}
\frac{\d{\mathbb{Q}}}{\d{\mathbb{Q}^M}}=\frac{\d{\mathbb{Q}}}{\d{\mathbb{P}}}\left(\frac{\d{\mathbb{Q^M}}}{\d{\mathbb{P}}}\right)^{-1}=\frac{Z_T^\mathbb{Q}}{Z_T^{\mathbb{Q}^M}}\xlongequal{(\ref{eq:measure change})}\exp\left(-\int_0^T\psi_u\d{\widehat{W}_u^\perp}-\frac12\int_0^T\psi_u^2\d u\right).
\end{equation}
Then again, apply Lemma 5.2.2 from \cite[p.~212]{S2004}, giving the conditional expectation under measure change to compute the density process of $\mathbb{Q}$ with respect to $\mathbb{Q}^M$,
\begin{equation}
\label{eq:density with respect to MMM}
\mathbb{E}^{\mathbb{Q}^M}\left[\frac{Z_T^\mathbb{Q}}{Z_T^{\mathbb{Q}^M}}\,\middle\vert\,\widehat{\mathcal{F}}_t\right]=\frac{1}{Z_t^{\mathbb{Q}^M}}\mathbb{E}\left[Z_T^\mathbb{Q}\,\middle\vert\,\widehat{\mathcal{F}}_t\right]=\frac{Z_t^\mathbb{Q}}{Z_t^{\mathbb{Q}^M}}=:Z_t^{\mathbb{Q},\mathbb{Q}^M},\quad 0\leq t\leq T.
\end{equation}
As in Definition~\ref{def:conditional relative entropy}, for any measure $\mathbb{Q}\in\mathcal{M}_{e,f}$, the conditional density is given by
\begin{equation}
\label{eq:conditional density under MMM}
Z_{t,T}^{\mathbb{Q},\mathbb{Q}^M}:=\frac{Z_T^{\mathbb{Q},\mathbb{Q}^M}}{Z_t^{\mathbb{Q},\mathbb{Q}^M}}\xlongequal{(\ref{eq:density with respect to MMM})}\frac{Z_T^\mathbb{Q}}{Z_T^{\mathbb{Q}^M}}\frac{Z_t^{\mathbb{Q}^M}}{Z_t^\mathbb{Q}}\xlongequal{(\ref{eq:Radon-Nikodym derivative under MMM})}\exp\left(-\int_t^T\psi_u\d{\widehat{W}_u^\perp}-\frac12\int_t^T\psi_u^2\d u\right).
\end{equation}
Then, directly compute the conditional relative entropy over the interval $[t,T]$,
\begin{equation*}
\mathcal{H}_t(\mathbb{Q},\mathbb{Q^M})=\mathbb{E}^\mathbb{Q}\left[\log Z_{t,T}^{\mathbb{Q},\mathbb{Q}^M}\,\middle\vert\,\widehat{\mathcal{F}}_t\right]\xlongequal{(\ref{eq:conditional density under MMM})}\frac12 \mathbb{E}^\mathbb{Q}\left[\int_t^T\psi_u^2\d u\,\middle\vert\,\widehat{\mathcal{F}}_t\right],
\end{equation*}
by using the $(\mathbb{Q}^M,\widehat{\mathbb{F}})$-martingale property of $\int_t^T\psi_u \d{\widehat{W}}_u^\perp$. The conditional relative entropy $\mathcal{H}_t(\mathbb{Q}^M,\mathbb{P})$ is analogously determined given $\widehat{W}^{S,\mathbb{Q}}=\widehat{W}^{S,\mathbb{Q}^M}$.
\end{proof}
Proposition~\ref{thm:representation of conditional relative entropy} implies the generalised additivity formula
\begin{equation}
\label{eq:relative entropy additivity}
{H}_t(\mathbb{Q},\mathbb{P})=\mathcal{H}_t(\mathbb{Q},\widetilde{\mathbb{Q}})+\mathcal{H}_t(\widetilde{\mathbb{Q}},\mathbb{P}),
\end{equation}
for any martingale measure $\widetilde{\mathbb{Q}}\in\mathcal{M}_{e,f}$ (see also Monoyios \cite[p.~902]{M2013}).
\begin{theorem}[Forward indifference price valuation]
\label{thm:forward indifference price valuation}
The forward indifference price is the solution of the stochastic control problem
\begin{equation}
\label{eq:forward indifference price control problem}
p(t,S_t,Y_t)=-\frac{1}{\gamma}\essinf_{\psi\in\Psi}\left(\mathcal{H}_t(\mathbb{Q},\mathbb{Q}^M)-\gamma \mathbb{E}^\mathbb{Q}\left[C(Y_T)\,\middle\vert\,\widehat{\mathcal{F}}_t\right]\right),
\end{equation}
with optimal control
\begin{equation}
\label{eq:optimal control}
\psi^H(t,S_t,Y_t)=-\gamma\sqrt{1-\rho^2}\sigma^YY_t p_y(t,S_t,Y_t),
\end{equation}
and solves the semi-linear partial differential equation of second order
\begin{equation}
\label{eq:p PDE}
p_t+\mathcal{A}_{S,Y}^{\mathbb{Q}^M}p+\frac12\gamma(1-\rho^2)\left(\sigma^Yyp_y\right)^2=0,\quad p(T,s,y)=C(y).
\end{equation}
The marginal performance-based price (marginal forward indifference price) is
\begin{equation}
\label{eq:marginal performance-based price}
p^M(t,S_t,Y_t):=\lim_{\gamma\to 0}p(t,S_t,Y_t)=\mathbb{E}^{\mathbb{Q}^M}\left[C(Y_T)\,\middle\vert\,\widehat{\mathcal{F}}_t\right].
\end{equation}
\end{theorem}
\begin{proof}
Denote with $\Psi$ the set of volatility risk premia $\psi$ such that (\ref{eq:conditional relative entropy representation}) is satisfied. Then $\Psi$ parametrises all the $\mathbb{Q}\in\mathcal{M}_{e,f}$ through the well-defined map induced by (\ref{eq:measure change}). Therefore in control theory $\Psi$ is called the \textit{control set} and $\psi$ a \textit{control}. Using Proposition~\ref{thm:representation of conditional relative entropy}, the minimal entropy process (\ref{eq:minimal entropy function}) can be represented as
\begin{align}
\label{eq:minimal entropy process with claim}
\notag
H^C(t,S_t,Y_t)&=\essinf_{\psi\in\Psi}\mathbb{E}^\mathbb{Q}\!\left[\frac12\int_t^T\!\left[\left(\widehat{\lambda}_u^S\right)^2\!\! +\psi_u^2\right]\d u-\gamma\left[C(Y_T)+\frac{1}{2\gamma}\int_0^T\!\!\left(\widehat{\lambda}_u^S\right)^2\!\d u\right]\,\middle\vert\widehat{\mathcal{F}}_t\,\right]\\
\notag
&=\essinf_{\psi\in\Psi}\mathbb{E}^\mathbb{Q}\!\left[\frac12\int_t^T\!\left[\left(\widehat{\lambda}_u^S\right)^2\!\!\! +\psi_u^2\right]\!\d u-\!\frac12\int_t^T\!\!\!\left(\widehat{\lambda}_u^S\right)^2\!\!\d u-\gamma C(Y_T)- \frac12 A_t\middle\vert\widehat{\mathcal{F}}_t\right]\\
\notag
&=\essinf_{\psi\in\Psi}\mathbb{E}^\mathbb{Q}\left[\frac12\int_t^T\psi_u^2\d u-\gamma C(Y_T)-\frac12 A_t\,\middle\vert\,\widehat{\mathcal{F}}_t\right]\\
&\xlongequal{(\ref{eq:conditional relative entropy under MMM})}-\frac12 A_t+\essinf_{\psi\in\Psi}\left(\mathcal{H}_t(\mathbb{Q},\mathbb{Q}^M)-\gamma \mathbb{E}^\mathbb{Q}\left[C(Y_T)\,\middle\vert\,\widehat{\mathcal{F}}_t\right]\right).
\end{align}
Hence, the dual value process (\ref{eq:dual problem}) becomes
\begin{equation}
\label{eq:forward value process with claim entropic representation}
\!\!\!\!\!v^C(t,X_t,S_t,Y_t)=-\exp\!\left(\!-\gamma X_t\!+\!\frac12 A_t-\essinf_{\psi\in\Psi}\!\left(\mathcal{H}_t(\mathbb{Q},\mathbb{Q}^M)-\gamma \mathbb{E}^\mathbb{Q}\!\left[C(Y_T)\middle\vert\widehat{\mathcal{F}}_t\right]\right)\!\!\right)\!\!.\!\!\!
\end{equation}
Denote by $\psi^C$ the optimal control in (\ref{eq:minimal entropy process with claim}) with the claim. Analogous, let $\psi^0$ be the optimal control in absence of the claim. In the latter case, the relative entropy is minimised by the control $\psi^0=\psi^M=0$ and gives the minimal entropy process
\begin{equation}
\label{eq:minimal entropy process without claim}
H^0(t,S_t,Y_t)=-\frac 12 A_t+\essinf_{\mathbb{Q}\in\mathcal{M}_{e,f}}\mathcal{H}_t(\mathbb{Q},\mathbb{Q}^M)=-\frac12 A_t+\mathcal{H}_t(\mathbb{Q}^M,\mathbb{Q}^M)=-\frac12 A_t,
\end{equation}
so that the value function has the optimal wealth $X_t^0=X_t$ and simplifies to
\begin{equation}
\label{eq:forward value process without claim entropic representation}
v^0(t,X_t,S_t,Y_t)=-\exp\left(-\gamma X_t+\frac12 A_t\right)=U_t(X_t)=U(X_t^0).
\end{equation}
With (\ref{eq:minimal entropy process with claim}), (\ref{eq:minimal entropy process without claim}), the forward indifference price from Corollary~\ref{thm:dual representation of the indifference price} has the expression
\begin{align}
\notag
p(t,S_t,Y_t)&=-\frac{1}{\gamma}\essinf_{\psi\in\Psi}\left(\mathcal{H}_t(\mathbb{Q},\mathbb{Q}^M)-\gamma \mathbb{E}^\mathbb{Q}\left[C(Y_T)\,\middle\vert\,\widehat{\mathcal{F}}_t\right]\right)\\
\label{eq:forward indifference price control problem psi}
&=\esssup_{\psi\in\Psi}\left(\mathbb{E}^\mathbb{Q}\left[C(Y_T)\,\middle\vert\,\widehat{\mathcal{F}}_t\right]-\frac{1}{2\gamma}\int_t^T\psi_u^2(u,S_u,Y_u)\d u\right).
\end{align}
By (\ref{eq:dS, dY partial}), (\ref{eq:measure change Brownian motions}), the asset prices have the $\mathbb{Q}$-dynamics
\begin{equation}
\label{eq:dS, dY measure change}
\begin{aligned}
\d S_t&=\sigma^S S_t\d{\widehat{W}}_t^{S,\mathbb{Q}},\\
\d Y_t&=\sigma^Y Y_t\left[(\widehat{\lambda}_t^Y-\rho\widehat{\lambda}_t^S-\sqrt{1-\rho^2}\psi_t)\d t+\d{\widehat{W}}_t^{Y,\mathbb{Q}}\right],
\end{aligned}
\end{equation}
with the $(\mathbb{Q},\widehat{\mathbb{F}})$-Brownian motion
\begin{equation}
\label{eq:brownian motion correlation decomposition}
\widehat{W}^{Y,\mathbb{Q}}=\rho\widehat{W}^{S,\mathbb{Q}}+\sqrt{1-\rho^2}\widehat{W}^{\perp,\mathbb{Q}}.
\end{equation}
Then, with $\mathcal{A}_{S,Y}^{\mathbb{Q}^M}$ as the generator of $(S,Y)$ under $\mathbb{Q}^M$ (cf. (\ref{eq:HJB supremum v})), the HJB equation for $p$, by (\ref{eq:forward indifference price control problem psi}), is given by
\begin{equation}
\label{eq:HJB minimal entropy function}
p_t+\mathcal{A}_{S,Y}^{\mathbb{Q}^M}p+\max_{\psi\in\Psi}\left[-\frac{1}{2\gamma} \psi^2-\sqrt{1-\rho^2}\sigma^Yy\psi p_y\right]=0,
\end{equation}
with terminal value $p(T,s,y)=C(Y_T)$ and $\psi=\psi(t,s,y)$. Consider the function $f(\psi)=-\frac{1}{2\gamma} \psi^2-\sqrt{1-\rho^2}\sigma^Yy\psi p_y$ and determine its maximum by solving the equation $f'(\psi^C)=0$, which gives the optimal control (\ref{eq:optimal control}), because of $\psi^H=\psi^C-\psi^0=\psi^C$. Substituting this into the HJB equation (\ref{eq:HJB minimal entropy function}) yields the PDE (\ref{eq:p PDE}) for the forward indifference price. The \textit{marginal performance-based indifference price} (\ref{eq:marginal performance-based price}) follows from the Feynman-Kac theorem (cf. Theorem 6.4.1 from \cite[p.~268]{S2004}), when the non-linear term in the PDE (\ref{eq:p PDE}) vanishes for $\gamma\to 0$.
\end{proof}
Le us compare the forward indifference price valuation results of Theorem~\ref{thm:forward indifference price valuation} with the classical theory from \cite{M2010}. The classical minimal entropy process with the claim admits the representation
\begin{equation}
\label{eq:H^C classical}
H^C(t,S_t,Y_t):=\essinf_{\psi\in\Psi}\left(\mathcal{H}_t(\mathbb{Q},\mathbb{P})-\gamma \mathbb{E}^\mathbb{Q}\left[C(Y_T)\,\middle\vert\,\widehat{\mathcal{F}}_t\right]\right),
\end{equation}
with $\mathcal{H}_t(\mathbb{Q},\mathbb{P})$ as in (\ref{eq:conditional relative entropy representation}). Without the claim, the formula turns into
\begin{align}
\notag
H^0(t,S_t,Y_t)&:=\essinf_{\psi\in\Psi}\mathcal{H}_t(\mathbb{Q},\mathbb{P})=\essinf_{\psi\in\Psi}\mathbb{E}^\mathbb{Q}\left[\frac12\int_t^T\left[\left(\widehat{\lambda}_u^S\right)^2 +\psi_u^2\right]\d u\,\middle\vert\,\widehat{\mathcal{F}}_t\right]\\
\label{eq:H^0 classical}
&=\frac12\mathbb{E}^\mathbb{Q}\left[\int_t^T\left[\left(\widehat{\lambda}_u^S\right)^2 +\left(\psi_u^E\right)^2\right]\d u\,\middle\vert\,\widehat{\mathcal{F}}_t\right]=\mathcal{H}_t(\mathbb{Q}^E,\mathbb{P}),
\end{align}
where $\mathbb{Q}^E$ is the MEMM, $\mathcal{H}_t(\mathbb{Q}^E,\mathbb{P})$ the \textit{minimal conditional relative entropy} and $\psi^E=\psi^0$ the \textit{minimal entropy control process}. The value processes with and without the claim are
\begin{equation}
\label{eq:classical value processes}
\begin{aligned}
v^C(t,X_t,S_t,Y_t)&=-\exp\left(-\gamma X_t-H^C(t,S_t,Y_t)\right),\\
v^0(t,X_t,S_t,Y_t)&=-\exp\left(-\gamma X_t-\mathcal{H}_t(\mathbb{Q}^E,\mathbb{P})\right).
\end{aligned}
\end{equation}
Because of the relative entropy additivity $\mathcal{H}_t(\mathbb{Q},\mathbb{P})=\mathcal{H}_t(\mathbb{Q},\mathbb{Q}^E)+\mathcal{H}_t(\mathbb{Q}^E,\mathbb{P})$ from (\ref{eq:relative entropy additivity}), the classical indifference price is the solution of the dual control problem
\begin{align}
\notag
p(t,S_t,Y_t)&=-\frac{1}{\gamma}\essinf_{\psi\in\Psi}\left(\mathbb{E}^\mathbb{Q}\left[\frac12\int_t^T\left[\psi_u^2 -\left(\psi_u^E\right)^2\right]\d u\,\middle\vert\,\widehat{\mathcal{F}}_t\right]-\gamma \mathbb{E}^\mathbb{Q}\left[C(Y_T)\,\middle\vert\,\widehat{\mathcal{F}}_t\right]\right)\\
\notag
&=-\frac{1}{\gamma}\essinf_{\psi\in\Psi}\left(\mathcal{H}_t(\mathbb{Q},\mathbb{Q}^E)-\gamma \mathbb{E}^\mathbb{Q}\left[C(Y_T)\,\middle\vert\,\widehat{\mathcal{F}}_t\right]\right)\\
\label{eq:classical indifference price probabilistic representation}
&=\esssup_{\psi\in\Psi}\left(\mathbb{E}^\mathbb{Q}\left[C(Y_T)\,\middle\vert\,\widehat{\mathcal{F}}_t\right]-\frac{1}{\gamma}\mathcal{H}_t(\mathbb{Q},\mathbb{Q}^E)\right),
\end{align}
which was shown by Monoyios \cite[p.~903]{M2013}. The HJB equation for (\ref{eq:H^C classical}) is
\begin{equation*}
H_t^C+\mathcal{A}_{S,Y}^{\mathbb{Q}^M}H^C+\frac12\left(\widehat{\lambda^S}\right)^2+\min_{\psi\in\Psi}\left[\frac12 \psi^2-\sqrt{1-\rho^2}\sigma^Yy\psi p_y\right]=0,\; H^C(T,s,y)=-\gamma C(y).
\end{equation*}
Solving yields the optimal control $\psi^C=\sqrt{1-\rho^2}\sigma^Y y H_y^C$ and further the PDE
\begin{equation}
\label{eq:H^C PDE classical}
H_t^C+\mathcal{A}_{S,Y}^{\mathbb{Q}^M}H^C+\frac12\left(\widehat{\lambda^S}\right)^2-\frac12(1-\rho^2)\left(\sigma^YyH_y^C\right)^2=0,\; H^C(T,s,y)=-\gamma C(y).
\end{equation}
Without the claim, the same approach returns $\psi^E=\sqrt{1-\rho^2}\sigma^YyH_y^0$ and an analogous PDE for $H^0$ with $H^0(T,s,y)=0$.
Hence, the optimal (hedging) control is
\begin{equation}
\label{eq:classical optimal control}
\psi^H=\psi^C-\psi^E=-\gamma\sqrt{1-\rho^2}\sigma^Yyp_y.
\end{equation}
Subtract the PDEs (\ref{eq:H^C PDE classical}) for $H^C$ and $H^0$ according to (\ref{eq:indifference price entropic representation}) and apply the identities $-\gamma p_y=H_y^C-H_y^0$ and $\frac12\gamma p_y^2-p_yH_y^0=\frac{1}{2\gamma}\left(\left(H_y^C\right)^2-\left(H_y^0\right)^2\right)$ to obtain the PDE for $p$,
\begin{equation*}
p_t+\mathcal{A}_{S,Y}^{\mathbb{Q}^M}p+\frac12\gamma(1-\rho^2)\left(\sigma^Yyp_y\right)^2-\sqrt{1-\rho^2}\sigma^Yyp_y\psi^E=0,\quad p(T,s,y)=C(y).
\end{equation*}
Expressed by the differential operator $\mathcal{A}_{S,Y}^{\mathbb{Q}^E}$, the indifference price PDE has the form
\begin{equation}
\label{eq:p PDE classical}
p_t+\mathcal{A}_{S,Y}^{\mathbb{Q}^E}p+\frac12\gamma(1-\rho^2)\left(\sigma^Yyp_y\right)^2=0,\quad p(T,s,y)=C(y),
\end{equation}
and the marginal utility-based price process is $p^M(t,S_t,Y_t)=\mathbb{E}^{\mathbb{Q}^E}\left[C(Y_T)\,\middle\vert\,\widehat{\mathcal{F}}_t\right]$.\par
In comparison to the well known classical case (\ref{eq:classical indifference price probabilistic representation}), (\ref{eq:p PDE classical}), the relative entropy term in the forward indifference problem of Theorem \ref{thm:forward indifference price valuation} is computed with respect to $\mathbb{Q}^M$ instead of $\mathbb{Q}^E$. The reason is that through the suitable choice of the mean-variance trade-off process $A_t=\int_0^t(\widehat{\lambda}_u^S)^2 \d u$, the conditional relative entropy $\mathcal{H}_t(\mathbb{Q},\mathbb{P})$ under the physical measure $\mathbb{P}$ in the minimal entropy function $H^C$ is transformed into the relative entropy $\mathcal{H}_t(\mathbb{Q},\mathbb{Q}^M)$ under the MMM by eliminating the entropy term $\mathcal{H}_t(\mathbb{Q}^M,\mathbb{P})$. These representations for American versions of the indifference prices were derived by Leung and Sircar \cite{LS2009} for classical utility and by Leung, Sircar and Zariphopoulou \cite{LSZ2012} for forward utility. The optimal hedging control $\psi^H$ in (\ref{eq:optimal control}) and (\ref{eq:classical optimal control}) have the same representation formula. The classical value processes (\ref{eq:classical value processes}) have, in comparison to the value processes (\ref{eq:forward value process with claim entropic representation}), (\ref{eq:forward value process without claim entropic representation}) of the forward model, no trade-off term, which only shows up in the forward performance process. In the forward problem, the optimal control without the claim vanishes, i. e. $\psi^0=\psi^M=0$, but in the classical model $\psi^0=\psi^E$ is not in general zero. This difference only occurs in the partial information scenario when $z_0^S>z_0^Y$ from (\ref{eq:MPR dependencies}). In the case $z_0^S\leq z_0^Y$, the stock's MPR $\widehat{\lambda}^S$ loses the dependence on the non-traded asset price $Y$, so that, after (\ref{eq:dS, dY measure change}), $Y$ is directly affected by $\psi$ and therefore $\widehat{\lambda}^S$ becomes independent of $\psi$. Thus, the drift term is excluded from the minimal entropy process (\ref{eq:H^0 classical}). If the full information scenario is applied, then the classical problem takes the MMM $\mathbb{Q}^E=\mathbb{Q}^M$ because of $\psi^E=0$ and the trade-off term with the drift $\lambda^S$ under the background filtration $\mathbb{F}$ is again excluded from (\ref{eq:H^0 classical}). In conclusion, an appropriate selection of the initial variance estimations $z_0^S, z_0^Y$ with $z_0^S\leq z_0^Y$ in the Kalman-Bucy filter under partial information from Proposition~\ref{thm:model under partial information}, ensures the same pricing in the forward and classical model.
\begin{remark}[Distortion solution of the indifference price]
\label{rem:distortion}
Monoyios \cite{M2006} proved, that if the asset prices follow SDEs with stochastic volatilities of the form
\begin{equation*}
\d S_t = \sigma(Y_t)S_t(\lambda(Y_t)\d t+\d W_t),\quad \d Y_t = a(Y_t)\d t + b(Y_t)\left(\rho\d W_t+\sqrt{1-\rho^2}\d W_t^\perp\right),
\end{equation*}
then the \textit{distortion transformation} from Musiela and Zariphopoulou \cite[pp.~222--223]{MZ2004} leads to the solution of the classical indifference price PDE (\ref{eq:p PDE classical}),
\begin{equation}
\label{eq:indifference price probabilistic representation}
p(t,y)=\frac{1}{\gamma(1-\rho^2)}\log\mathbb{E}^{\mathbb{Q}^E}_{t,y}\left[\exp\left(\gamma(1-\rho^2)C(Y_T)\right)\right],
\end{equation}
given by Oberman and Zariphopoulou \cite{OZ2003}. Leung et al. \cite[pp.~16--17]{LSZ2012} gave the solution in the forward performance model using the appropriate measure $\mathbb{Q}^M$. The indifference price solution under full information looks like (\ref{eq:indifference price probabilistic representation}) with $\mathbb{Q}^M$ and can be found in Henderson and Hobson \cite[p.~344]{HH2002}, Musiela and Zariphopoulou \cite[p.~233]{MZ2004} and Monoyios \cite[p.~248]{M2004}.
\end{remark}
\subsection{Residual risk}
In Subsection~\ref{subsec:Perfect hedging in a complete market} we have discussed the complete market case, where perfect hedging of the stock $S$ by the derivative $C(Y)$ is possible, when the underlying non-traded asset price $Y$ perfectly correlates with the stock price. One says, that the stock is \textit{replicated} by the derivative on the non-tradeable underlying. If the asset prices are not perfectly correlated, as in Subsection~\ref{subsec:Forward performance problem in an incomplete market}, then the hedge becomes imperfect and a non-hedgeable \textit{basis risk (hedging error)} remains. The \textit{basis} is the difference between the price of the asset to be hedged and the price of the hedging instrument, which is why \textit{residual risk} is commonly also referred to as basis risk. Hedging a financial instrument by another correlated instrument is called \textit{cross hedging}.\par
The reasons why an instrument is practically non-tradeable are diverse. For instance, its liquidity (trading volume) in the market could be very low or the spreads and commission fees very high, so that trading is not economical. Or it is simply not tradeable, because the instrument is an abstract synthetic product, like an index. There are many examples in the commodities and \textit{OTC (over-the-counter)} and derivatives markets with exotic products like weather and insurance indices or credit default derivatives. Ankirchner and Imkeller \cite{AI2001} introduced a typical example for a cross hedge, where an airline company wants to manage kerosene price risk. Since there is no liquid kerosene futures market, the airline company may fall back on futures on less refined oil, such as crude oil futures, for hedging its kerosene risk. This is a reasonable approach, if the price evolutions of kerosene and crude oil are highly correlated. Ankirchner, Imkeller and Popier \cite{AIP2008} dealt with optimal cross hedging strategies for insurance related derivatives. Other papers dealing with cross hedging including practical examples are Ankirchner et al. \cite{ADHP2010},  \cite{AID2010}.
\begin{definition}[Residual risk process]
\label{def:residual risk}
Suppose, the investor shorts the claim $C(Y)$ at time $t=0$ for the price $p(0,S_0,Y_0)$. To hedge this position over $[0,T]$, the optimal hedging strategy $\theta^H$ is used. His overall portfolio value is given by the \textit{residual risk process} $\varrho:=(\varrho_t)_{0\leq t\leq T}$, defined by
\begin{equation}
\label{eq:residual risk}
\varrho_t=X_t-p(t,S_t,Y_t),
\end{equation}
with initial and terminal values $\varrho_0=0,\; \varrho_T=X_T-C(Y_T)$ and the forward indifference price $p$. The terminal residual risk is the terminal portfolio value that appeared in the forward performance problem (\ref{eq:value process C}). The stock's position value process is given by
\begin{equation}
\label{eq:stock's position value}
\d X_t=\theta_t^H\d S_t+r(X_t-\theta_tS_t)\d t=\theta_t^H\d S_t,\quad X_0=p(0,S_0,Y_0)
\end{equation}
and the riskless interest rate $r=0$.
\end{definition}
\begin{proposition}[Residual risk process]
\label{thm:residual risk SDE}
The residual risk process of the forward performance-based model under the partial information scenario solves the SDE
\begin{equation}
\label{eq:residual risk SDE}
\d\varrho_t=\frac12\gamma(1-\rho^2)\left(\sigma^YY_tp_y(t,S_t,Y_t)\right)^2\d t-\sqrt{1-\rho^2}\sigma^Y Y_tp_y(t,S_t,Y_t)\d{\widehat{W}}_t^\perp.
\end{equation}
\end{proposition}
\begin{proof}
By Definition~\ref{def:residual risk}, the residual risk has the differential expression
\begin{equation*}
\d\varrho_t\xlongequal{(\ref{eq:residual risk})}\d X_t-\d p(t,S_t,Y_t)\xlongequal{(\ref{eq:stock's position value})}\theta_t^H\d S_t-\d p(t,S_t,Y_t).
\end{equation*}
Using Theorem~\ref{thm:optimal hedging strategy indifference price}, Theorem~\ref{thm:forward indifference price valuation} and It\^{o}'s lemma, we obtain the SDE
\begin{align*}
\d\varrho_t\xlongequal{(\ref{eq:optimal hedging strategy indifference price})}&\left(p_s+\rho\frac{\sigma^YY_t}{\sigma^SS_t}p_y\right)\d S_t\\
&-\left(p_t\d t+p_s\d S_t+p_y\d Y_t+\frac12\Big(p_{ss}\d\langle S\rangle_t+p_{yy}\d\langle Y\rangle_t+p_{sy}\d\langle S,Y\rangle_t\Big)\right)\\
\xlongequal{(\ref{eq:dS, dY measure change})}&\rho\sigma^YY_tp_y\d{\widehat{W}}_t^{S,\mathbb{Q}}-\left(p_t+\mathcal{A}_{S,Y}^{\mathbb{Q}^M}p\right)\d t-\sigma^Y Y_tp_y\d{\widehat{W}}_t^{Y,\mathbb{Q}^M}\\
\xlongequal{(\ref{eq:brownian motion correlation decomposition})}&
-\left(p_t+\mathcal{A}_{S,Y}^{\mathbb{Q}^M}p\right)\d t-\sqrt{1-\rho^2}\sigma^Y Y_tp_y\d{\widehat{W}}_t^\perp\\
\xlongequal{(\ref{eq:p PDE})}&
\frac12\gamma(1-\rho^2)\left(\sigma^YY_tp_y\right)^2\d t-\sqrt{1-\rho^2}\sigma^Y Y_tp_y\d{\widehat{W}}_t^\perp,
\end{align*}
for the residual risk process $\varrho$, where $\widehat{W}^{Y,\mathbb{Q}^M}=\rho\widehat{W}^{S,\mathbb{Q}}+\sqrt{1-\rho^2}\,\widehat{W}^\perp$.
\end{proof}
The version of the residual risk SDE (\ref{eq:residual risk SDE}) under full information and classical utility is in \cite{MZ2004} and \cite{M2007}. The residual risk evolution is expressed by a forward indifference price-based drift term containing the coefficient $1-\rho^2$ together with a stochastic term including the orthogonal Brownian motion $\widehat{W}^\perp$ and the scale parameter $\sqrt{1-\rho^2}$. In the complete market scenario $\abs{\rho}=1$, the residual risk $\varrho$ vanishes and no hedging error remains. But even if the absolute correlation is very high, meaning close to $1$, then a considerably high residual risk remains. If the correlation was high as $\rho=98\%$, the scale parameter of the drift term would be $1-\rho^2\approx 4\%$ and of the stochastic term even $\sqrt{1-\rho^2}\approx20\%$. This means, that the standard deviation of the basis would still represent about $20\%$ of the total risk induced by the stochastic term. If the correlation is almost perfect, a small change leads to significant change in the percentage of the basis risk relative to total risk. Conversely, in the virtually uncorrelated case, a small change in the correlation leads to essentially no change in the percentage of basis risk relative to total risk (see Figure \ref{fig:correlation}). This fact complicates effective hedging, since asset correlations in real markets do not tend to perfectly correlate. Boucrelle et al.~\cite{BLS1996} analysed the U.S. stock and bond markets and figured out that correlations fluctuate widely over time. In addition, correlations increase in periods of high market volatility. Sandoval Junior and De Paula Franca~\cite{SD2012} have come to a similar conclusion with more recent data. Using eigenvalues and eigenvectors of correlation matrices of main financial market indices, they have shown on the basis of price data from the largest crises of the last decades, that high volatility of markets is directly linked with strong correlations between them. When instruments like \textit{Exchange Traded Funds (ETF)} or derivatives try to replicate another (untradeable) instrument like an index, then the measured correlation is not always perfect as desired and a so-called \textit{tracking-error} arises. This was shown by Jorion~\cite{J2003}, Aber, Can and Li~\cite{ACL2009} and Lobe, R\"oder and Schmidhammer~\cite{LRS2011} in various settings. Models for dynamic conditional correlation were studied by Engle~\cite{E2002} and Franses and Hafner~\cite{FH2009}.
\vspace{-1em}
\begin{figure}[ht]
\caption{Effect of the correlation coefficient on the residual risk}
\label{fig:correlation}
\centering
\includegraphics[scale=0.64]{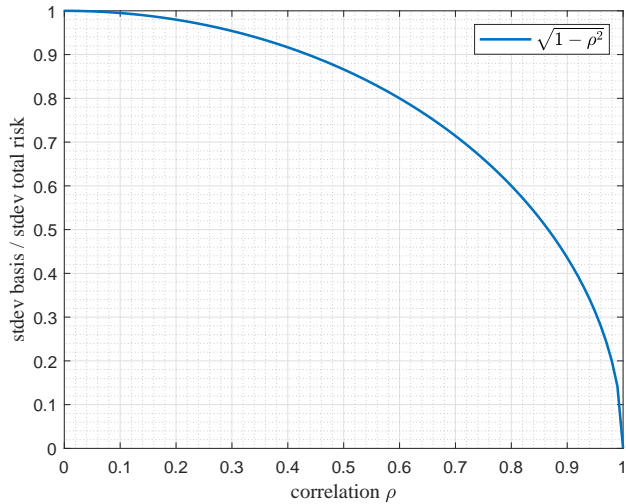}
\end{figure}
\vspace{-1em}
\begin{remark}[Effect of correlation on diversification]
As the instability of the residual risk $\varrho(\rho)$ for absolute correlations close to $1$ makes hedging more difficult, a similar effect can be found in classical portfolio theory from Markowitz~\cite{M1952}. For diversification purpose, consider a portfolio $P=\vartheta S+(1-\vartheta)Y$ containing two assets $S,Y$ with relative weights $\vartheta,1-\vartheta\in[0,1]$. The standard deviation of $P$ is then
\begin{equation*}
\sigma^P=\sqrt{\vartheta^2\left(\sigma^S\right)^2+(1-\vartheta)^2\left(\sigma^Y\right)^2+2\rho\vartheta(1-\vartheta)\sigma^S\sigma^Y}\leq\vartheta\sigma^S+(1-\vartheta)\sigma^Y.
\end{equation*}
The inequation follows from the evaluation of the binomial $(\vartheta\sigma^S+(1-\vartheta)\sigma^Y)^2$ and delivers an equation when the assets are perfectly correlated. Since $\sqrt{\rho}$ has a low slope when $\rho$ is close to $1$, a decrease of $\sigma^P$ and therefore a diversification effect only occurs, when $\rho$ rapidly falls towards $0$. In the case of negative correlation this effect reverses. A small negative correlation may significantly lower the portfolio volatility. Sharpe~\cite{S1963}, \cite{S1964} and Lintner~\cite{L1965} deal also with classical portfolio theory.
\end{remark}
\subsection{Pay-off decompositions and asymptotic expansions}
In this subsection, we shall obtain pay-off decompositions of the claim followed by an asymptotic representation for the forward indifference price valid for small values of  risk aversion. We pursue an approach as for classical utility from Monoyios \cite{M2010}.\par
Recall from (\ref{eq:measure change Brownian motions}) and (\ref{eq:dS, dY measure change}) the asset price dynamics under $\mathbb{Q}^M$,
\begin{align*}
\d S_t&=\sigma^S S_t\d{\widehat{W}}_t^{S,\mathbb{Q}},\\
\d Y_t&=\sigma^Y Y_t\left[(\widehat{\lambda}_t^Y-\rho\widehat{\lambda}_t^S)\d t+\d{\widehat{W}}_t^{Y,\mathbb{Q}^M}\right],
\end{align*}
with $\widehat{W}^{S,\mathbb{Q}^M}=\widehat{W}^{S,\mathbb{Q}},\;\widehat{W}^{\perp,\mathbb{Q}^M}=\widehat{W}^\perp$ and $\widehat{W}^{Y,\mathbb{Q}^M}=\rho\widehat{W}^{S,\mathbb{Q}}+\sqrt{1-\rho^2}\,\widehat{W}^\perp$.
\begin{definition}[Preference-adjusted exponential of the residual risk]
\label{def:preference-adjusted exponential residual risk}
The process
\begin{equation}
\label{eq:preference-adjusted exponential residual risk}
L:=(L_t)_{0\leq t\leq T},\quad 
L_t:=-\exp\left(-\gamma \varrho_t\right),\quad L_0=-1,
\end{equation}
is called \textit{preference-adjusted exponential of the residual risk (PAERR)}.
\end{definition}
\begin{corollary}[Preference-adjusted exponential of the residual risk]
\label{thm:PAERR}
The PAERR process $L$ from Definition \ref{def:preference-adjusted exponential residual risk} is a $(\mathbb{P}, \widehat{\mathbb{F}})$-martingale with dynamics
\begin{equation}
\label{eq:dL}
\d L_t = \sqrt{1-\rho^2}\sigma^Y Y_tp_y(t,S_t,Y_t)\d{\widehat{W}}_t^\perp.
\end{equation}
\end{corollary}
\begin{proof}
By Proposition \ref{thm:residual risk SDE} and It\^{o}'s lemma it is
\begin{equation*}
\d L_t\xlongequal{(\ref{eq:preference-adjusted exponential residual risk})}-\gamma L_t\d\varrho_t+\frac12 \gamma^2L_t\d\langle\varrho\rangle_t\xlongequal{(\ref{eq:residual risk SDE})}\sqrt{1-\rho^2}\sigma^Y Y_tp_y(t,S_t,Y_t)\d{\widehat{W}}_t^\perp.
\end{equation*}
The martingale property follows, because the orthogonal Brownian motion $\widehat{W}^\perp$ is a martingale under both measures $\mathbb{Q}^M$ and $\mathbb{P}$.
\end{proof}
Corollary \ref{thm:PAERR} is similar to Proposition 6 of \cite[p.~237]{MZ2004} under full information and classical utility, but with the forward indifference price depending on $(S,Y)$ rather than the single variable $Y$ due to the partial information scenario. Under classical utility and partial information as in \cite[Subsection 4.1]{M2010}, the dynamics (\ref{eq:dL}) is in general a $(\mathbb{Q}^E,\widehat{\mathbb{F}})$-martingale. Remark, that therein the process $L$ starts with $L_0=0$ rather than $L_0=-1$. Since $L$ is a martingale, the classical exponential utility of the residual risk is $\mathbb{E}[U_0(\varrho_t)]=\mathbb{E}[-\exp(\gamma\varrho_t)]=\mathbb{E}[L_t]=L_0=-1$ and therefore remains constant, whereas the exponential forward utility of the residual risk,
\begin{equation*}
\mathbb{E}[U_t(\varrho_t)]=\mathbb{E}\left[-\exp\left(-\gamma\varrho_t+\frac12\int_0^t\left(\widehat{\lambda}_u^S\right)\d u\right)\right]=-\mathbb{E}\left[\exp\left(\frac12\int_0^t\left(\widehat{\lambda}_u^S\right)^2\d u\right)\right],
\end{equation*}
decreases over time.
\begin{corollary}[Pay-off decomposition]
\label{thm:pay-off decomposition}
The claim pay-off decomposes into
\begin{equation}
\label{eq:pay-off decomposition}
C(Y_T)=p(t,S_t,Y_t)+\int_t^T\theta_u^H\d S_u + L_T-L_t +\frac12\gamma \left(\langle L\rangle_T-\langle L\rangle_t\right),\quad 0\leq t \leq T,
\end{equation}
where $\theta^H$ is the optimal hedging strategy for the claim, given in Theorem \ref{thm:optimal hedging strategy indifference price}.
\end{corollary}
\begin{proof}
By Proposition \ref{thm:residual risk SDE} and Corollary \ref{thm:PAERR}, the differential of the forward indifference price is
\begin{align*}
\d p(t,S_t,Y_t)&\xlongequal{(\ref{eq:residual risk SDE})}-\frac12\gamma(1-\rho^2)\left(\sigma^YY_tp_y\right)^2\d t+\sqrt{1-\rho^2}\sigma^Y Y_tp_y\d{\widehat{W}}_t^\perp+\theta_t^H\d S_t\\
&\xlongequal{(\ref{eq:dL})}-\frac12\gamma\d \langle L\rangle_t+\d L_t +\theta_t^H\d S_t.
\end{align*}
Integration from $t$ to $T$ delivers the pay-off decomposition (\ref{eq:pay-off decomposition}).
\end{proof}
The classical version under the full information scenario of Corollary \ref{thm:pay-off decomposition} is Theorem 7 of \cite[p.~238]{MZ2004}. Under the partial information classical model of \cite[Lemma~1]{M2010}, the pay-off decomposition (\ref{eq:pay-off decomposition}) is measured under $\mathbb{Q}^E$, whereas the forward version always takes $\mathbb{Q}^M$. Pay-off decomposition is the suitable term, because $L$ is a $\mathbb{Q}^M$-martingale with respect to $\widehat{W}^\perp$, which is strongly orthogonal to the $\mathbb{Q}^M$-martingale $X_T-X_t=\int_t^T\theta_u^H\d S_u$, that in turn, is defined as a stochastic integral with respect to the $\mathbb{Q}^M$-Brownian motion $\widehat{W}^{S,\mathbb{Q}}$ induced by $S$. Mania and Schweizer \cite[pp.~2129--2130]{MS2005} obtained an analogous pay-off decomposition in a more general \textit{backward SDE model} under the classical framework.
\begin{definition}[Marginal preference-adjusted exponential of the residual risk]
\label{def:MPAERR}
Define the \textit{marginal preference-adjusted exponential of the residual risk (MPAERR)} by the process $L^M:=(L_t)_{0\leq t\leq T},\;L_t^M:=\lim_{\gamma\to 0}L_t$. With the marginal performance-based price $p^M$ from (\ref{eq:marginal performance-based price}) the evolution is $\d L_t^M=\sqrt{1-\rho^2}\sigma^Y Y_tp_y^M(t,S_t,Y_t)\d{\widehat{W}}_t^\perp$. The MPAERR is also a $\mathbb{Q}^M$-martingale.
\end{definition}
\begin{corollary}[F\"ollmer-Schweizer-Sondermann pay-off decomposition]
The claim pay-off admits the decomposition
\begin{equation}
\label{eq:Foellmer-Schweizer-Sondermann pay-off decomposition}
C(Y_T)=p^M(t,S_t,Y_t)+\int_t^T\theta_u^M\d S_u+L_T^M-L_t^M,\quad 0\leq t\leq T,
\end{equation}
where $p^M$ is the marginal performance-based price (\ref{eq:marginal performance-based price}) and $\theta^M$ the optimal hedging strategy (\ref{eq:optimal hedging strategy indifference price}) with $p^M$ in place of $p$.
\end{corollary}
\begin{proof}
Equation (\ref{eq:Foellmer-Schweizer-Sondermann pay-off decomposition}) is the \textit{F\"ollmer-Schweizer-Sondermann pay-off decomposition} \cite{FS1991}, \cite{FSS1986} under $\mathbb{Q}^M$ in our model and is immediately implied by Corollary~\ref{thm:pay-off decomposition} and Definition~\ref{def:MPAERR} as $\gamma\to 0$ (cf. \cite[Corollary~1]{M2010} for the classical model).
\end{proof}
\begin{corollary}[Forward indifference price representation]
\label{thm:forward indifference price representation}
The forward indifference price admits the representation
\begin{equation}
\label{eq:forward indifference price representation}
p(t,S_t,Y_t)=p^M(t,S_t,Y_t)+\frac12\gamma\mathbb{E}^{\mathbb{Q}^M}\left[\langle L\rangle_T-\langle L\rangle_t\,\middle\vert\,\widehat{\mathcal{F}}_t\right].
\end{equation}
\end{corollary}
\begin{proof}
Applying the conditional $\mathbb{Q}^M$-expectation given $\widehat{\mathcal{F}}_t$ on the pay-off decomposition (\ref{eq:pay-off decomposition}) eliminates $\int_t^T\theta_u^H\d S_u + L_T-L_t$, due to the martingale property. By the marginal performance-based price formula (\ref{eq:marginal performance-based price}) the representation (\ref{eq:forward indifference price representation}) follows.
\end{proof}
The classical version of Corollary \ref{thm:forward indifference price representation} under $\mathbb{Q}^E$ is dealt in \cite[Corollary~2]{M2010}. Again, in the forward performance framework, the measure $\mathbb{Q}^M$ is used.
\begin{proposition}[Asymptotic expansion of the forward indifference price]
\label{thm:asymptotic expansion of the forward indifference price}
The forward indifference price has the asymptotic representation
\begin{equation}
\label{eq:asymptotic expansion of the forward indifference price}
p(t,S_t,Y_t)=p^M+\frac12\gamma\left(\operatorname{Var}^{\mathbb{Q}^M}\!\!\left[C(Y_T)\,\middle\vert\,\widehat{\mathcal{F}}_t\right]-\mathbb{E}^{\mathbb{Q}^M}\!\!\left[\langle X^M\rangle_{t,T}\,\middle\vert\,\widehat{\mathcal{F}}_t\right]\right)+\mathcal{O}(\gamma^2),
\end{equation}
where $X_{t,T}^M:=X^M_T-X_t^M:=\int_t^T\theta_u^M\d S_u$ denotes the profit and loss of the wealth from $t$ to $T$ under the marginal hedging strategy and $\langle X^M\rangle_{t,T}$ its covariation.
\end{proposition}
\begin{proof}
We make the same ansatz as in the classical version from \cite[Theorem~2]{M2010} and write the asymptotic expansion
\begin{equation}
\label{eq:asymptotic expansion ansatz}
p(t,S_t,Y_t)=p^M(t,S_t,Y_t)+\gamma g(t,S_t,Y_t)+\mathcal{O}(\gamma^2),
\end{equation}
with an appropriate process $g:=(g_t)_{0\leq t\leq T}$. By Corollary \ref{thm:forward indifference price representation} and Corollary \ref{thm:PAERR} it follows
\begin{align*}
\gamma g(t,S_t,Y_t)+\mathcal{O}\left(\gamma^2\right)&\xlongequal{(\ref{eq:forward indifference price representation})}\frac12\gamma\mathbb{E}^{\mathbb{Q}^M}\left[\langle L\rangle_T-\langle L\rangle_t\,\middle\vert\,\widehat{\mathcal{F}}_t\right]\\
&\xlongequal{(\ref{eq:dL})}\frac12\gamma(1-\rho^2)\left(\sigma^Y\right)^2\mathbb{E}^{\mathbb{Q}^M}\left[\int_t^T Y_u^2p_y^2(u,S_u,Y_u)\d u\,\middle\vert\,\widehat{\mathcal{F}}_t\right]\\
&\xlongequal{(\ref{eq:asymptotic expansion ansatz})}\frac12\gamma(1-\rho^2)\left(\sigma^Y\right)^2\mathbb{E}^{\mathbb{Q}^M}\!\!\left[\int_t^T\!\! Y_u^2\left(p_y^M\!+\!\gamma g_y\!+\!\mathcal{O}(\gamma^2)\right)^2\!\d u\,\middle\vert\,\widehat{\mathcal{F}}_t\right]\\
&=\frac12\gamma(1-\rho^2)\left(\sigma^Y\right)^2\mathbb{E}^{\mathbb{Q}^M}\!\!\left[\int_t^T\!\!\! \left(Y_u p_y^M(u,S_u,Y_u)\right)^2\!\d u\,\middle\vert\,\widehat{\mathcal{F}}_t\right]\!+\!\mathcal{O}(\gamma^2)
\end{align*}
and further leads to the solution
\begin{equation}
\label{eq:asymptotic expansion first order term}
g(t,S_t,Y_t)=\frac12\mathbb{E}^{\mathbb{Q}^M}\left[\langle L^M\rangle_T-\langle L^M\rangle_t\,\middle\vert\,\widehat{\mathcal{F}}_t\right].
\end{equation}
Inserting (\ref{eq:asymptotic expansion first order term}) into (\ref{eq:asymptotic expansion ansatz}) gives the asymptotic expansion of the indifference price
\begin{equation}
\label{eq:asymptotic expansion MPAERR}
p(t,S_t,Y_t)=p^M(t,S_t,Y_t)+\frac12\gamma\mathbb{E}^{\mathbb{Q}^M}\left[\langle L^M\rangle_T-\langle L^M\rangle_t\,\middle\vert\,\widehat{\mathcal{F}}_t\right]+\mathcal{O}(\gamma^2).
\end{equation}
Notice, that by switching from the PAERR $L$ in (\ref{eq:forward indifference price representation}) to the MPAERR $L^M$ in (\ref{eq:asymptotic expansion MPAERR}), an expansion term of order $\mathcal{O}(\gamma^2)$ is added to the indifference price representation. The F\"ollmer-Schweizer-Sondermann decomposition (\ref{eq:Foellmer-Schweizer-Sondermann pay-off decomposition}) implies the pay-off variance 
\begin{align}
\notag
\operatorname{Var}^{\mathbb{Q}^M}\left[C(Y_T)\,\middle\vert\,\widehat{\mathcal{F}}_t\right]&=\mathbb{E}^{\mathbb{Q}^M}\left[\left(C(Y_T)-p^M(t,S_t,Y_t)\right)^2\,\middle\vert\,\widehat{\mathcal{F}}_t\right]\\
\notag
&=\mathbb{E}^{\mathbb{Q}^M}\left[\left(\int_t^T\theta_u^M\d S_u+L_T^M-L_t^M\right)^2\,\middle\vert\,\widehat{\mathcal{F}}_t^2\right]\\
\label{eq:pay-off variance}
&=\mathbb{E}^{\mathbb{Q}^M}\left[\langle X^M\rangle_T-\langle X^M\rangle_t+\langle L^M\rangle_T-\langle L^M\rangle_t\,\middle\vert\,\widehat{\mathcal{F}}_t\right],
\end{align}
because $L$ and $X$ are orthogonal $\mathbb{Q}^M$-martingales. Inserting (\ref{eq:pay-off variance}) after a rearrangement into (\ref{eq:asymptotic expansion MPAERR}) gives the asymptotic expansion (\ref{eq:asymptotic expansion of the forward indifference price}).
\end{proof}
\section{Exponential forward valuation and hedging of American options under partial information}
\label{sec:american option}
Early exercise claims arise often in situations in which a certain project is undertaken or abandoned (Smith and Nau \cite{SN1995}, Smith and McCardle \cite{SM1998}), executives exercise their employee stock options (Aboody \cite{A1996}, Huddart \cite{H1994}), household owners prepay their mortgages or sell their property (Hall~\cite{H1985}, Kau and Keenan \cite{KK1995}, Schwartz and Torous \cite{ST1993}). Allowing early exercise gives rise to stochastic control problems with stopping times. Early exercise options were priced for the first time by Davis and Zariphopoulou \cite{DZ1995} in the setting, where the option's underlying asset is traded but with proportional transaction costs. Karatzas and Wang \cite{KW2000} studied utility maximisation problems of mixed optimal stopping and control type in complete markets, which can be solved by reduction to a family of related pure optimal stopping problems. Oberman and Zariphopoulou \cite{OZ2003} introduced a utility-based methodology for the valuation of early exercise contracts in incomplete markets. Henderson and Hobson \cite{HH2007} considered the case of infinite time horizon, where the problem is expressed with respect to \textit{horizon-unbiased utility functions}, a class of utility functions satisfying certain\linebreak consistency conditions over time, which are nothing less than forward utilitie. Leung and Sircar \cite{LS2009} studied problems of hedging American options with exponential utility within a general incomplete market model. In Leung, Sircar and Zariphopoulou \cite{LSZ2012} this theory was expanded to the forward performance framework.\par
In this section, we apply the forward performance model under the partial information scenario from Proposition~\ref{thm:model under partial information} to American options to derive hedging and valuation results comparable to the European counterparts of Section~\ref{sec:european option}.
\subsection{Optimal control and stopping problem}
Suppose $C$ is now an \textit{early exercise claim (American option)} written on the non-traded asset $Y$. The investor sets up a hedging portfolio consisting of a long position in the stock $S$ and a short position in the option $C$ as in the European scenario of Section~\ref{sec:european option}. 
\begin{definition}[Admissible exercise times]
The collection of \textit{admissible exercise times} is the set $\mathcal{T}$ of stopping times $\tau$ with respect to the observation filtration $\widehat{\mathbb{F}}=(\widehat{\mathcal{F}}_t)_{0\leq t\leq T}$ that take values in $[0,T]$. For $0\leq t\leq u\leq T$, define the subset $\mathcal{T}_{t,u}:=\left\{\tau\in\mathcal{T}\,\middle\vert\,t\leq \tau\leq u\right\}$ of stopping times taking values in $[t,u]$.
\end{definition}
In addition to the dynamic trading strategy $\theta\in\Theta$, the investor chooses an \textit{exercise time} $\tau\in\mathcal{T}$, in order to maximise his expected forward performance of his hedging portfolio $X_t-C_t=\theta_tS_t-C_t$. Therefore, let $\Theta_{t,\tau}$ denote the subset of strategies starting at $t$ and terminating at $\tau$. The claim pay-off becomes $C(Y_\tau):=C(\tau,Y_\tau)=C_\tau$ with the exercise time $\tau$ as the terminal date instead of the fixed date $T$.
\begin{definition}[Optimal control and stopping problem]
The value process of the investor's portfolio is the combined stochastic control and optimal stopping problem
\begin{equation}
\label{eq:value process C American}
v^C(t,X_t,S_t,Y_t):=\esssup_{\tau\in\mathcal{T}_{t,T}}\esssup_{\theta\in\Theta_{t,\tau}}{\mathbb{E}\left[U_\tau(X_\tau-C(Y_\tau))\,\middle\vert\,\widehat{\mathcal{F}}_t\right]},\quad 0\leq t\leq T.
\end{equation}
The double essential supremum notation will be shortened to $\esssup_{\tau\in\mathcal{T}_{t,T},\theta\in\Theta_{t,\tau}}$.
\end{definition}
In comparison to the European case (\ref{eq:value process C}), the optimisation is additionally performed under the stopping time. The forward indifference price is defined as in Definition \ref{def:value process, indifference price and optimal hedging strategy} and is useful to characterise the \textit{optimal exercise time} $\tau^*$.
\begin{corollary}[Optimal stopping time]
\label{thm:optimal stopping time}
By (\ref{eq:indifference price}) and (\ref{eq:value process C American}), the optimal stopping time $\tau^*$ is the first time the value process reaches the forward performance process,~i.~e.
\begin{align*}
\tau_t^*&=\inf\left\{u\in[t,T]\,\middle\vert\,v^C(u,X_u,S_u,Y_u)=U_u(X_u-C(Y_u))\right\}\\
&=\inf\left\{u\in[t,T]\,\middle\vert\,v^0(u,X_u-p(u,S_u,Y_u),S_u,Y_u)=U_u(X_u-C(Y_u))\right\}\\
&=\inf\left\{u\in[t,T]\,\middle\vert\,U_u(X_u-p(u,S_u,Y_u))=U_u(X_u-C(Y_u))\right\}\\
&=\inf\left\{u\in[t,T]\,\middle\vert\, p(u,S_u,Y_u)=C(Y_u)\right\},
\end{align*}
under appropriate integrability conditions (see \cite[Theorem~D.12]{KS1998}).
\end{corollary}
Corollary \ref{thm:optimal stopping time} implies, that the investor exercises the American option as soon as the forward indifference price reaches from above the option pay-off and allows analysing the optimal exercise time through the forward indifference price.
\begin{corollary}[Primal forward performance problem with American claim]
\label{thm:primal forward performance problem with American claim}
Under the exponential forward performance (\ref{eq:exponential forward performance}), the primal problem (\ref{eq:value process C American}) becomes
\begin{align*}
v^C(t,X_t,S_t,Y_t)&=\esssup_{\tau\in\mathcal{T}_{t,T},\theta\in\Theta_{t,\tau}}\mathbb{E}\left[-\exp\left(-\gamma\left(X_\tau - C(Y_\tau)\right)+\frac12\int_t^\tau\left(\widehat{\lambda}_u^S\right)^2\d u\right)\,\middle\vert\,\widehat{\mathcal{F}}_t\right],\\
&=\underbrace{e^{-\gamma X_t+\frac12\int_0^t\left(\widehat{\lambda}_u^S\right)^2\d u}}_{=U_t(X_t)}\!\!\esssup_{\tau\in\mathcal{T}_{t,T},\theta\in\Theta_{t,\tau}}\!\!\!\mathbb{E}\left[-e^{-\gamma\left(\int_t^\tau\theta_u\d S_u - C(Y_\tau)\right)+\frac12\int_t^\tau\left(\widehat{\lambda}_u^S\right)^2\d u}\,\middle\vert\,\widehat{\mathcal{F}}_t\right]\!.
\end{align*}
\end{corollary}
To obtain the dual optimal control and stopping problem, some preparation is required. Firstly, a reconsideration and extension of the conditional relative entropy from Definition \ref{def:conditional relative entropy} is needed, to include the case of stopping times. Secondly, a relation between the conditional relative entropies up to time $\tau$ and $T$ is derived. Lastly, a particular dynamic programming property of the classical Merton problem is recalled and applied to the American option case.
\begin{definition}[Stopped conditional relative entropy]
\label{def:stopped conditional relative entropy}
Define by
\begin{equation}
\label{eq:right stopped conditional relative entropy}
\mathcal{H}_{t,\tau}(\mathbb{Q},\mathbb{P}):=\mathbb{E}^\mathbb{Q}\left[\log Z_{t,\tau}^\mathbb{Q}\,\middle\vert\,\widehat{\mathcal{F}}_t\right],\quad 0\leq t\leq \tau\in\mathcal{T},
\end{equation}
the \textit{right stopped (conditional) relative entropy} over the stochastic interval $[t,\tau]$~and~by
\begin{equation}
\label{eq:left stopped conditional relative entropy}
\mathcal{H}_{\tau,T}(\mathbb{Q},\mathbb{P}):=\mathbb{E}^\mathbb{Q}\left[\log Z_{\tau,T}^\mathbb{Q}\,\middle\vert\,\widehat{\mathcal{F}}_\tau\right],\quad \mathcal{T}\ni\tau\leq t\leq T,
\end{equation}
the \textit{left stopped (conditional) relative entropy} over the stochastic interval $[\tau,T]$.
\end{definition}
By Proposition \ref{thm:representation of conditional relative entropy}, the right stopped relative entropy (\ref{eq:right stopped conditional relative entropy}) is given by
\begin{equation*}
\mathcal{H}_{t,\tau}(\mathbb{Q},\mathbb{P})=\frac12\mathbb{E}^\mathbb{Q}\left[\int_t^\tau\left[\left(\widehat{\lambda}_u^S\right)^2 +\psi_u^2\right]\d u\,\middle\vert\,\widehat{\mathcal{F}}_t\right].
\end{equation*}
The only difference to the European case is, that $T$ is replaced by $\tau$. From here, the new notation $\mathcal{H}_{t,T}(\mathbb{Q},\mathbb{P})$ is used for $\mathcal{H}_t(\mathbb{Q},\mathbb{P})$. Remark, that the left stopped relative entropy (\ref{eq:left stopped conditional relative entropy}) is $\widehat{\mathcal{F}}_\tau$-conditional. According to Definition \ref{def:stopped conditional relative entropy}, the conditional relative entropy over $[t,T]$ splits into $\mathcal{H}_{t,T}(\mathbb{Q},\mathbb{P})=\mathcal{H}_{t,\tau}(\mathbb{Q},\mathbb{P})+\mathbb{E}^\mathbb{Q}[\mathcal{H}_{\tau,T}(\mathbb{Q},\mathbb{P})\,\vert\,\widehat{\mathcal{F}}_t]$.
\begin{lemma}[Decomposition of the relative entropy under stopping times]
The conditional relative entropy $\mathcal{H}_{t,T}(\mathbb{Q},\mathbb{P})$ decomposes into the right and left entropies
\begin{equation*}
\essinf_{\psi\in\Psi}\mathcal{H}_{t,T}(\mathbb{Q},\mathbb{P})=\essinf_{\psi\in\Psi}\left(\mathcal{H}_{t,\tau}(\mathbb{Q},\mathbb{P})+\mathbb{E}^\mathbb{Q}\left[\essinf_{\psi\in\Psi}\mathcal{H}_{\tau,T}(\mathbb{Q},\mathbb{P})\,\middle\vert\,\widehat{\mathcal{F}}_t\right]\right),
\end{equation*}
under the stopping time $\tau$.
\end{lemma}
\begin{proof}
A proof is given by Leung and Sircar \cite[Lemma 2.7]{LS2009}.
\end{proof}
\begin{proposition}[Primal and dual classical Merton problem with stopping time]
For an investor with starting wealth $X_\tau$ at $\tau\in\mathcal{T}$, the classical Merton value process \begin{equation*}
v^0(\tau,X_\tau)=\esssup_{\theta\in\Theta_{\tau,T}}\mathbb{E}\left[U_0(X_T)\,\middle\vert\,\widehat{\mathcal{F}}_t\right]
\end{equation*}
has the dual separable representation
\begin{equation*}
v^0(\tau,X_\tau,S_\tau)=U_0(X_\tau)\exp\left(-\essinf_{\psi\in\Psi}\mathcal{H}_{\tau,T}(\mathbb{Q},\mathbb{P})\right).
\end{equation*}
With starting wealth $X_t$ at $t\in[0,T]$, the classical value process can be written as
\begin{equation}
\label{eq:horizon-unbiased condition}
v^0(t,X_t,S_t)=\esssup_{\theta\in\Theta_{t,\tau}}\mathbb{E}\left[v^0(\tau,X_\tau,S_\tau)\,\middle\vert\,\widehat{\mathcal{F}}_t\right],\quad \tau\in\mathcal{T}.
\end{equation}
\end{proposition}
\begin{proof}
We refer to \cite[Propositions 2.5 and 2.6]{LS2009}.
\end{proof}
The dynamic programming property (\ref{eq:horizon-unbiased condition}) is called the \textit{self-generating condition} by Musiela and Zariphopoulou \cite{MZ2007}, and \textit{horizon-unbiased condition} by Henderson and Hobson \cite{HH2007}.
\begin{proposition}[Dual classical problem with American option]
\label{thm:dual classical problem with American option}
The dual classical value process in the American option case is given by
\begin{align*}
v^C&(t,X_t,S_t,Y_t)\\
&=U_0(X_t)\exp\left(-\esssup_{\tau\in\mathcal{T}_{t,T}}\essinf_{\psi\in\Psi}\left(\mathcal{H}_{t,\tau}(\mathbb{Q},\mathbb{P})+\mathbb{E}^{\mathbb{Q}}\left[\mathcal{H}_{\tau,T}(\mathbb{Q}^E,\mathbb{P})-\gamma C(Y_\tau)\,\middle\vert\,\widehat{\mathcal{F}}_t\right]\right)\!\!\right).
\end{align*}
The classical exponential indifference price is given by
\begin{equation*}
p(t,S_t,Y_t)=-\frac{1}{\gamma}\esssup_{\tau\in\mathcal{T}_{t,T}}\essinf_{\psi\in\Psi}\left(\mathcal{H}_{t,\tau}(\mathbb{Q},\mathbb{Q}^E)-\gamma\mathbb{E}^\mathbb{Q}\left[C(Y_\tau)\,\middle\vert\,\widehat{\mathcal{F}}_t\right]\right).
\end{equation*}
\end{proposition}
\begin{proof}
A detailed proof is given in \cite[Propositions 2.4 and 2.8]{LS2009}. Therein, the claim is additionally dependent on the stock $S$, i.~e. $C_\tau=C(\tau,S_\tau,Y_\tau)$.
\end{proof}
\begin{theorem}[Forward indifference price valuation with American option]
The dual forward performance problem with the American option has the representation
\begin{align}
\label{eq:forward value process American}
\notag
v^C&(t,X_t,S_t,Y_t)\\
&=U_t(X_t)\exp\left(-\esssup_{\tau\in\mathcal{T}_{t,T}}\essinf_{\psi\in\Psi}\left(\mathcal{H}_{t,\tau}(\mathbb{Q},\mathbb{Q}^M)-\gamma\mathbb{E}^{\mathbb{Q}}\left[C(Y_\tau)\,\middle\vert\,\widehat{\mathcal{F}}_t\right]\right)\right)
\end{align}
with the entropic representation of the forward indifference price
\begin{equation}
\label{eq:forward indifference price control problem American}
p(t,S_t,Y_t)=-\frac{1}{\gamma}\esssup_{\tau\in\mathcal{T}_{t,T}}\essinf_{\psi\in\Psi}\left(\mathcal{H}_{t,\tau}(\mathbb{Q},\mathbb{Q}^M)-\gamma\mathbb{E}^\mathbb{Q}\left[C(Y_\tau)\,\middle\vert\,\widehat{\mathcal{F}}_t\right]\right).
\end{equation}
\end{theorem}
\begin{proof}
Follow the approach from \cite[Proposition~2.7]{LSZ2012} by transforming the~pay-off~into
\begin{equation}
\label{eq:transformed claim pay-off}
\widetilde{C}(\tau,S_\tau,Y_\tau)=C(\tau,Y_\tau)+\frac{1}{2\gamma}\int_t^\tau\left(\widehat{\lambda}^S(u,S_u,Y_u)\right)^2\d u+\frac{1}{\gamma}\mathcal{H}_{\tau,T}(\mathbb{Q}^E,\mathbb{P}).
\end{equation}
Then recall the primal forward performance problem from Corollary~\ref{thm:primal forward performance problem with American claim},
\begin{align*}
v^C(t,X_t,S_t,Y_t)&=U(X_t)\esssup_{\tau\in\mathcal{T}_{t,T},\theta\in\Theta_{t,\tau}}\mathbb{E}\left[-e^{-\gamma\left(\int_t^\tau\theta_u\d S_u - C(Y_\tau)\right)+\frac12\int_t^\tau\left(\widehat{\lambda}_u^S\right)^2\d u}\,\middle\vert\,\widehat{\mathcal{F}}_t\right]\\
&=U(X_t)\esssup_{\tau\in\mathcal{T}_{t,T},\theta\in\Theta_{t,\tau}}\mathbb{E}\left[-e^{-\gamma\left(\int_t^\tau\theta_u\d S_u - \widetilde{C}(\tau,S_\tau,Y_\tau)\right)-\mathcal{H}_{\tau,T}(\mathbb{Q}^E,\mathbb{P})}\,\middle\vert\,\widehat{\mathcal{F}}_t\right],
\end{align*}
through a substitution of the claim pay-off $C$ by the transform $\widetilde{C}$. Now, the application of Proposition~\ref{thm:dual classical problem with American option} yields
\begin{align*}
v^C&(t,X_t,S_t,Y_t)\\
&=U_t(X_t)\exp\left(-\esssup_{\tau\in\mathcal{T}_{t,T}}\essinf_{\psi\in\Psi}\left(\mathcal{H}_{t,\tau}(\mathbb{Q},\mathbb{P})+\mathbb{E}^{\mathbb{Q}}\left[\mathcal{H}_{\tau,T}(\mathbb{Q}^E,\mathbb{P})-\gamma \widetilde{C}_\tau\,\middle\vert\,\widehat{\mathcal{F}}_t\right]\right)\right)\\
&=U_t(X_t)\exp\left(-\esssup_{\tau\in\mathcal{T}_{t,T}}\essinf_{\psi\in\Psi}\left(\mathcal{H}_{t,\tau}(\mathbb{Q},\mathbb{P})-\mathbb{E}^{\mathbb{Q}}\left[\frac{1}{2}\int_t^\tau\!\!\left(\widehat{\lambda}_u^S\right)^2\d u+\gamma C_\tau\,\middle\vert\,\widehat{\mathcal{F}}_t\right]\right)\right)\\
&=U_t(X_t)\exp\left(-\esssup_{\tau\in\mathcal{T}_{t,T}}\essinf_{\psi\in\Psi}\left(\mathcal{H}_{t,\tau}(\mathbb{Q},\mathbb{P})-\mathcal{H}_{t,\tau}(\mathbb{Q^M},\mathbb{P})-\gamma\mathbb{E}^{\mathbb{Q}}\left[C(Y_\tau)\,\middle\vert\,\widehat{\mathcal{F}}_t\right]\right)\right)\\
&=U_t(X_t)\exp\left(-\esssup_{\tau\in\mathcal{T}_{t,T}}\essinf_{\psi\in\Psi}\left(\mathcal{H}_{t,\tau}(\mathbb{Q},\mathbb{Q}^M)-\gamma\mathbb{E}^{\mathbb{Q}}\left[C(Y_\tau)\,\middle\vert\,\widehat{\mathcal{F}}_t\right]\right)\right),
\end{align*}
which proves (\ref{eq:forward value process American}). The forward indifference price representation (\ref{eq:forward indifference price control problem American}) is then implied by (\ref{eq:indifference price}).
\end{proof}
\section{Conclusions and future research directions}
\label{sec:conclusions}
In this thesis we applied the forward performance framework, defined by Musiela and Zariphopoulou~\cite{MZ2009} to the basis risk model with partial information from Monoyios \cite{M2010} to solve the forward utility maximisation problem of exponential type of an investor with a hedging portfolio consisting of a long position in the traded stock $S$ and a short position of a claim written on the non-traded asset $Y$. We obtained the optimal hedging strategy, value function and indifference price representation using methods from duality theory. In the case of an European option, we discussed the main result containing the change of the MEMM $\mathbb{Q}^E$ to the MMM $\mathbb{Q}^M$ after Theorem~\ref{thm:forward indifference price valuation}. We derived the residual risk, pay-off decompositions and an asymptotic expansion of the indifference price. Then we changed to the market model with an American option having a random exercise time inspired by Leung, Sircar and Zariphopoulou~\cite{LSZ2012}. We formulated the optimal control  problem with stopping time and obtained the representations for the value function and forward indifference price. Hereinafter, we take up some points of the thesis to discuss future research topics. We carry out a comprehensive review of the parameter uncertainty in the Kalman-Bucy filter used in our partial information model and discuss alternatives from recent publications. Furthermore, we present the semi-martingale framework for utlity maximisation problems allowing to use weaker assumptions on the model. Moreover, we outline the approach of solving the forward indifference price PDE in the European option's case with numerical methods, by applying the asymptotic expansion of the indifference price as an approximation. For the case with an American option, we describe the variational inequality for the forward indifference price to be expected and a suggestion for solving it numerically. In addition, we propose a larger market model by making more claims available for the market agent, and discuss other large markets in utility maximisation theory. Lastly, we present a generalisation of stochastic utilities used in forward utility-based optimisation theory.
\subsection{Parameter uncertainty in the Kalman-Bucy filter}
\label{subsec:parameter uncertainty}
In Subsection~\ref{subsec:partial information scenario} and Subsection~\ref{subsec:Kalman-Bucy filtering} we developed our partial information model through a Kalman-Bucy filter with known Gaussian prior distribution based on Monoyios~\cite{M2007}, \cite{M2010}.
We assumed in Definition~\ref{def:observation and signal process} the signal process $\Lambda=\begin{pmatrix}\lambda^S\\\lambda^Y\end{pmatrix}$ with unknown MPR constants $\lambda^S,\lambda^Y$ of the asset prices $S,Y$ to have a Gaussian prior distribution \begin{equation*}\Lambda\,\vert\,\widehat{\mathcal{F}}_0 \sim\mathcal{N}(\Lambda_0,\Sigma_0),\; \Lambda_0:=\begin{pmatrix}

              \lambda_0^S\\
              \lambda_0^Y
           \end{pmatrix},\; \Sigma_0:=\begin{pmatrix}
                                         z_0^S & c_0\\
                                         c_0   & z_0^Y
                                      \end{pmatrix},\; c_0:=\rho\min\{z_0^S,z_0^Y\},
\end{equation*}
for given constants $\lambda_0^S,\lambda_0^Y, z_0^S,z_0^Y$. This is the underlying distribution of the Kalman-Bucy Filter introduced in Definition~\ref{def:Kalman-Bucy filter}. The first assumption is made by choosing Gaussian random variables, the second is the knowledge of the prior distribution parameters. If the second assumption is omitted, the problem of uncertain MPRs is shifted to the problem of unknown parameters of the Gaussian prior distribution. One could specify intervals for the parameters, e. g. for $\lambda^S$, using the best estimate approach from Subsection~\ref{subsec:partial information scenario}. The single standard deviation interval $[\overline{\lambda^S}(t)-\frac{1}{\sqrt{t}}, \overline{\lambda^S}(t)+\frac{1}{\sqrt{t}}]$ with confidence $66.27\%$ leads to approximately $t\approx 10$ years of empirical data. With 90\% confidence $t\approx 271$ years of market data is required. A higher confidence of 95\% needs historical data collected since the Early Middle Ages. However, this example shows the statistical error by adopting the estimated interval, which affects the accuracy of the filter. Decision making as utility optimisation based on the filter gets an additional inherent risk.\par
Monoyios~\cite[Section~6]{M2010} carried out extensive numerical simulations with empirical examples of hedging under the partial information model with classical utility. He demonstrated, that the filtering procedure can improve the performance of the hedge, provided that the prior is not extremely poor. The rate of learning by the filter on the asset price MPRs is too slow to counteract parameter uncertainty without the extra insurance of an increased option premium. Monoyios concluded, that considering the combined valuation and hedging program, taking parameter uncertainty into account via an increased option premium and using a filtering approach is of benefit.\par
\textit{Robustness} with respect to model uncertainty in stochastic filtering has been considered for diverse linear and non-linear systems. Miller and Pankov~\cite{MP2005} and Siemenikhin~\cite{Si2016}, for instance, studied linear dynamics with parameter uncertainty in the noise covariance matrices using a so-called \textit{minimax filter}, which is basically an estimator minimising the maximal expected loss over a range of possible models. This idea emerged in Wald~\cite{W1945} in 1945, in which the problem is to find a distribution minimising the maximum risk, which is a general \textit{statistical inference problem}. Therein, the risk is defined as an integral function of the unknown parameters and weighted \textit{statistical decision functions}. Martin and Mintz~\cite{MM1983} examined the existence and behaviour of game-theoretic solutions for robust linear filters and predictors in the context of time-discrete models. They discovered that robust Kalman-Bucy filters can be realised when the least favourable prior distribution is either independent, of, or only weakly dependent upon the specific decision interval. Moreover, they concluded, based on practical experience with times series data, that uncertain dynamics (drifts) can have far greater effect on filter and predictor performance than typical uncertainties in either the signal or observation noise covariances. Verd\'{u} and Poor~\cite{VP1984} noted that minimax estimators are criticised as being too pessimistic and having a poor performance in the most statistically probable model, since they are dependent on the specification of an often arbitrary uncertainty class, likely taking implausible models into consideration.\par
Allan and Cohen~\cite{AC2018} have recently discussed some other filter techniques and proposed a new approach to parameter uncertainty in stochastic filtering, specifically when working with the time-continuous Kalman-Bucy filter by making evaluations via a non-linear expectation, represented in terms of a \textit{penalty function}. The \textit{penalty} is a measure for the error evolving in time caused by the uncertainty, and is calculated by propagating the a priori uncertainty forward through time using filter dynamics. An idea, that has been taken from Cohen~\cite{C2017} and \cite{C2018}, in which the investigation concerned time-discrete models in a binomial and Markov chain framework, respectively.\par
We proposed in Remark~\ref{rem:OU process} an Ornstein-Uhlenbeck model for the signal process (MPRs) and mentioned the parameter uncertainty issue. This model is more complicated than the constant signal $\Lambda$ filtered with the determined Gaussian prior in the sense, that it has multiple parameter uncertainties. An alternative model is the linear equation
\begin{equation*}
\d\lambda_t^i=\alpha_t\lambda_t^i\d t+\beta_t\d W_t^i,\quad i=S,Y,
\end{equation*}
with Gaussian prior $\Lambda\,\vert\,\widehat{\mathcal{F}}_0\sim\mathcal{N}(\Lambda_0,\Sigma_0)$ as defined in (\ref{eq:prior}), and measurable, locally bounded, deterministic functions $\alpha$ and $\beta$ of time on appropriate real intervals as defined in \cite{AC2018}. The parameter functions $\alpha$ and $\beta$ are assumed to be uncertain. Through following the methods from Allan and Cohen~\cite{AC2018}, its feasible to tackle this issue by formulating the penalty problem and measuring penalties dependent of different true and estimated parameters. Robust upper and lower expectations of the signal can provide error bounds for the Kalman-Bucy filter. If the signal $\Lambda$ follows the Ornstein-Uhlenbeck process (\ref{eq:Ornstein-Uhlenbeck}), then one could try to apply the theory of \cite{AC2018} with the aim to analyse penalties and calculate robust bounds for the Kalman-Bucy filter.\par
In addition, one may consider a broader class of prior distributions for the Kalman-Bucy filter. For instance, Bene\v{s} and Karatzas \cite{BK1983} analysed filtering with non-Gaussian prior distribution and showed that the conditional distribution is a mixture of Gaussians, which is propagated by two sets of \textit{sufficient statistics}. These statistics obey usually non-linear SDEs implementable of a filter. For a Gaussian initial distribution, there is only one random sufficient statistic propagating the conditional density, in accordance with the classical theory.\par
Mostovyi and S\^{i}rbu~\cite{MS2017} have recently studied the sensitivity of an expected utility maximisation problem in a continuous semi-martingale market with respect to small changes in the MPR. They analyse the stochastic control problem under the perturbation and give an explicit form of the correction terms for an example with power utility. Eventually, this discussion brings up the question, whether and how parameter uncertainty of the Kalman-Bucy filter affects the forward utilities, (optimal) hedging strategies, residual risk, forward indifference price and claim representations treated in this dissertation, which is a good topic for future research.
\subsection{Utility maximisation in semi-martingale financial models}
For our dual performance maximisation problems, we expressed in Definition~\ref{def:Radon-Nikodym derivative process} the equivalent local martingale measures (ELMMs) $\mathbb{Q}\in\mathcal{M}_{e,f}$ by the Radon-Nikodym derivative processes $Z^\mathbb{Q}$ under $\widehat{\mathbb{F}}$. Since Karatzas and Kardaras~\cite{KK2007}, it has been acknowledged that one does not need ELMMs. Karatzas and Kardaras studied optimal utility-based hedging strategies in a general semi-martingale model  with a weaker assumption, the \textit{``No Unbounded Profit with Bounded Risk'' (NUPBR)} instead of the stronger \textit{``No Free Lunch with Vanishing Risk'' (NFLVR)} condition. They proved, that the optimal portfolio even exists, when the NFLVR assumption is replaced by NUPBR and filled the gap between the \textit{``No Arbitrage'' (NA)} and NFLVR conditions. The NUPBR rule involves the boundedness in probability of the terminal values of wealth processes and is the minimal a priori assumption required in order to proceed with utility optimisation (cf. \cite[p.~449]{KK2007}). Using semi-martingale models with NUPBR is a topic of current research, for example, treated by Mostovyi and S\^{i}rbu~\cite{MS2017}, \cite{MS2018} and Mostovyi~\cite{Mo2018}.
\subsection{Numerical simulations}
In Theorem~\ref{thm:forward indifference price valuation}, we gave the forward indifference price for the European option in terms of a control problem (\ref{eq:forward indifference price control problem}), solving the semi-linear PDE (\ref{eq:p PDE}), which is
\begin{equation*}
p_t+\mathcal{A}_{S,Y}^{\mathbb{Q}^M}p+\frac12\gamma(1-\rho^2)\left(\sigma^Yyp_y\right)^2=0,\quad p(T,s,y)=C(y).
\end{equation*}
In our partial information model from Subsection~\ref{subsec:Kalman-Bucy filtering}, it was not possible to derive a closed probabilistic representation with the distortion method similar to the full information scenario (\ref{eq:indifference price probabilistic representation}) in Remark~\ref{rem:distortion}. The same issue was present under classical utility in Monoyios~\cite{M2010}.\par
A potential further action is the derivation of a numerical solution to the aforementioned PDE for the forward indifference price and investigate valuation and hedging performances inspired by the classical case from \cite[Section~6]{M2010}. Comparable to \cite[Section~5]{M2010}, one can try to obtain an analytic formula for the conditional variance of the claim $\operatorname{Var}^{\mathbb{Q}^M}\left[C(Y_T)\,\middle\vert\,\widehat{\mathcal{F}}_t\right]$ in the asymptotic expansion of the indifference price (\ref{eq:asymptotic expansion of the forward indifference price}) from Proposition~\ref{thm:asymptotic expansion of the forward indifference price}, and to specify the distribution parameters of $\log Y_T$ in terms of the partial information model parameters from Proposition~\ref{thm:model under partial information}. The next step is the attempt to give the Black-Scholes representations of the marginal forward indifference price $p^M$ and marginal hedging strategy $\theta^M$ , and approximate the indifference price by the asymptotic expansion to obtain the derivatives of $p^M$ and of the claim variance $\operatorname{Var}^{\mathbb{Q}^M}\left[C(Y_T)\,\middle\vert\,\widehat{\mathcal{F}}_t\right]$. Using the explicit formulas for the marginal indifference price and hedging strategy, the final step is to give an integral representation of the expected covariation of the profit and loss $\mathbb{E}^{\mathbb{Q}^M}\left[\langle X^M\rangle_{t,T}\,\middle\vert\,\widehat{\mathcal{F}}_t\right]$ in (\ref{eq:asymptotic expansion of the forward indifference price}). If this approach succeeds, then one can try to numerically evaluate this expression using a \textit{Monte-Carlo simulation}. One expects to find that the forward utility approach is like a low risk aversion limit of the classical approach, since the minimal martingale measure $\mathbb{Q}^M$ is used.\par
In the case of an American option from Section~\ref{sec:american option}, a further approach is to derive the \textit{variational inequality} for the forward indifference price $p(t,s,y)$ under partial information similar to the one given by Leung et al.~\cite[Subsection~3.1]{LSZ2012} under full information. One expects, that the indifference price solves the free boundary problem
\begin{equation}
\label{eq:pde american}
\begin{aligned}
\begin{cases}
&\!\!\!\!p_t+\mathcal{A}_{S,Y}^{\mathbb{Q}^M}p+\frac12\gamma(1-\rho^2)\left(\sigma^Yyp_y\right)^2\leq 0,\\
&\!\!\!\!p(t,s,y)\geq C(t,y),\\
&\!\!\!\!\left(p_t+\mathcal{A}_{S,Y}^{\mathbb{Q}^M}p+\frac12\gamma(1-\rho^2)\left(\sigma^Yyp_y\right)^2\right)(C(t,y)-p(t,s,y))=0,\\
&\!\!\!\!p(T,s,y)=C(T,y),
\end{cases}
\end{aligned}
\end{equation}
for $(t,s,y)\in[0,\infty)\times\mathbb{R}\times[0,T]$. The crucial difference to \cite{LSZ2012} is that the indifference price under our partial information model depends additionally on $S$ and $Y$, rather than only on $Y$. Nevertheless, \cite[equation~(32)]{LSZ2012} displays the variational inequality in the general case, when $p$ as well as $C$ depend on $S$ and $Y$. To derive a numerical solution for the indifference price, one needs to solve the free boundary problem (\ref{eq:pde american}) in three dimensions, which is a non-trivial task. In \cite[Section~4]{LSZ2012} early exercise problems of \textit{employee stock options (ESOs)} are modelled under the full information scenario with constant MPRs $\lambda^S, \lambda^Y$, and solved numerically using a \textit{fully explicit finite-difference scheme} for the exponential forward performance case. Full and partial information models of ESOs are analysed, for instance, by Henderson et al.~\cite{HKM2018} and Monoyios and NG~\cite{MN2011}.
\subsection{Utility maximisation in larger markets}
In our basis risk market model, we considered a single European (American) option $C$ with fixed expiry $T$ (early exercise time $\tau$) on the non-traded asset $Y$. This market model can be enlarged by offering $n$ ($n>1$) European (American) claims $C_1,\dots,C_n$ written on $Y$ with expiries $T_1,\dots,T_n$ (early exercise times $\tau_1,\dots,\tau_n$) and pay-offs $C_1(Y_{T_1}),\dots,C_n(Y_{T_n})$ ($C_1(Y_{\tau_1}),\dots,C_n(Y_{\tau_n})$). Furthermore, a setting of mixed European and American claims may be considered. When creating the hedging portfolio, the market agent must decide, how many and which claims he is going to include. One may simply value each single option $C_j,\;j=1,\dots,n$ separately by setting up $n$ different hedging portfolios with a long position in the stock $S$ and a short position with a unit of the option $C_j$, but it is not clear, if one misses on possible effects in the risk management strategies.  Specifically, with forward utility and different, flexible exercise times, the number of options held in the portfolio at the same time can vary.\par
Generally, one may enlarge the financial market by adding more assets. The concept of a large security market was described by Kabanov and Kramkov~\cite{KK1998} as a sequence of probability spaces (general models), whereas Bj\"ork and N\"aslund~\cite{BN1998} defined a large market to be one probability space with countably many assets $((S_t^i)_{0\leq t\leq T})_{i=1}^\infty$. Donno, Guasoni and Pratelli~\cite{DGP2005} applied the classical utility maximisation theory on a large market, studied them with duality methods and characterised replicable claims. Mostovyi~\cite{Mo2018} considered the model from \cite{DGP2005} with stochastic utility. He concluded, that the value function with countably many assets is the limit of the value functions of the finite-dimensional models, but the optimal strategy with infinite assets is not a limit of the trading strategies of the finite-dimensional markets, in general.
\subsection{General utility random fields}
We used the forward utility $U_t(x)=-e^{-\gamma x+\frac12 \int_0^t\left(\widehat{\lambda}_u^S\right)^2\d u}$ of exponential type, defined in Subsection~\ref{subsec:Forward performance problem in an incomplete market}, for the utility optimisation problems. It is a specific forward utility derived by the class of asymptotically linear local risk tolerance functions dealt in Subsection~\ref{subsec:forward utility and local risk tolerance}. Musiela and Zariphopoulou~\cite{MZ2007}, \cite{MZ2010} suggested this model to give more flexibility to the individual risk preferences of an investor adapting the market development. El Karoui and M'Rad~\cite{EM2013a}, \cite{EM2013b} studied the consistency of dynamic utilities. They introduced the general notion of \textit{progressive utility}, which is a collection of It\^{o} semi-martingales with dynamics
\begin{equation}
\label{eq:dU}
\d U(t,X_t)= \beta(t,X_t)\d t+\gamma(t,X_t)\d W_t,
\end{equation}
including drift and volatility processes $\beta, \gamma$. The stochastic utilities are often referred to as \textit{utility random fields}. Utility random fields of investment and consumption were considered at first by Berrier and Tehranchi~\cite{BT2009} and Berrier et al.~\cite{BRT2009}. El Kaouri et al.~\cite{EHM2018} extended the forward utility setting through market-consistent utility random fields that are calibrated to a given learning $\sigma$-algebra. They provided differential regularity conditions on stochastic utility properties ensuring the existence of consistency and optimal strategies. Defining utility random fields through SDEs of the form (\ref{eq:dU}) offers an opportunity for future research in utiliy-based valuation and hedging.
\newpage
\thispagestyle{plain}

\end{document}